\documentclass[a4paper,preprint,reqno]{imsart}
\usepackage[utf8]{inputenc}
\RequirePackage{amsthm,amsmath,amsfonts,amssymb}
\RequirePackage[numbers]{natbib}

\usepackage{hyperref}
\hypersetup{colorlinks=true,citecolor=blue,urlcolor=blue,breaklinks=true}
\usepackage{float}
\usepackage{dsfont}
\usepackage{mathrsfs,booktabs}
\usepackage{nicefrac}
\usepackage{stmaryrd}
\usepackage{multirow}
\usepackage{array}
\usepackage{graphicx}
\usepackage{mathtools}
\usepackage{xcolor}
\usepackage{algpseudocode}
\usepackage[ruled,linesnumbered]{algorithm2e}

\startlocaldefs

\newcolumntype{C}[1]{>{\centering\let\newline\\\arraybackslash\hspace{0pt}}m{#1}}
\newcommand{\R}{\mathbb{R}}
\newcommand{\PP}{\mathbb{P}}

\newcommand{\N}{\mathbb{N}}
\newcommand{\E}{\mathbb{E}}

\newcommand{\dd}{\mathrm{d}}

\newcommand\bfG{\mathbf G}

\newcommand\bfK{\mathbf K}
\newcommand\bfP{\mathbf P}
\newcommand\bfQ{\mathbf Q}
\newcommand\bN{\boldsymbol N}

\newcommand\bfSigma{\mathbf\Sigma}
\newcommand\bfH{\mathbf H}

\newcommand\bfS{\mathbf S}
\newcommand\bfM{\mathbf M}
\newcommand\bfI{\mathbf I}

\newcommand\bfA{\mathbf A}

\newcommand\bZ{\boldsymbol Z}
\newcommand\bW{\boldsymbol W}
\newcommand\bV{\boldsymbol V}

\newcommand\bX{\boldsymbol X}
\newcommand\bY{\boldsymbol Y}
\newcommand\bx{\boldsymbol x}

\newcommand\bv{\boldsymbol v}
\newcommand\ba{\boldsymbol a}
\newcommand\bb{\boldsymbol b}

\newcommand\bG{\boldsymbol G}

\newcommand\by{\boldsymbol y}
\newcommand\bz{\boldsymbol z}
\newcommand\bxi{\boldsymbol\xi}

\newcommand\btheta{\boldsymbol\theta}
\newcommand\bzero{\mathbf{0}}
\newcommand\bfzero{\boldsymbol{0}}

\DeclarePairedDelimiter\floor{\lfloor}{\rfloor}
\theoremstyle{plain}

\newtheorem{theorem}{Theorem}
\newtheorem{proposition}{Proposition}

\newtheorem{assumption}{Assumption}

\newtheorem{lemma}{Lemma}

\newtheorem{?}{Question}

\endlocaldefs

\usepackage[nameinlink,noabbrev]{cleveref}
\crefname{assumption}{assumption}{Assumptions}

\usepackage{parskip}
\usepackage{a4wide}

\begin{document}

\begin{frontmatter}

\title{Metropolis Adjusted Langevin Trajectories: a robust alternative to Hamiltonian Monte Carlo}
\runtitle{\, Metropolis Adjusted Langevin Trajectories}

\begin{aug}

%\author{\fnms{Lionel} \snm{Riou-Durand} \ead[label=e1]{lionel.riou-durand@warwick.ac.uk}}
%\and
%\author{\fnms{Jure} \snm{Vogrinc} \ead[label=e2]{jure.vogrinc@warwick.ac.uk}}

\author{\fnms{Lionel} \snm{Riou-Durand\hspace{-0.1cm}} \ead[label=e1]{lionel.riou-durand@insa-rouen.fr }}\and
%\author{\fnms{Jure} \snm{Vogrinc} \ead[label=e2]{jure.vogrinc@warwick.ac.uk}}
\author{\fnms{Jure} \snm{Vogrinc} \ead[label=e2]{jure.vogrinc@warwick.ac.uk}}

\runauthor{Riou-Durand \& Vogrinc\, }

%\affiliation{University of Warwick}

%{\renewcommand{\addtocontents}[2]{}
\address{INSA Rouen, Normandie Université, France ; University of Warwick, United Kingdom.\\ \printead{e1,e2}}
%\address{INSA Rouen Normandie, \\Saint-Étienne-du-Rouvray, 76800, France.\\ \printead{e1}}
%\address{University of Warwick, \\Coventry, CV4 7AL, United Kingdom.\\ \printead{e2}}
%\address{University of Warwick, Coventry, CV4 7AL, United Kingdom.\\ \printead{e1,e2}}
%\\ e-mail: %\href{mailto:lionel.riou-durand@warwick.ac.uk}{lionel.riou-durand@warwick.ac.uk};
%\href{mailto:jure.vogrinc@warwick.ac.uk}{jure.vogrinc@warwick.ac.uk}.}
%}

\end{aug}

\begin{abstract}
We introduce MALT: a new Metropolis adjusted sampler built upon the (kinetic) Langevin diffusion. Compared to Generalized Hamiltonian Monte Carlo (GHMC), the Metropolis correction is applied to whole Langevin trajectories, which prevents momentum flips, and allows for larger step-sizes. We argue that MALT yields a neater extension of HMC, preserving many desirable properties. We extend optimal scaling results of HMC to MALT for isotropic targets, and obtain the same scaling with respect to the dimension without additional assumptions. We show that MALT improves both the robustness to tuning and the sampling performance of HMC on anisotropic targets. We compare our approach with Randomized HMC, recently praised for its robustness. We show that, in continuous time, the Langevin diffusion achieves the fastest mixing rate for strongly log-concave targets. We then assess the accuracies of MALT, GHMC, HMC and RHMC when performing numerical integration on anisotropic targets, both on toy models and real data experiments on a Bayesian logistic regression. We show that MALT outperforms GHMC, standard HMC, and is competitive with RHMC.
\end{abstract}



\end{frontmatter}
\maketitle

%Preprint
%\parskip=9pt

\setcounter{tocdepth}{1}
\tableofcontents

\section{Introduction}\label{sec_intro}

 We consider the problem of sampling from a probability distribution with positive density $\Pi$ with respect to Lebesgue's measure on $\R^d$. We define $\Pi$ through a potential function $\Phi:\R^d\rightarrow \R$ satisfying $\int_{\R^d}\exp\{-\Phi(\by)\}\dd \by<\infty$, as follows
\begin{equation*}
    \Pi(\bx)\propto\exp\{-\Phi(\bx)\},\qquad\bx\in\R^d.
\end{equation*}
%The function $\Phi$ is called the potential. 
 
Exact sampling methods (e.g. rejection sampler) can tackle this problem, but their complexities increase exponentially fast with $d$. In this work we focus on approximate sampling methods known as Markov Chain Monte Carlo (MCMC) algorithms, for which polynomial complexities with $d$ can be obtained under suitable assumptions. These methods rely on sampling from a tractable Markov chain which converges in distribution to the target $\Pi$.

 %A widely used method for tackling this problem consists of building a discrete time Markov chain targeting the stationary distribution corresponding to $\Pi$, for which updates can be sampled exactly through a tractable algorithm. These 
 
%In order to obtain polynomial complexities with $d$ recursive sampling methods are commonly referred to as Markov Chain Monte Carlo (MCMC) algorithms. 
%If MCMC algorithms can solve the sampling problem by drawing a long run of the Markov chain, they can also enable the approximation of intractable expectations with respect to $\Pi$. %Indeed, empirical averages built upon random samples drawn from a Markov chain are convergent estimators of such expectations by the ergodic theorem. 
MCMC algorithms are powerful tools to estimate statistical models involving intractable integrals. They are widely used
% MCMC algorithms are powerful tools for providing numerical approximations of statistical models involving intractable integrals. Their use is widely spread 
in Bayesian statistics, to approximate moments of high-dimensional posterior distributions. Efficient sampling algorithms for large $d$ often rely on discretized Markov processes guided by the gradient of the potential $\Phi$. These algorithms are commonly referred to as gradient-based samplers. In this work, we make the following smoothness assumption on $\Phi$. We note $|\bx|\triangleq (\bx^\top \bx)^{1/2}$ the Euclidean norm of $\bx\in \R^d$.

\begin{assumption}\label{assumption:grad_lipschitz}
The potential $\Phi\in C^1(\mathbb{R}^d)$ has a Lipschitz gradient
$$
\exists\, M>0,\qquad|\nabla\Phi(\bx)-\nabla\Phi(\by)|\le M|\bx-\by|,\qquad \bx,\by\in\mathbb{R}^d.
$$
\end{assumption}
The discretization of a Markov process induces a bias in its stationary distribution. 
%The complexity of the resulting algorithm can be measured via the number of iterations needed to sample from a distribution within an $\varepsilon>0$ distance from $\Pi$, e.g. for the total variation or the Wasserstein metric.
Accurate sampling can be ensured by choosing a small enough step-size to reduce the bias while running a long enough chain to get close to stationarity. Solving this trade-off has received a lot of attention lately, and polynomial complexities with $d$ were obtained for log-concave targets satisfying \Cref{assumption:grad_lipschitz}; see \cite{dalalyan2017theoretical, dalalyan2019user, durmus2017nonasymptotic,  durmus2019high, durmus2019analysis, karagulyan2020penalized} for the overdamped Langevin diffusion, \cite{cheng2018underdamped,dalalyan2020sampling,dalalyan2019bounding,ma2021there,monmarche2020high, sanz2021wasserstein} for the (kinetic) Langevin diffusion, and \cite{bou2021mixing,bou2020convergence,bou2020coupling,chen2019optimal,mangoubi2019mixing,mangoubi2021mixing,mangoubi2018dimensionally} for Hamiltonian dynamics. A common practice is to adjust the long-term bias of the chain by facing each update with an accept reject test known as Metropolis correction; see \cite{Hastings::1970,metropolis1953equation}. %This is the case of
Among the samplers using this correction are the Metropolis Adjusted Langevin Algorithm (MALA) and Hamiltonian Monte Carlo (HMC); see \cite{besag1994comments,Duane::1987,RobertsTweedie::1996}. %An extensive comparison between adjusted and unadjusted schemes is beyond the scope of our study. In this work, we focus on these adjusted samplers, which enable the tuning of the step-size by monitoring the acceptance rate of the algorithm.
%One advantage of Metropolis adjusted samplers is that
%One limitation of these samplers is that whenever the discretization error scales polynomially with the time-step, the number of gradient evaluations required to reach a given precision will increase polynomially with the precision level at best. Another method is to adjust the bias of the stationary distribution is to apply the Metropolis-Hastings correction; see \cite{Hastings::1970,metropolis1953equation}. This correction was applied to various discretized Markov processes and lead to several algorithms, among the most prominent algorithms are the Metropolis Adjusted Langevin Algorithm (MALA) and Hamiltonian Monte Carlo (HMC); see \cite{besag1994comments,Duane::1987,RobertsTweedie::1996}. For Metropolis-adjusted samplers, a logarithmic scaling with respect to the precision level is achievable as soon as the Markov chain is geometrically ergodic. In this work, we focus on Metropolis adjusted samplers. An extensive comparison between adjusted and unadjusted samplers is beyond the scope of our study.
Measuring the complexity of Metropolis adjusted samplers was initiated by the study of scaling limits for product form targets; see \cite{Beskos::2013,pillai2012optimal,Roberts::1998,Roberts::1997}. These studies show that accurate integration at stationarity requires a number of gradient queries that scales asymptotically as $d^{1/3}$ for MALA and $d^{1/4}$ for HMC. Non-asymptotic guarantees for Metropolis adjusted samplers have also been obtained; see \cite{chen2020fast,chewi2021optimal,dwivedi2018log,lee2020logsmooth}.
While the complexity of gradient-based samplers is actively studied,
their sensitivity to tuning often remains a major issue for their practical implementations.
%It is now well established that in a typical application gradient-based samplers can outperform the others in terms of computational cost for high dimensional target distributions. However, their lack of robustness to tuning is often a major issue for their practical implementations. 
Developing samplers that combine efficiency and robustness is a key challenge, which receives a growing interest; see \cite{bou2017randomized, hoffman2021adaptive,livingstone2019barker,vogrinc2022optimal}. Our work is motivated by this objective.
 %In this work, we focus on these adjusted samplers, which enable the tuning of the step-size by monitoring the acceptance rate of the algorithm.

%\jure{I think if we can give the assumptions 1,2,3 short descriptive names apart from numbers it is better. If we refer back to them within some proofs it can give the reader an instant reminder better than a number. A3:Product form, A2: strong log-concavity, A1:Lipschitz gradient potential}
 %Hamiltonian Monte Carlo (\cite{Duane::1987}) is a sampler 
 We focus in particular on samplers built upon Hamiltonian dynamics, which describe the motion of a position $\bX_{t}\in\R^d$ and a velocity $\bV_{t}\in\R^{d}$ (a.k.a momentum). These are defined as the solution of the following system of Ordinary Differential Equations (ODE).
%Among the family of gradient-based MCMC algorithms, HMC is often presented as a state-of-the-art sampler for high dimensional targets, justified by its gold standard $d^{1/4}$ scaling with respect to dimension; see \cite{Beskos::2013}. The HMC algorithm and many of its variations are built upon a system of Ordinary Differential Equations (ODE), known as Hamiltonian dynamics. The solution of these deterministic dynamics at time $t\ge 0$ is composed by a position $\bX_{t}\in\R^d$ and a velocity $\bV_{t}\in\R^{d}$ (a.k.a momentum). These correspond to the solution of the ODE 
\begin{equation}\label{eq_hamilton}
    \dd\begin{bmatrix}\bX_{t}\\
\bV_{t}
\end{bmatrix}=\begin{bmatrix}
\bV_{t}\\
-\nabla \Phi(\bX_{t})
\end{bmatrix}\dd t.
\end{equation}
Hamltonian dynamics preserve the density
\begin{equation*}
    \Pi_*(\bx,\bv)\propto\exp\{-\Phi(\bx)-|\bv|^2/2\},\qquad(\bx,\bv)\in\R^{2d}.
\end{equation*}
%We use the notations $\Pi$ and $\Pi_*$ to refer both to the densities and to their corresponding probability distributions, when the context is clear. 
This density is invariant in the sense that $(\bX_{0},\bV_{0})\sim\Pi_*\, \Rightarrow\, (\bX_{t},\bV_{t})\sim\Pi_*$ for any $t\ge0$. Yet Hamiltonian trajectories follow the contours of $\Pi_*$ so they cannot be ergodic without introducing random refreshments; see below.
The density $\Pi_*$ %on $\R^{2d}$
is a product between %two marginal densities, corresponding to independent position and velocity drawn respectively from 
$\Pi$ and a standard Gaussian velocity. Thus sampling from $\Pi_*$ yields marginal samples from $\Pi$ by keeping only the positions. 
In general, \eqref{eq_hamilton} has no closed form solutions, and needs to be discretized. One of the most well known discretizations is the Störmer-Verlet integrator (a.k.a leapfrog), defined as follows. 
 For a step-size $h>0$, let $\btheta_h:(\bx_0,\bv_0)\mapsto(\bx_1,\bv_1)$ such that
    \begin{align*}
    \bv_{1/2}&=\bv_0-(h/2)\nabla \Phi(\bx_0)\\
    \bx_{1}&=\bx_0+h\bv_{1/2}\\
    \bv_1&=\bv_{1/2}-(h/2)\nabla \Phi(\bx_1).
\end{align*}
A numerical trajectory of length $T\ge h$ is composed by $L=\lfloor T/h\rfloor$ Störmer-Verlet updates, noted  $\btheta_h^L\triangleq\btheta_h\circ\cdots \circ\btheta_h$. This map is volume preserving: the Jacobian of $(\bx,\bv)\mapsto\btheta_h^L(\bx,\bv)$ equals one; and time-reversible: the map $\varphi\circ\btheta_h^L$ is an involution for $\varphi(\bx,\bv)\triangleq(\bx,-\bv)$.

%if $\varphi(\bx,\bv)\triangleq(\bx,-\bv)$ then $\varphi\circ\btheta_h^L$ is an involution.
%; see \cite{neal2011mcmc}. Formally, this means respectively that the Jacobian of the map $(\bx,\bv)\mapsto\btheta_h^L(\bx,\bv)$ is equal to one, ans that if $\varphi(\bx,\bv)\triangleq(\bx,-\bv)$ then $\varphi\circ\btheta_h^L$ is an involution.

%preserves several properties of Hamiltonian dynamics; see \cite{neal2011mcmc}. In particular, it is known to be time reversible. Let $\varphi(\bx,\bv)\triangleq(\bx,-\bv)$ stands for the flip of velocity. Time reversibility means that flipping the velocity along the trajectory reduces to going backwards in time, or equivalently, that the map $\varphi\circ\btheta_h^L$ is an involution. 
The HMC algorithm introduced in \cite{Duane::1987}, proposes numerical Hamiltonian trajectories
%The hybrid, or Hamiltonian Monte Carlo algorithm was first introduced in \cite{Duane::1987}. It consists of proposing a Hamiltonian trajectory of length $T>0$, approximated with the leapfrog method for a step-size $h>0$, 
faced with a Metropolis correction. The momentum is refreshed by a Gaussian draw at the start of each trajectory, inducing the randomness needed to ensure ergodicity; see \cite{durmus2020irreducibility}. An extension of this algorithm, called Generalized Hamiltonian Monte Carlo (GHMC), was then suggested in \cite{horowitz1991generalized} to allow for partial momentum refreshments with a persistence parameter $\alpha\in[0,1)$.
%(a.k.a cosines of the angle of refreshment). 
GHMC is defined explicitly in \Cref{alg:ghmc}. It reduces to standard HMC when $\alpha=0$. 

%\begin{algorithm}[H]
%\SetAlgoLined
%\SetKwInOut{Input}{Input}\SetKwInOut{Output}{Output}
%\Input{Starting point $(\bX_0,\bV_0)\in\R^{2d}$, number of MCMC samples $N\ge 1$, integration time $T>0$, step-size $h>0$, and persistence $\alpha\in[0,1)$.
%}
%\BlankLine
%\DontPrintSemicolon
%Set $L\leftarrow\lfloor T/h\rfloor$\\
%\For{$n\leftarrow1$ to $N$}{
%	propose $(\bX',\bV')\leftarrow\btheta_h^L(\bX_{n-1},\bV_{n-1})$\\ draw a uniform random variable $U$ on $(0,1)$\\
%	\If{$U>\Pi_*(\bX',\bV')/\Pi_*(\bX_{n-1},\bV_{n-1})$}{reject and flip the momentum $(\bX',\bV')\leftarrow(\bX_{n-1},-\bV_{n-1})$}
%    draw a standard Gaussian random vector $\bxi$ on $\R^d$\\ refresh the momentum $\bV'\leftarrow \alpha \bV'+\sqrt{1-\alpha^2}\bxi$\\
%    set $(\bX_{n},\bV_{n})\leftarrow(\bX',\bV')$
%}
%\KwRet{$(\bX_1,\bV_1),\cdots,(\bX_N,\bV_N)$}.\;
%\caption{Generalized Hamiltonian Monte Carlo}\label{alg:ghmc_bis}
%\end{algorithm}
%\vspace{-0.2cm}
\begin{algorithm}[!ht]
\SetAlgoLined
\SetKwInOut{Input}{Input}\SetKwInOut{Output}{Output}\normalsize
\Input{Starting point $(\bX^0,\bV^0)\in\R^{2d}$, number of MCMC samples $N\ge 1$, step-size $h>0$, integration time $T\ge h$, and persistence $\alpha\in[0,1)$.
}
\BlankLine
\DontPrintSemicolon
Set $L\gets\lfloor T/h\rfloor$\\
\For{$n\gets1$ to $N$}{
    draw $\bxi\sim\mathcal{N}_d(\bzero_d,\bfI_d)$ and refresh the momentum $\bV'\gets \alpha \bV^{n-1}+\sqrt{1-\alpha^2}\bxi$\\
    propose a Hamiltonian trajectory $(\bX^{n},\bV^{n})\gets\btheta_h^L(\bX^{n-1},\bV')$\\
    compute the energy difference $\Delta\gets\Phi(\bX^{n})-\Phi(\bX^{n-1})+(|\bV^{n}|^2-|\bV'|^2)/2$\\
	draw a uniform random variable $U$ on $(0,1)$\\
	\If{$U>\exp\{-\Delta\}$}{reject and flip the momentum $(\bX^{n},\bV^{n})\gets(\bX^{n-1},-\bV')$}
}
\KwRet{$(\bX^1,\bV^1),\cdots,(\bX^N,\bV^N)$}.\;
\caption{Generalized Hamiltonian Monte Carlo}\label{alg:ghmc}
\end{algorithm}
%\vspace{-0.2cm}
%set $(\bx_0,\bv_0)\leftarrow(\bX_{n-1},\bV_{n-1})$\\
    %\For{$i\leftarrow1$ to $L$}{
%	set $(\bx_i,\bv_i)\leftarrow\btheta_h(\bx_{i-1},\bv_{i-1})$
%	}
%	set $(\bX_{n},\bV_{n})\leftarrow(\bx_L,\bv_L)$\\
%The Metropolis-Hastings kernel is known to be a reversible Markov kernel with respect to its invariant measure.
%Time-reversibility %is a dynamic notion, 
%for an (approximate) Hamiltonian trajectory is different from the notion of reversibility for a Markov chain (a.k.a detailed balance). Yet the two notions can be
%which 
%relates to the notion of
 \Cref{alg:ghmc} defines a skew-reversible Markov kernel; see \cite[Lemma 2]{andrieu2021peskun}. It satisfies detailed balance up to a momentum flip, noted $\varphi(\bx,\bv)=(\bx,-\bv)$. One way to deconstruct \Cref{alg:ghmc} is to split the proposed trajectory as $\btheta_h^L=\varphi\circ\varphi\circ\btheta_h^L$ so that the involution $\varphi\circ\btheta_h^L$ is faced with an accept-reject test. Overall, 
%when composing the output of the test with $\varphi$, 
a momentum flip occurs whenever a move is rejected. This is of little importance when $\alpha=0$, i.e. for HMC, as momentum flips are erased by full refreshments; the Markov kernel is reversible, see \cite[Remark 13]{andrieu2020general}. However, momentum flips are only partially erased when $\alpha\in(0,1)$; the Markov kernel is non-reversible.
%Beyond this technical distinction, 
%It is important to remark 
%We emphasize that
Momentum flips typically slow down the mixing of the sampler,
by causing backtracking which limits the exploration of the space.
%Intuitively, momentum flips cause backtracking which slows down the exploration of the space.
This is supported by Peskun-ordering results for skew-reversible kernels; see \cite[Corrolary 2]{andrieu2021peskun}. It illustrates an important
%We insist here on a crucial
distinction between efficiency of the exploration, and non-reversibility of the Markov kernel.

%In this work we focus on an alternative strategy to introduce randomness in the Hamiltonian trajectories. 
The Langevin diffusion (a.k.a kinetic, geometric, underdamped, or second order Langevin) combines \eqref{eq_hamilton} with a continuous momentum refreshment with damping $\gamma\ge0$ (a.k.a friction), from a Brownian motion $(\bW_t)_{t\ge0}$ on $\R^d$. Under \Cref{assumption:grad_lipschitz}, the Langevin diffusion is the strong solution of the following Stochastic Differential Equation (SDE). 
%The resulting process is referred in the molecular dynamics' literature as Langevin diffusion (a.k.a kinetic, geometric, underdamped, or second order Langevin). 
\begin{equation}\label{eq_langevin}
\dd\begin{bmatrix}\bX_{t}\\
\bV_{t}
\end{bmatrix}=\begin{bmatrix}\bV_{t}
\\
-\nabla \Phi(\bX_{t})
\end{bmatrix}\dd t+\begin{bmatrix}\bzero_d
\\
-\gamma\bV_{t}\,\dd t+\sqrt{2\gamma}\,\dd\bW_{t}
\end{bmatrix}.
\end{equation}
The Langevin diffusion preserves $\Pi_*$ as well. It coindices with Hamiltonian dynamics for $\gamma=0$, but a positive damping is required for the process to be ergodic. Quantitative rates of convergence were established for various damping regimes; see \cite{dalalyan2020sampling,eberle2019couplings}. The overdamped Langevin diffusion is obtained as $\gamma\rightarrow\infty$ while rescaling the time as $\bar{\bX}_t=\bX_{\gamma t}$; see \cite{nelson1967dynamical}. 

One tuning strategy for GHMC, advocated in \cite{horowitz1991generalized}, is to approximate the Langevin diffusion as $h\rightarrow 0$ by choosing $\alpha\sim \exp(-\gamma h)$ and $T=h$ in \Cref{alg:ghmc}. A Metropolis correction is then applied to each single leapfrog step. Refreshments become more and more partial as $\gamma\rightarrow 0$, while momentum flips get erased more and more partially. This emphasizes backtracking upon rejection, which severely impacts the sampling performance of GHMC; see \cite{horowitz1991generalized,kennedy2001cost}. Several variants of GHMC were proposed; see \cite{bou2010pathwise, ottobre2016function,ricci2003algorithms,scemama2006efficient}, but these analyses always consider small enough step-sizes to maintain a negligible amount of rejections. This strong constraint on the step-size motivated several methodological variations of GHMC aiming at mitigating the impact of momentum flips; see \cite{campos2015extra,sohl2014hamiltonian,wagoner2012reducing}. Some of these are related to delayed rejection methods, introduced in \cite{green2001delayed, mira2001metropolis, tierney1999some}; see \cite{park2020markov} for recent developments. Reducing flips/rejections in these samplers always comes at the price of a higher number of proposals, inducing another trade-off to solve in terms of computational cost.\newline

The article is organised as follows.
%, together with a summary of our contributions.

\begin{itemize}
\item In \Cref{sec_malt}, we present the Metropolis Adjusted Langevin Trajectories (MALT) sampler. Compared to previous samplers built upon the kinetic Langevin diffusion, like GHMC and its variations (see \cite{bou2010pathwise, horowitz1991generalized, ottobre2016function, ricci2003algorithms, scemama2006efficient}), MALT applies a Metropolis adjustment on entire Langevin trajectories, which arguably defines a neater extension of HMC. In particular, MALT avoids the momentum flip problem of GHMC, and yields a reversible and ergodic Markov kernel for any positive damping.
\item In \Cref{sec_optimal_scaling}, we derive optimal scaling limits of MALT for product-form targets. Our study extends the result of \cite{Beskos::2013} for HMC to any choice of damping. We show that the $d^{1/4}$ scaling and the $65\%$ acceptance rate are optimal without further assumptions. Our result ensures that on isotropic targets, MALT preserves the rate of convergence of HMC. It also yields a simple guideline for tuning the step-size, which gives MALT another edge over GHMC.
\item In \Cref{sec_robustness}, we show that MALT greatly improves the robustness of HMC on anisotropic targets. We first highlight this property with explicit calculations on a toy Gaussian model. When the target has heterogeneous scales, HMC suffers from a loss of efficiency and a greater sensitivity to the choice of trajectory length, while MALT solves both issues for an explicit damping. We extend our study to strongly log-concave targets, and relate the Langevin diffusion to another robust approach known as Randomized HMC (see \cite{bou2017randomized}). Our analysis refines the mixing rate of \cite{deligiannidis2021randomized} for RHMC, and shows that the Langevin diffusion is a limit of RHMC that achieves the fastest mixing rate. To our knowledge, this connection is new and motivates the use of partial momentum refreshment.
\item In \Cref{sec_comparisons}, we compare the accuracies of MALT, HMC, GHMC and RHMC when doing numerical integration. This study aims at approximating the moments of various target distributions, including toy examples as well as a real data Bayesian model with log-concave posterior distribution. We measure the worst effective sample size across the components for various test functions. Our results show that MALT outperforms both HMC and GHMC, while it offers competitive results compared to RHMC.  
%\item A table of notations is presented in \Cref{sec_notations}.
\end{itemize}

 A table of notations is presented in \Cref{sec_notations}.

\section{Metropolis Adjusted Langevin Trajectories}\label{sec_malt}
We introduce MALT: a new Metropolis adjustment of a standard time discretization of the (kinetic) Langevin diffusion. We explain its foundations and establish connections with previous Metropolis adjusted samplers built upon the Langevin diffusion such as GHMC; see \cite{bou2010pathwise,bussi2007accurate,horowitz1991generalized,ottobre2016function,ricci2003algorithms,scemama2006efficient}. Compared to these approaches, the Metropolis correction is applied to a whole Langevin trajectory of length $T>0$. This mechanism allows us to get rid of GHMC's perturbations induced by momentum flips, by performing full refreshments of the velocity at the start of each trajectory. We describe MALT as an extension of HMC, which uses an arguably neater adjustment than GHMC. We support this claim by presenting several comparative advantages of MALT, both in terms of tuning and sampling performance.

For any time-step $h>0$, we introduce $\eta=e^{-\gamma h/2}$. We also let $\bxi,\bxi'\sim\mathcal{N}_d(\bfzero_d,\bfI_d)$ be independent. Starting from $(\bx_0,\bv_0)\in\R^{2d}$, we consider the following discretization 
\begin{align*}
\bv_0'&=\eta \bv_0+\sqrt{1-\eta^2}\bxi\tag{O}\\
    \bv_{1/2}&=\bv_0'-(h/2)\nabla \Phi(\bx_0)\tag{B}\\
    \bx_{1}&=\bx_0+h\bv_{1/2}\tag{A}\\
    \bv_1'&=\bv_{1/2}-(h/2)\nabla \Phi(\bx_1)\tag{B}\\
 \bv_1&=\eta \bv_1'+\sqrt{1-\eta^2}\bxi'\tag{O}
\end{align*}
The updates denoted by the letters O, B and A correspond to exact solutions of a splitting of \eqref{eq_langevin} into three parts. These are respectively referred in the literature as momentum refreshment, acceleration and free transport parts of the Langevin dynamics; see \cite{leimkuhler2013rational}. The BAB composition reduces to the St\"ormer-Verlet update introduced in \Cref{sec_intro}. Combined with an infinitesimal momentum refreshment as $h\rightarrow0$, the OBABO composition yields a natural time discretization of the Langevin diffusion, built as an extension of the Leapfrog integrator for Hamiltonian dynamics. For this reason, this splitting scheme has received a significant interest; see \cite{bou2010pathwise,bussi2007accurate,monmarche2020high,ricci2003algorithms,scemama2006efficient}. Several other splittings of the Langevin diffusion have been proposed and studied; see \cite{alamo2016technique,sanz2021wasserstein}. This work focuses on the OBABO update in order to preserve several properties of the St\"ormer-Verlet update, useful for constructing a Metropolis correction. On the principle, our construction could rely on other integrators, as long as these are time-reversible and volume preserving; see \cite{bou2010pathwise}.

For any $\gamma>0$, the distribution of $\bz_1=(\bx_1,\bv_1)$ given $\bz_0=(\bx_0,\bv_0)$ admits a positive density $\bz_1\mapsto q_{h,\gamma}(\bz_0,\bz_1)$ with respect to Lebesgue's measure on $\R^{2d}$. This density is formally defined as the product of the two conditional densities corresponding to the Gaussian distributions of $\bx_1$ given $(\bx_0,\bv_0)$ and $\bv_1$ given $(\bx_0,\bv_0,\bx_1)$. The case $\gamma=0$ reduces to considering the deterministic St\"ormer-Verlet update $\btheta_h$. The OBABO update characterizes a Markov kernel defined for any Borel set $A$ of $\R^{2d}$ by
    $$\bfQ_{h,\gamma}(\bz_0,A)\triangleq\left\{\begin{array}{ll}
        \int_Aq_{h,\gamma}(\bz_0,\bz_1)\dd \bz_1  & \mbox{if } \gamma>0 \\
        \delta_{\btheta_h(\bz_0)}(A) & \mbox{if } \gamma=0.
    \end{array}\right.$$
    For any distribution $\nu_0$ on $\R^{2d}$, starting from $\bz_{0}\sim \nu_0$ we define the numerical Langevin trajectory for $i\ge 1$ by
    \begin{equation}\label{eq_obabo_trajectory}
           \bz_{i}\sim\bfQ_{h,\gamma}(\bz_{i-1},.)
    \end{equation}
    We refer to the synchronized processes $(\bZ_t)_{t\ge0}$ and $(\bz_{i})_{i\ge 0}$ for the respective solutions of \eqref{eq_langevin} and \eqref{eq_obabo_trajectory}, starting from $\bZ_0=\bz_{0}\sim\nu_0$ with identical momentum refreshments. This synchronization is formally ensured in the OBABO update by considering $\bxi=\bxi_{0,h/2}$ and $\bxi'=\bxi_{h/2,h}$ such that for $t> s\ge 0$
    \begin{equation}\label{eq_synchronization_langevin_obabo}
        \bxi_{s,t}\triangleq\sqrt{\frac{2\gamma}{1-e^{-\gamma (t-s)}}} \int_s^te^{-\gamma(t-u)}\dd \bW_u.
    \end{equation}
      In \Cref{prop:StrongAccuracy}, we show that numerical Langevin trajectories built upon recursive OBABO updates are strongly accurate. Similar results have been proposed and discussed in \cite{bou2010long,bou2010pathwise}. 
\begin{proposition}
\label{prop:StrongAccuracy}
 Suppose that \Cref{assumption:grad_lipschitz} holds. Let $(\bZ_t)_{t\ge0}$ and $(\bz_{i})_{i\ge 0}$ be the respective solutions of \eqref{eq_langevin} and \eqref{eq_obabo_trajectory}, synchronized with respect to \eqref{eq_synchronization_langevin_obabo}. For any fixed $d\ge 1$, $T>0$ and $\gamma\ge0$,
there exists $C>0$ such that for any square integrable start $\bZ_0\sim\nu_0$ on $\R^{2d}$ and any $t,h\in (0,T]$
$$
 \left(\E[|\bZ_{\floor*{t/h}h}-\bz_{\floor*{t/h}}|^2]\right)^{1/2}\le C\left(1+\E[|\bZ_0|^2]\right)^{1/2}h.
$$
\end{proposition}
  \Cref{prop:StrongAccuracy} ensures that, measured by the $\mathbb{L}_2$-norm, the numerical error of approximating Langevin trajectories with the OBABO update scales at worst linearly with the time-step $h>0$. Beyond its natural interpretation, this claim is used to establish several results in \Cref{sec_optimal_scaling}. A detailed proof is therefore derived in \Cref{sec:StrongAccuracy}.
    
We introduce a Metropolis correction applied to numerical Langevin trajectories of length $T>0$. This correction involves a local numerical error, defined for $h>0$ and $\bx,\by\in\R^{d}$ by
\begin{equation}\label{eq_local_error}
      \mathcal{E}_h(\bx,\by)\triangleq
    \Phi(\by)-\Phi(\bx)
    -
    \frac{1}{2}(\by-\bx)^\top \left(\nabla\Phi(\by)+\nabla\Phi(\bx)\right)
    +
    \frac{h^2}{8}\left(\left|\nabla\Phi(\by)\right|^2-\left|\nabla\Phi(\bx)\right|^2\right).
\end{equation}
Composed by $L=\lfloor T/h \rfloor$ steps, a Langevin trajectory $(\bx_0,\bv_0),...,(\bx_L,\bv_L)\in\R^{2d}$ is drawn from \eqref{eq_obabo_trajectory}. 
The trajectory is faced with an accept-reject test to compensate for the total error
\begin{equation}\label{eq_total_error}
\Delta(\bx_0,...,\bx_L)\triangleq\sum_{i=1}^L\mathcal{E}_h(\bx_{i-1},\bx_i).
\end{equation}
Applied to the positions in the OBABO discretization, the local error $\mathcal{E}_h$ reduces to the energy difference induced by the St\"ormer Verlet update. Indeed the BAB composition yields
\begin{equation}\label{eq_energy_difference}
    \mathcal{E}_h(\bx_0,\bx_1)=\Phi(\bx_1)-\Phi(\bx_0)+\frac{1}{2}(|\bv_1'|^2-|\bv_0'|^2)=-\log\left(\frac{\Pi_*(\bx_1,\bv_1')}{\Pi_*(\bx_0,\bv_0')}\right).
\end{equation}
The total error $\Delta$ corresponds to the sum of energy differences incurred by the Leapfrog integrator (disregarding the partial velocity refreshments). Therefore, it can be interpreted as the error of approximating the Hamiltonian part of Langevin dynamics. Each numerical Langevin trajectory is finally accepted with probability $1\wedge \exp\{-\Delta\}$. The resulting algorithm, named Metropolis Adjusted Langevin Trajectories (MALT), is presented hereafter; see \Cref{alg:malt}.

\begin{algorithm}[!ht]
\SetAlgoLined
\SetKwInOut{Input}{Input}\SetKwInOut{Output}{Output}\normalsize
\Input{Starting point $(\bX^0,\bV^0)\in\R^{2d}$, number of MCMC samples $N\ge 1$, step-size $h>0$, integration time $T\ge h$, and friction $\gamma\ge0$.
}
\BlankLine
\DontPrintSemicolon
Set $L\gets\lfloor T/h\rfloor$\\
\For{$n\gets1$ to $N$}{
 {\color{blue}draw a full refresh of the momentum $\bV'\sim\mathcal{N}_d(\bfzero_d,\bfI_d)$\\
    set $(\bx_0,\bv_0)\gets(\bX^{n-1},\bV')$ and $\Delta\gets0$\\
    \For{$i\gets1$ to $L$}{
	draw an OBABO step $(\bx_i,\bv_i)\sim \bfQ_{h,\gamma}((\bx_{i-1},\bv_{i-1}),.)$\hspace*{3.5em}%
        \rlap{\smash{$\left.\begin{array}{@{}c@{}}\\{}\\{}\\{}\\{}\\{}\end{array}\color{blue}\right\}%
          \color{blue}\begin{tabular}{l}Propose a \\Langevin\\ trajectory.\end{tabular}$}}\\
	update the energy difference $\Delta\gets \Delta+\mathcal{E}_h(\bx_{i-1},\bx_i)$
	}
	set $(\bX^{n},\bV^{n})\gets(\bx_L,\bv_L)$}\\
	draw a uniform random variable $U$ on $(0,1)$\\
	\If{$U>\exp\{-\Delta\}$}{reject and flip the momentum $(\bX^{n},\bV^{n})\gets(\bX^{n-1},-\bV')$}
}
\KwRet{$(\bX^1,\bV^1),\cdots,(\bX^N,\bV^N)$}.\;
\caption{Metropolis Adjusted Langevin Trajectories}\label{alg:malt}
\end{algorithm}
%\Cref{alg:malt} is comparable to HMC in terms of running cost. Particularly, it requires the same number of evaluations of the potential and its gradient per trajectory step. The only deficit with respect to the HMC is that there is in general no simple closed expression for the logarithm of the acceptance ratio, instead the local energy differences $\mathcal{E}_h(\bx_{i-1},\bx_i)$ have to be summed along the trajectory.

From \eqref{eq_total_error} and \eqref{eq_energy_difference}, we see that the total energy error $\Delta$ in \Cref{alg:ghmc} and \Cref{alg:malt} coincide when $\gamma=\alpha=0$. In other words, MALT is an extension of HMC to any choice of friction $\gamma\ge0$. We also observe that computing $\Delta$ involves the gradients of $\bx_0,\cdots,\bx_L$ which are already computed when proposing the Langevin trajectory. In that sense, running MALT does not require more gradient evaluations than running HMC. In \Cref{prop_malt_reversible}, we show that MALT defines a reversible Markov kernel with respect to $\Pi$. We recall that a kernel $\bfP$ is reversible with respect to $\Pi$ if for any Borel sets $A,B$ of $\R^d$
$$
\int_B\Pi(\dd \bx)\bfP(\bx,A)=\int_A\Pi(\dd \bx)\bfP(\bx,B).
$$
\begin{proposition}\label{prop_malt_reversible}
For any $\gamma\ge0$ and $T\ge h>0$, the sequence $(\bX_i)_{i\ge0}$ in \Cref{alg:malt} defines a Markov chain characterized by a kernel $\bfP$ reversible with respect to $\Pi$.
\end{proposition}
A proof is derived in \Cref{sec:proof_reversibility}.
The claim is known for HMC; see \cite[Remark 13]{andrieu2020general}. For $\gamma>0$, the result follows from remarking that the acceptance ratio $\exp\{-\Delta\}$ in \Cref{alg:malt} is the product of the single step acceptance ratios obtained in \cite[Eq 5.13]{bou2010pathwise}. Our analysis is more generally built upon a Markov kernel on the space of trajectories $\bz_{0:L}\triangleq(\bz_0,\cdots,\bz_L)$. In particular, the kernel is shown to be reversible with respect to the extended measure
\begin{equation}\label{eq_trajectory_target}
    \mu(\dd\bz_{0:L})\triangleq\Pi_*(\dd\bz_0)\prod_{i=1}^L \bfQ_{h,\gamma}(\bz_{i-1},\dd\bz_i).
\end{equation}
% defined on Borel sets $A$ of $\R^{d(L+1)}$.
% For $\gamma>0$, the distribution $\mu$ admits a density with respect to Lebesgue's measure. Noting the backward trajectory $\beta(\bz_{0:L})\triangleq(\varphi(\bz_L),\varphi(\bz_{L-1}),\cdots,\varphi(\bz_0))$, we obtain from \eqref{eq_total_error} and \cite[Eq 5.13]{bou2010pathwise} that
% \begin{equation}\label{eq_sum_log_ratio}
%     -\Delta(\bz_{0:L})=\sum_{i=1}^L\log\left(\frac{\Pi_*(\bz_i)q_{h,\gamma}(\varphi(\bz_i),\varphi(\bz_{i-1}))}{\Pi_*(\bz_{i-1})q_{h,\gamma}(\bz_{i-1},\bz_i)}\right)=\log\left(\frac{\mu(\beta(\bz_{0:L}))}{\mu(\bz_{0:L})}\right).
% \end{equation}
We highlight that reversibility is ensured only when full momentum refreshment are performed at the start of each trajectory. In particular, momentum flips are completely erased in \Cref{alg:malt} whereas this is not the case in \Cref{alg:ghmc} if $\alpha>0$. 
%The length of the trajectories $T>0$ in MALT is an additional degree of freedom compared to one-step Metropolis corrections aiming for the Langevin diffusion; see \cite{bou2010pathwise,bussi2007accurate,horowitz1991generalized,ottobre2016function,ricci2003algorithms,scemama2006efficient}. 
%We argue that applying a Metropolis correction to whole Langevin trajectories improves flexibility of tuning and simplifies the study of its Markov kernel. In particular, we present MALT as an extension of HMC for which the tuning of $T>0$ can be made robust by choosing a positive friction $\gamma>0$. This property is illustrated in \Cref{sec_worst_ACF} and \Cref{sec_mixing_rates} by deriving uniform bounds on the ACFs of the Langevin diffusion. It is further supported by 
In \Cref{prop:ErgodicityMALT}, %in which 
we show that MALT is ergodic in total variation for any $T\ge h>0$ if a positive friction $\gamma>0$ is chosen. The total variation between two measures $\nu,\nu'$ on $\R^d$ is denoted as $\left\|\nu-\nu'\right\|_{\rm TV}\triangleq\sup\{ |\nu(A)-\nu'(A)|$, $A$ Borel set of $\R^d\}$.
% \begin{equation}\label{eq_total_variation}
%     \left\|\nu-\nu'\right\|_{\rm TV}\triangleq\underset{A\in\mathcal{B}(\R^d)}{\sup} |\nu(A)-\nu'(A)|.
% \end{equation}
\begin{proposition}\label{prop:ErgodicityMALT}
 Let $\bfP$ be the Markov kernel of the chain $(\bX_i)_{i\ge0}$ in \Cref{alg:malt}, and suppose that $\gamma>0$. Then for 
% $\Pi$-almost every $\bx\in\R^d$
any $T\ge h>0$ and for $\Pi$-almost every $\bx\in\R^d$ we have
%any starting measure $\nu\ll\Pi$
% on $\R^d$
\[
\lim_{n\to\infty}\left\|\delta_{\bx}\bfP^n-\Pi\right\|_{\rm TV}
=
0\,.
\]
\end{proposition}
%     Noting $\bfQ_{h,\gamma}^{1}=\bfQ_{h,\gamma}$, the composition is defined recursively for $i\ge 1$ by
% $$\bfQ_{h,\gamma}^{i+1}(\bz,A)=\int_{\R^{2d}} \bfQ_{h,\gamma}^{i}(\bz,\dd\bz')\bfQ_{h,\gamma}(\bz',A).$$
% $$\bfP_{T,h,\gamma}(\bz,A)=\bfQ_{h,\gamma}^{\lfloor T/h\rfloor}(\bz,A)$$
% $$\bfP_{T,h,\gamma}^N(\bz,A)$$
A proof is derived in \Cref{sec:proof_ergodicity}. For any $\gamma>0$, the measure $\mu$ defined in \eqref{eq_trajectory_target} admits a positive density with respect to Lebesgue's measure, ensuring $\Pi$-irreducibility and aperiodicity of the Markov kernel. Beyond their convergence in total variation, ergodic chains satisfy a strong law of large numbers; see \cite[Eq 6]{roberts2004general}. We remark that a positive damping simplifies considerably the study of MALT's ergodicity, compared to standard HMC. Indeed, ergodicity in total variation for HMC obtained in \cite[Theorem 2]{durmus2020irreducibility} makes the additional assumption that
\begin{equation}
    \left(1+hM^{1/2}\omega(hM^{1/2})\right)^{\lfloor T/h\rfloor}<2,\qquad\omega(s)=(1+s/2+s^2/4),
\end{equation}
where $M>0$ is the Lipshitz constant of $\nabla\Phi$; see \Cref{assumption:grad_lipschitz}.
In particular, the integration time of HMC is required to be small: of the order $M^{-1/2}$. A positive choice of friction in \Cref{alg:malt} avoids this restrictive condition, enabling the use of longer trajectories for anisotropic targets. We refer to \cite{livingstone2019geometric} for ergodicity results on Randomized HMC. The geometric ergodicity for \Cref{alg:malt} is beyond the scope of our study.

\section{Optimal scaling limits for isotropic targets}\label{sec_optimal_scaling}
In this section, we consider the problem of tuning the time step $h>0$ in MALT for any choice of friction $\gamma\ge 0$ and integration time $T>0$. Our study connects to several results known as \textit{optimal scaling};
%To study the optimal scaling of the method with increasing dimension we use the classical setting of optimal scaling results for MCMC,
see \cite{Roberts::1997} and \cite{Roberts::1998} for initial results on Random Walk Metropolis and MALA and \cite{Beskos::2013} for HMC. 
One goal of these works is to derive scaling limits of Metropolis adjusted algorithms as the dimension $d$ goes to infinity. The obtained results rely on assuming that the target distribution has a product form; i.e. that the components are IID. This simplified framework enables the study of the accept-reject mechanism for high-dimensional targets. This study provides simple tuning guidelines of the time step to obtain a non-trivial acceptance rate. This result is convenient for assessing the number of gradient evaluations to reach a certain integration time, as a scaling of the dimension.
%We study a sequence of IID targets, meaning the marginal stationary distributions of single coordinates (with the rest of the coordinates integrated out) are independent and identical. This can summarised in the potential as:
We consider here the following assumption:
\begin{assumption}\label{Assumption:ProductForm}
The potential satisfies $\Phi(\bx)=\sum_{i=1}^d\phi(\bx(i))$ for some $\phi\in\mathcal{C}^4(\R)$, with uniformly bounded derivatives $\phi^{(k)}$ for $k=2,3,4$, and $\int_\R x^8\exp(-\phi(x)) \dd x<\infty$.
\end{assumption}
We highlight that the smoothness and integrability conditions above are similar, although slightly weaker, than those suggested in \cite[Proposition 5.5]{Beskos::2013} for HMC. Under \Cref{Assumption:ProductForm},
%Under this assumption, understanding of the dynamics of the MALT chain with such target potential essentially reduces to understanding the behavior of its single coordinates. 
the asymptotic dynamics of MALT
%of these single coordinate processes 
can be described in terms of independent components of the Langevin diffusion defined in \eqref{eq_langevin}. Throughout this section, we denote $(X_t,V_t)\in\R^2$ for the solution of a single component of \eqref{eq_langevin}
%with a one-dimensional potential $\phi$ 
initiated at stationarity: $(X_0,V_0)\sim e^{-\phi}\otimes\mathcal{N}(0,1)$. The $d$-dimensional chain generated by MALT has a potential satisfying \Cref{Assumption:ProductForm}, corresponding to a marginal potential $\phi$ (independent of the dimension). The product form of the target together with the proposal density enforces a product structure also on the proposal. In other words, the components of the numerical Langevin trajectories proposed are also independent and identically distributed. We denote $\mu_{\bx}(A)\triangleq\mu(\{\bx_{0:L}\in A\})$ the marginal measure of $\bx_{0:L}\triangleq(\bx_0,\cdots,\bx_L)$ characterized by \eqref{eq_trajectory_target}.
For any fixed choice of physical time $T$ and friction $\gamma$ we establish asymptotic normality of the total energy difference. 
\begin{theorem}\label{thm:MALT_CLT}
Suppose that \Cref{Assumption:ProductForm} holds.
Choose $h=\ell d^{-1/4}$ for some constant $\ell>0$. For any $T>0$ and $\gamma\geq0$, let $L=\lfloor T/h\rfloor$ and $\bx_{0:L}\sim \mu_{\bx}$.
%Assume potentials satisfy \Cref{Assumption:ProductForm} with the same marginal potential $\phi$ but varying dimension and take $\bx_{0:L}\sim \mu_{\bx}$.
Then as $d\to\infty$ 
%the total energy differences satisfy 
\[
\Delta(\bx_{0:L})
~\Rightarrow~
\mathcal{N}\Big(\frac{1}{2}\ell^4\Sigma,\ell^4\Sigma\Big)\,,
\]
such that for $S(x,v)\triangleq\frac{1}{12}v^3\phi^{(3)}(x)+\frac{1}{4}v\phi''(x)\phi'(x)$ we have 
\[
\Sigma~=~\E\bigg[\Big(\int_0^TS(X_t,V_t)\dd t\Big)^2\bigg]\,.
\]
\end{theorem}

The proof is given in \Cref{sec:proofs:Theorem2}. It relies on extending the optimal scaling framework for Metropolis-Hastings methods recently introduced in Section~3 of \cite{vogrinc2021counterexamples} to also include algorithms based on trajectories, such as HMC and MALT. This is achieved through the study of the asymptotic properties of the \emph{total energy difference} $\Delta$, or equivalently the log Metropolis-Hastings rates $-\Delta$, as they are called in \cite{vogrinc2021counterexamples}. Establishing strong accuracy of the trajectories for the function $S$ is crucial. Our analysis relies on approximating $S$ by a sum of Lipschitz functions for which strong accuracy follows by \Cref{prop:StrongAccuracy}. Next, we use \Cref{thm:MALT_CLT} to deliver guidelines for the tuning of MALT. We generalize the results discussed in Sections~3~and~4~of~\cite{Beskos::2013} so that they also hold for non-negative friction $\gamma$. In the special case of $\gamma=0$ we recover the results on HMC.

\begin{proposition}\label{prop:SingleCoordinateDisplacement}
Under the assumptions of \Cref{thm:MALT_CLT} the following statements hold:
\begin{enumerate}
\item[(i)] The acceptance rate satisfies, with $\Psi$ denoting the CDF of a standard Gaussian
\[
\E\left[1\wedge e^{-\Delta(\bx_{0:L})}\right]
\quad\xrightarrow{d\to\infty}\quad
a(\ell)~\triangleq~2\Psi\left(-\frac{\ell^2\sqrt{\Sigma}}{2}\right)\,.
\]
\item[(ii)] Let $f\colon\R\to\R$ be locally Lipschitz such that $|f(x)-f(y)|\leq C(|y|+|x|)|x-y|$ for some $C>0$. Denote $\Upsilon_f\triangleq\E\left[\left(f(X_T)-f(X_0)\right)^2\right]$. Then
\[
\E\left[\left(f(\bX^{n+1}(1))- f(\bX^n(1))\right)^2\right]
\quad\xrightarrow{d\to\infty}\quad
\Upsilon_f
a(\ell)\,.
\]
\item[(iii)] Set $X_T'=X_T\mathbf{1}_{[0,a(\ell)]}(U)+X_0\mathbf{1}_{(a(\ell),1]}(U)$ for a uniform random variable $U$. Then
\[
(\bX^n(1),\bX^{n+1}(1))
~\Rightarrow~
(X_0,X_T')\,.
\]
\end{enumerate}
\end{proposition}

Point (i) identifies the asymptotic average acceptance rate under the $h=\ell d^{-1/4}$ scaling.
Note that $h=\ell d^{-1/4}$ is the only decay rate of the time-step that will lead to a non-trivial distributional limit in \Cref{thm:MALT_CLT} and to a non-trivial limiting average acceptance rate. Any other decay rate will lead to a limiting acceptance rate of either zero or one (to see this formally check proof of \Cref{prop:ScalingOptimality}).
This result for general friction $\gamma$ is exactly the same as for the HMC case $\gamma=0$ studied in Section~3.3~of~\cite{Beskos::2013}.

Apart from extending the results from the HMC case to any choice of friction, point (ii) generalizes Section~3.4~of~\cite{Beskos::2013} where only the mean function is considered. The result describes the asymptotic lag-one autocorrelation of different functions which can be used as heuristic indicators of the performance of MALT. For the mean function, it corresponds to the asymptotic expected squared jump distance, a performance criterion often used for adaptive tuning of MCMC methods (see \cite{pasarica2010adaptively})). Under \Cref{Assumption:ProductForm} we can guarantee this for functions growing as fast as quadratic but with a stronger moment condition we could extend this further for functions that grow more rapidly.

Point (iii) is again an extension of the results in HMC case, studied in Section~3.5~of~\cite{Beskos::2013}. It describes the limiting behavior of each marginal coordinate of the chain generated by MALT (with other coordinates integrated out with respect to the stationary measure). If MALT is initiated in stationarity and only a single coordinate is observed, then its behavior can be described in terms of a Langevin trajectory and an independent coin. %Like in the HMC case (and unlike the RWM and MALA cases) it is important to know that the rest of the coordinates are in stationarity, as a single coordinate is not Markovian with respect to its own filtration, not even asymptotically. 
\cite[Section 4]{Beskos::2013} considers two measures of efficiency for HMC, the number of successful transitions per gradient evaluation and the expected squared jump distance per gradient evaluation. maximizing the first one is equivalent to minimizing the computational cost of a single accepted transition, while maximizing the second will also prioritizes schemes which move further.

Number of successful transitions per gradient evaluation is proportional to the acceptance rate and inversely proportional to the number of steps $L=\lfloor T/h\rfloor$. For a fixed $T$ this corresponds to maximizing
$h\E\left[1\wedge e^{-\Delta(\bx_{0:L})}\right]$.
Similar reasoning holds for other measures of efficiency. We will consider the following measures of efficiency.
For $f$ as in \Cref{prop:SingleCoordinateDisplacement}(ii) we define \emph{expected squared $f$-distance} per gradient evaluation as
\[
\frac{1}{L}\E\left[\left(f\left(\bX^{n+1}(1)\right)- f\left(\bX^n(1)\right)\right)^2\right]\,,
\]
If we set $f$ to be linear we recover the expected squared jump distance per gradient evaluation as in \cite{Beskos::2013}. This efficiency measure is proportional to the Dirichlet form of MALT evaluated at $f$ divided by the length of the trajectory. For fixed $T$ and $\gamma$ the efficiency measures corresponding to different functions $f$ asymptotically differ from each other only by a factor that is independent of the time-step. Hence, optimizing any of them is asymptotically equivalent.
We are now in position to show that the step-size decay rate $d^{-1/4}$ is optimal according to any of these asymptotic measures. 

\begin{proposition}\label{prop:ScalingOptimality}
 Let the \Cref{Assumption:ProductForm} be satisfied and let $f$ be as in \Cref{prop:SingleCoordinateDisplacement}(ii). Let $T,\ell>0$ and $\gamma\geq0$ be constants. Let $h\to0$ be a sequence of time-steps and $L=\lfloor T/h\rfloor$. 
 If either $d^{1/4}h\to0$ or $d^{1/4}h\to\infty$, then
 \[
 d^{1/4}\times \frac{1}{L}\E\left[\left(f\left(\bX^{n+1}(1)\right)- f\left(\bX^n(1)\right)\right)^2\right]
 \quad\xrightarrow{d\to\infty}\quad
0\,.
 \]
 If $d^{1/4}h\to\ell$ for some $\ell\in(0,\infty)$, then for ${\rm eff}(\ell)~\triangleq~\ell a(\ell)$ we have
 \[
 d^{1/4}\times \frac{1}{L}\E\left[\left(f\left(\bX^{n+1}(1)\right)- f\left(\bX^n(1)\right)\right)^2\right]
 \quad\xrightarrow{d\to\infty}\quad
 \Upsilon_f \times {\rm eff}(\ell)
\,.
 \]
 There exists a unique optimal $\ell^*$ maximizing ${\rm eff}(\ell^*)$, for which the corresponding optimal acceptance rate equals
 \[
 a(\ell^*)
 ~\approx~
 0.651\,.
 \]
\end{proposition}

This extends the optimal scaling results in Section~4 of \cite{Beskos::2013} to MALT. The guidelines for tuning MALT are the same as for HMC:
\[
\textit{Scale the  step $h\propto d^{-1/4}$ and tune it to accept $65.1\%$ of proposals.}
\]

It is remarkable that the optimal choice of $\ell$ does not depend on the function $f$ considered, nor the distribution $\Pi$. The same is however, not true for the choices of $T$ and $\gamma$. Using \Cref{prop:ScalingOptimality} we can deduce that for an appropriate constant $C\approx 0.619$ we have ${\rm eff}(\ell^*)=C \Sigma^{-1/4}$. This implies that the smaller the $\Sigma$ the larger time-steps optimized sampler takes. However, the constants $\Upsilon_f$ also depend on $T$ and $\gamma$ and as a consequence the choice of $T$ and $\gamma$ such that the corresponding $\ell^*(T,\gamma)$ maximizes an efficiency measure depends greatly on the choice of that efficiency measure (in terms of $f$). The findings are not consistent (and are even contradictory) across a simple selection of functions $f$ even in the case of the standard Gaussian potential.

For a standard Gaussian marginal potential $\phi(x)=\frac{x^2}{2}$ we have $\phi'(x)=x$, $\phi''(x)=1$ and $\phi^{(3)}=0$. The identity $d(X_t^2)=2X_tdX_t=2X_tV_t$, the stationarity of $(X_t,V_t)$ and Isserlis' theorem imply
\begin{align*}
\Sigma_{\gamma,T}
~=~
\frac{1}{16}\E\left[\left(\int_0^TV_tX_t\dd t\right)^2\right]
~&=~
\frac{1}{64}\E\left[\left(X_T^2-X_0^2\right)^2\right]
~=~
\frac{1}{32}\left(\E[X_0^4]-\E[X_0^2X_T^2]\right)
\\&=~
\frac{1}{16}\left(\E[X_0^2]^2-\E[X_0X_T]^2\right)
~=~
\frac{1}{16}\left(1-\rho^2_{\gamma}(T)\right)\,,
\end{align*}
where $\rho_\gamma(T)=\rm{Corr}(X_0,X_T)$ follows formula \eqref{eq_corr_langevin} with $\sigma_i=1$.

The rescaled optimal time-step $\ell_\gamma^*(T)$ is proportional to  $(1-\rho^2_\gamma(T))^{-1/4}$. We remark that the map $T\mapsto\ell_\gamma^*(T)$ fluctuates a lot for small friction especially for high correlated samples: for $\gamma=0$ it even diverges for any $T=k\pi$, $k\in\N$. At the contrary, the map $T\mapsto\ell_\gamma^*(T)$ fluctuates less as $\gamma$ increases, and becomes monotonically decreasing as soon as $\gamma\ge 2$. This suggests that the joint tuning of $\ell$ and $T$ is more stable when choosing a positive friction. 
%We speculate that this holds beyond Gaussian case as well. This reemphasises the potential of MALT as a robust alternative to HMC. 

% \lionel{
% Finally, we note that the number of gradient evaluations to reach a certain integration time $T$ depends on the correlation $\rho_\gamma(T)$ but not directly on $\gamma$. In particular, for any $\gamma$ and any choice of $T$ such that $\rho_\gamma(T)\approx 0$ we have $\ell_\gamma^*(T)$ is approximately constant.
% }
\section{Robustness to anisotropic targets}\label{sec_robustness}
In this section, we explain how MALT can improve the robustness of HMC for anisotropic targets. We first highlight a phenomenon of resonances of the Auto-Correlation Functions (ACFs) on a toy Gaussian model with heterogeneous scales.
The worst ACF of HMC is highly sensitive to the choice of trajectory length, whereas the Langevin diffusion enables a uniform control of the ACFs for an explicit damping. This phenomenon is known to limit the efficiency and complicate the tuning of HMC for anisotropic targets; see \cite{bou2017randomized, hoffman2021adaptive, neal2011mcmc}. We then extend our study to strongly log-concave targets, and relate the Langevin diffusion to another robust approach known as Randomized HMC (see \cite{bou2017randomized}). Our analysis refines the mixing rate of \cite{deligiannidis2021randomized} for RHMC, and shows that the Langevin diffusion is a limit of RHMC that achieves the fastest mixing rate. To our knowledge, this connection is new and motivates the use of partial momentum refreshment.
%We extend our study to strongly log-concave targets and connect the Langevin diffusion to Randomized HMC. We refine the mixing rate of \cite{} 

\subsection{Anisotropic Gaussian target}\label{sec_worst_ACF}
We assume that $\Phi(\bx)=\sum_{i=1}^d\bx(i)^2/(2\sigma_i^2)$, for various scales $\sigma_1,\cdots,\sigma_d>0$. The Langevin diffusion reduces to an Ornstein-Uhlenbeck process, defined through independent SDEs for $1 \le i\le d$ by
\[
\dd\begin{bmatrix}\bX_t(i)\\
\bV_t(i)
\end{bmatrix}=-\bfA_{i,\gamma}\begin{bmatrix}\bX_t(i)\\
\bV_t(i)
\end{bmatrix} \dd t+\begin{bmatrix}
0\\
\sqrt{2\gamma}\,\dd\bW_{t}(i)
\end{bmatrix}, \qquad\bfA_{i,\gamma}\triangleq\begin{bmatrix}0 & -1\\
\sigma_i^{-2} & \gamma
\end{bmatrix}.
\]
 We note $e^\bfA\triangleq\sum_{k=0}^{\infty}\bfA^k/k!$ the matrix exponential of any square matrix $\bfA$. The solution of the $i^{th}$ component of the Langevin diffusion at time $T>0$ unfolds as follows (see \cite{gardiner1985handbook}).
 \begin{equation}\label{eq_gaussian_langevin_trajectory}
    \begin{bmatrix}\bX_T(i)\\
\bV_T(i)
\end{bmatrix}=e^{-T\bfA_{i,\gamma}}\begin{bmatrix}\bX_0(i)\\
\bV_0(i)
\end{bmatrix} +\int_0^T e^{-(T-t)\bfA_{i,\gamma}}\begin{bmatrix}
0\\
\sqrt{2\gamma}\,\dd\bW_{t}(i)
\end{bmatrix}.
\end{equation}
We note $\bfA(i,j)$ the $(i,j)^{th}$ element of $\bfA$. The ACF of the $i^{th}$ component writes as
$$
\rho_{i,\gamma}(T)\triangleq{\rm Corr}(\bX_T(i),\bX_0(i))=\E[\bX_T(i)\bX_0(i)]/\sigma_i^2= e^{-T\bfA_{i,\gamma}}(1,1).
$$
We compute this matrix exponential explicitly using an eigen-value decomposition of $\bfA_{i,\gamma}$. Noting $\omega_{i,\gamma}\triangleq|(\gamma/2)^2-(1/\sigma_i)^2|^{1/2}$, we get
\begin{equation}\label{eq_corr_langevin}
    \rho_{i,\gamma}(T)=\left\{
    \begin{array}{ll}
        e^{-\gamma T/2}\big(\cos\left(T\omega_{i,\gamma}\right)+(\gamma/(2\omega_{i,\gamma}))\sin\left(T\omega_{i,\gamma}\right)\big) & \mbox{if } 0\le\gamma<2/\sigma_i \\
        e^{-T/\sigma_i}\big(1+T/\sigma_i\big) & \mbox{if } \gamma=2/\sigma_i\\
        e^{-\gamma T/2}\big(\cosh\left(T\omega_{i,\gamma}\right)+(\gamma/(2\omega_{i,\gamma}))\sinh\left(T\omega_{i,\gamma}\right)\big) & \mbox{if } \gamma>2/\sigma_i.
    \end{array}
\right.
\end{equation}
Choosing $\gamma=0$ yields Hamiltonian dynamics, with periodic ACFs: $\rho_{i,0}(T)=\cos(T/\sigma_i)$. When $\gamma>0$, the Langevin diffusion becomes ergodic and $\rho_{i,\gamma}(T)\rightarrow0$ as $T\rightarrow+\infty$. The exponential rate of convergence is optimized for the critical damping $\gamma=2/\sigma_i$. %This choice of friction corresponds to a phase transition. In the underdamped regime, i.e. when $0<\gamma<2/\sigma_i$, the complex eigenvalues of $\bfA_{i,\gamma}$ induce an oscillatory behavior in the convergence of $\rho_{i,\gamma}(T)$. In the overdamped regime, i.e. when $\gamma\ge 2/\sigma_i$, the eigen-values of $\bfA_{i,\gamma}$ are real numbers and the convergence of the ACFs is monotonic. 
In \Cref{graph_autocorr_s}, we plot the ACFs: $T\mapsto\rho_{i,\gamma}(T)$, for Hamiltonian dynamics ($\gamma=0$) and for the Langevin diffusion ($\gamma=2$). The different curves correspond to various scales $\sigma_i>0$.
%The Langevin diffusion is critically damped for a reference scale $\sigma_{\rm ref}=1$ corresponding to the solid blue line.
% \begin{figure}[!ht]
%   \centering
% \includegraphics[width=.99\linewidth]{autocorr_g.pdf}
% \caption{\it ACFs of Langevin dynamics for a fixed scale $\sigma_i=1$, and various damping parameters $\gamma\ge0$. The solid grey line corresponds to Hamiltonian dynamics ($\gamma=0$). The solid blue line corresponds to the fastest exponential rate of convergence ($\gamma=2$). Dashed lines correspond to the underdamped regime ($0<\gamma<2$). Dotted lines correspond to the overdamped regime ($\gamma>2$).
% \vspace{-0.5cm}}\label{graph_autocorr_g}
% \end{figure}
%\vspace{-0.5cm}
\begin{figure}[!ht]
  \centering\includegraphics[width=.99\linewidth]{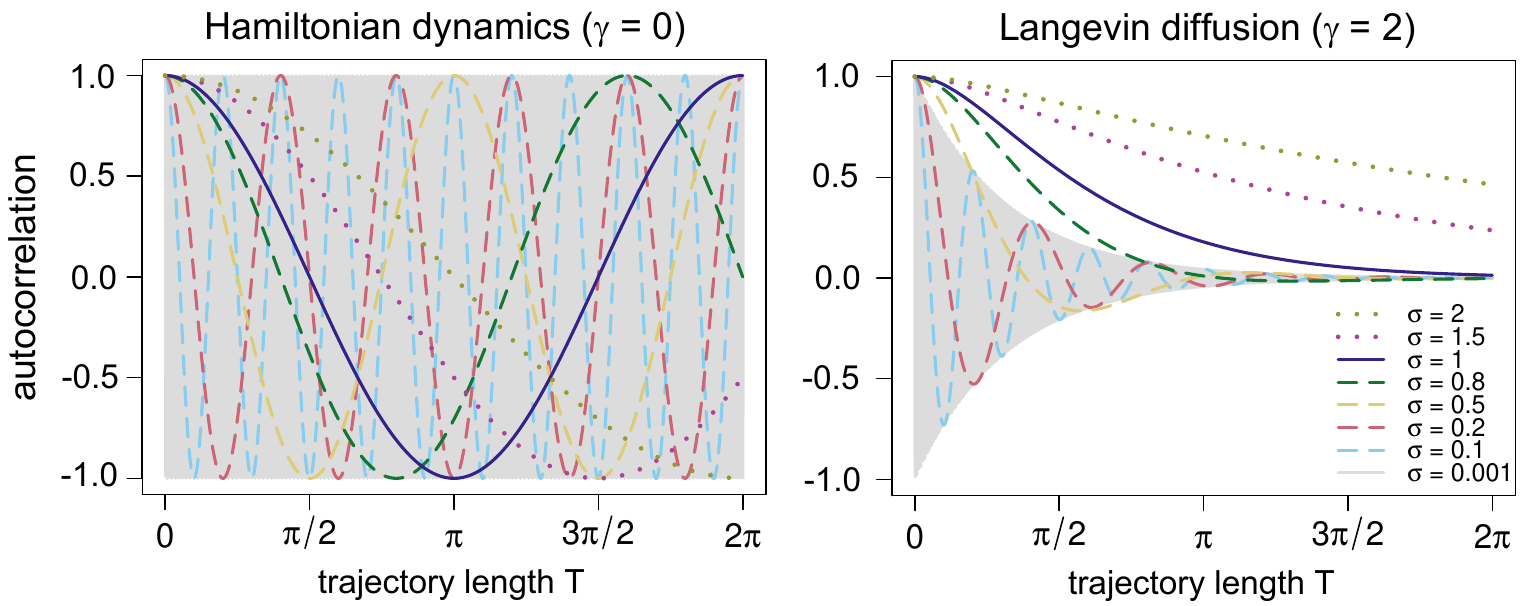}
\vspace{-0.2cm}\caption{\it ACFs of Hamiltonian dynamics ($\gamma=0$) and the Langevin diffusion ($\gamma=2$) depending on the trajectory length $T$, for various scales $\sigma_i>0$. The solid blue line corresponds to a reference scale $\sigma_{\rm ref}=1$. Dashed lines correspond to smaller scales ($\sigma_i<1$). Dotted lines correspond to larger scales ($\sigma_i>1$). The grey region corresponds to the oscillating ACF of a very thin scale ($\sigma_i=0.001$).
\vspace{-0.3cm}}\label{graph_autocorr_s}
\end{figure}
%\Cref{graph_autocorr_g} for a fixed scale $\sigma_i$, although the long term convergence to zero is optimized for $\gamma=2/\sigma_i$, on the short run the ACF of a Hamiltonian trajectory ($\gamma=0$) decays faster than any other choice of $\gamma>0$. In particular, it reaches zero the soonest, i.e. for $T=\sigma_i\pi/2$. Hamiltonian dynamics can therefore produce low correlated samples faster than Langevin dynamics for one particular component. Under heterogeneity of scales however, 

 For Hamiltonian dynamics, \Cref{graph_autocorr_s} shows that controlling the ACF of a particular scale can result in arbitrarily high correlations for other scales. The maximum of these ACFs is an erratic function, arbitrarily close to one. The presence of heterogeneous scales limits the efficiency and complicates the tuning of $T$. 
 For the Langevin diffusion, \Cref{graph_autocorr_s} illustrates the ACFs of the critical damping $\gamma=2/\sigma_{\rm ref}$ for a reference scale $\sigma_{\rm ref}=1$. We see that the ACF of any smaller scale $\sigma_i<\sigma_{\rm ref}$ is controlled by the ACF of the reference scale. This observation suggests a simple tuning rule: first choose $\gamma=2/\sigma_{\rm max}$ where $\sigma_{\rm max}\triangleq \max_{i}\sigma_i$, then select $T$ to optimize the sampling efficiency of the (principal) component corresponding to $\sigma_{\rm max}$. Adaptive tuning heuristics and algorithmic recommendations are further addressed in \cite{riou2023adaptive}.%For this choice of damping, the ACF corresponding to the largest scale yields a monotonically vanishing upper bound over all the correlations. This property highlights that the tuning problem of the integration time $T$ is easier to solve for Langevin dynamics than for Hamiltonian dynamics. 

%We also extend their robustness beyond the Gaussian framework by establishing uniform bounds on the correlations for strongly log-concave targets.

\subsection{Strongly log-concave targets: connection with Randomized HMC}\label{sec_mixing_rates}

The previous resonance phenomenon illustrates how a naive implementation of HMC with fixed integration time breaks down on an anisotropic Gaussian target.
% When properly damped, the Langevin diffusion yields a robust alternative to Hamiltonian dynamics. We propose a simple guideline for tuning the damping; see \cite{riou2023adaptive} for further recommendations. 
We highlight that robustness to anisotropic targets can be obtained by drawing $T$ randomly at each iteration to smooth the ACFs. The resulting process is often referred to as Randomized HMC (see \cite{bou2017randomized,deligiannidis2021randomized}). In this subsection, we connect the continuous time RHMC process with the Langevin diffusion and compare their exponential mixing rates for strongly log-concave targets. 
%To obtain simultaneous control of the ACFs of Hamiltonian dynamics, the main approach suggested in the literature consists of drawing at random the length of each Hamiltonian trajectory. With this approach, ACFs are averaged over multiple trajectories. The smoothing effect induced can enable control of the worst ACF, depending on the distribution chosen; see \cite{bou2017randomized, hoffman2021adaptive, neal2011mcmc}. In this work we consider exponentially distributed integration times with rate $\lambda>0$. 
 RHMC is a piecewise deterministic Markov process following Hamiltonian dynamics, with discrete momentum refreshments driven by a homogeneous Poisson process $(\bN_t)_{t\ge0}$ with rate $\lambda\ge0$. We allow for partial refreshments with persistence $\alpha\in[0,1)$, induced by independent standard Gaussian momentums $(\bxi_k)_{k\in \N}$. Under \Cref{assumption:grad_lipschitz}, it can be formally defined as the strong solution of the following SDE.
\begin{equation}\label{eq_randomized_hmc}
\dd\begin{bmatrix}\bX_{t}\\
\bV_{t}
\end{bmatrix}=\begin{bmatrix}\bV_{t}
\\
-\nabla \Phi(\bX_{t})
\end{bmatrix}\dd t+\begin{bmatrix}\bzero_d
\\
\big(\alpha\bV_{t-}+\sqrt{1-\alpha^2}\bxi_{\bN_{t-}}-\bV_{t-}\big)\,\dd\bN_{t}
\end{bmatrix}.
\end{equation}
Conditions ensuring the geometric ergodicity of the RHMC process in total variation have been established in \cite[Theorem 3.9]{bou2017randomized}. These are similar to the ones ensuring geometric ergodicity of the Langevin diffusion in \cite{mattingly2002ergodicity}. 
%More precisely, in \cite{bou2017randomized} the authors show that similar Lyapunov functions can be derived for Randomized HMC and for the Langevin diffusion. However, the minorization condition established for Randomized HMC relies on an alternative approach to account for discontinuities of the sample paths. 
Beyond this qualitative connection, to our knowledge, no quantitative comparison of the mixing rates of the two processes have been made prior to this work. Yet, explicit rates have been obtained in \cite{dalalyan2020sampling} for the Langevin diffusion, and in \cite{deligiannidis2021randomized} for RHMC, for strongly log-concave targets satisfying \Cref{assumption:strong_cvx}.
%beyond the Gaussian framework, namely
\begin{assumption}\label{assumption:strong_cvx}
The potential $\Phi\in C^2(\mathbb{R}^d)$, such that for some constants $M\ge m>0$
$$
m\bfI_d\preceq \nabla^2\Phi(\bx)\preceq M\bfI_d,\qquad \bx\in\mathbb{R}^d.
$$
\end{assumption} 
\noindent We note 
%Under \Cref{assumption:strong_cvx}, exponential mixing rates of the Langevin diffusion with respect to the $2$-Wasserstein distance have been established for various damping regimes in \cite[Theorem 1]{dalalyan2020sampling}. In particular, by choosing $\gamma=\sqrt{M+m}$, an exponential decay of the form $e^{-rt}$ with rate $r=m/\sqrt{M+m}$ is obtained. For Randomized HMC, exponential mixing rates for every persistence $\alpha\in[0,1)$ have been derived in \cite[Theorem 3]{deligiannidis2021randomized}.
% In particular, the following mixing rate and refreshment intensity are suggested:
%\begin{equation}\label{eq_rates_deligiannidis}
%r=\frac{(1+\alpha)m}{2\sqrt{M+m}}-\frac{\alpha m^{3/2}}{4(M+m)},\qquad \lambda=\frac{1}{1-\alpha^2}\left(2\sqrt{M+m}-\frac{(1-\alpha)m}{\sqrt{M+m}}\right).
%\end{equation} 
%In the sequel, we connect Randomized HMC to the Langevin diffusion, and compare the two in terms of exponential mixing rates with respect to the $2$-Wasserstein distance and the $\mathbb{L}_2(\Pi_*)$-norm. These metrics are defined as follows.
 the $2$-Wasserstein distance between two measures $\nu$ and $\nu'$ on $\R^d$ as
$$W_2(\nu,\nu')\triangleq \inf \{\E[|\bX-\bX'|^2]^{1/2}, \bX\sim\nu, \bX'\sim\nu'\}.$$
For any measure $\nu$ on $\R^{2d}$, we note $\nu_{\bx}$ its marginal measure defined on Borel sets $A$ of $\R^d$ by $\nu_{\bx}(A)\triangleq\nu(A\times \R^d)$. 
We note $\mathbb{L}_2(\Pi_*)\triangleq\{f:\R^{2d}\rightarrow \R,\, \int f^2\dd\Pi_*<\infty\}$ the set of square integrable functions with respect to $\Pi_*$. For any such $f$ and $g$, the scalar product and the norm on $\mathbb{L}_2(\Pi_*)$ are noted respectively $\langle f,g\rangle\triangleq\int f g\, \dd \Pi_*$ and $\|f\|\triangleq\langle f,f\rangle^{1/2}$. We also note $\mathbb{L}_2^0(\Pi_*)\triangleq\{f\in\mathbb{L}_2(\Pi_*),\, \int f\dd\Pi_*=0\}$ the set of centered functions in $\mathbb{L}_2(\Pi_*)$, and
$\mathbb{L}_2^0(\Pi)$
%$\mathbb{L}_2^0(\Pi)\triangleq\{f\in\mathbb{L}_2^0(\Pi_*), f(\bx,\bv)=f(\bx,\bfzero_d), (\bx,\bv)\in\R^{2d}\}$ 
the set of functions in $\mathbb{L}_2^0(\Pi_*)$ that depend only on the position. 
For any Markov process $(\bX_t,\bV_t)$ on $\R^{2d}$, we note $\bfP^t$ the transition kernel defined on Borel sets $A$ of $\R^{2d}$ by
%We use the notation $\bfP^t$ for any Markov kernel characterized by a SDE on $\R^{2d}$; e.g. solution of \eqref{eq_langevin} or \eqref{eq_randomized_hmc}. Any such kernel can be defined on Borel sets $A$ of $\R^{2d}$ by
$$
\bfP^t((\bx,\bv),A)\triangleq\mathbb{P}((\bX_t,\bV_t)\in A\,|\,(\bX_0,\bV_0)=(\bx,\bv)),\qquad (\bx,\bv)\in\R^{2d}.
$$
We note $\nu\bfP^t$ the law of $(\bX_t,\bV_t)$ starting from $(\bX_0,\bV_0)\sim\nu$. For any $f\in\mathbb{L}_2(\Pi_*)$, we note
$$
\bfP^tf(\bx,\bv)\triangleq \E[f(\bX_t,\bV_t)\,|\,(\bX_0,\bV_0)=(\bx,\bv)], \qquad (\bx,\bv)\in\R^{2d}.
$$
Mixing rates for the $2$-Wasserstein distance, and the $\mathbb{L}_2(\Pi_*)$-norm, have been obtained for the RHMC process in \cite[Theorem 3]{deligiannidis2021randomized}. In \Cref{thm:randomized_langevin_contraction}, we establish sharper rates for every persistence $\alpha\in(0,1)$ by refining their analysis.
\begin{theorem}\label{thm:randomized_langevin_contraction}
Let $\bfP^t$ be the transition kernel of the RHMC process, solution of \eqref{eq_randomized_hmc}, with persistence $\alpha\in[0,1)$ and refreshment intensity $\lambda=\frac{2\sqrt{M+m}}{1-\alpha^2}$. Suppose that \Cref{assumption:strong_cvx} holds.  Then there exist explicit constants $C,C'\le 1.56$ such that for any $t>0$
$$
W_2((\nu\bfP^t)_{\bx},\Pi)\le Ce^{-r t}W_2(\nu_{\bx}^{},\Pi),\qquad \nu=\nu_{\bx}\otimes\mathcal{N}_d(\bfzero_d,\bfI_d)
$$
%$$
%\|\bfP^tf\|\le C'e^{-r t}\|f\|,\qquad f\in\mathbb{L}_2^0(\Pi_*),\quad |f(\bx,\bv)|=|f(\bx,-\bv)|,\quad(\bx,\bv)\in\R^{2d}
%$$
$$
\|\bfP^tf\|\le C'e^{-r t}\|f\|,\qquad f\in\mathbb{L}_2^0(\Pi)
$$
where\vspace{-0.4cm}
$$
r=\frac{(1+\alpha)m}{2\sqrt{M+m}}.
$$
\end{theorem}

A proof of \Cref{thm:randomized_langevin_contraction} is derived in \Cref{proof_contraction}.  %while its limit coincides with the rate $m/\sqrt{M+m}$ established in \cite[Theorem 1]{dalalyan2020sampling} for the Langevin diffusion. 
Compared to \cite[Theorem 3]{deligiannidis2021randomized}, we strictly improve the rate $r$ obtained for any $\alpha\in(0,1)$, although the two rates get relatively close when the condition number $M/m\rightarrow\infty$. 
%More precisely, we study the distance between two copies of the process starting from different values, synchronized with the same Poisson process and Gaussian refreshments. This construction is relatively standard, although deriving sharp convergence rates relies on a carefully chosen twist of the metric to optimize the $2$-Wasserstein contraction. 
%The main argument in the proof relies on establishing uniform bounds on the coupling generator for which we obtain sufficient and necessary conditions in \Cref{proof_generator_bound}; see \eqref{eq_gen_bound}. 
Our improvement comes from the fact that the refreshment intensity $\lambda$ is chosen in order to saturate a sharper constraint; see \eqref{eq:solutions_existence}.  Restricting our attention to the actual target $\Pi$ rather than $\Pi_*$ is useful for obtaining small explicit constants $C,C'$ defined in \Cref{proof_contraction}; see \eqref{eq_C} and \eqref{eq_C'}. Similarly to \cite{dalalyan2020sampling, deligiannidis2021randomized}, our analysis is based on a synchronous coupling construction.
The main novelty of our result does not lie in the proof technique, but in the following interpretation. 
The mixing rate increases as $\alpha\rightarrow 1$, while $\lambda=2\sqrt{M+m}/(1-\alpha^2)\rightarrow\infty$. In other words, the RHMC process mixes faster when refreshments become more partial and frequent. In \Cref{prop:cv_generator}, we show that the Langevin diffusion is obtained as a limit when $\alpha\rightarrow1$. We establish convergence of the generator of the RHMC process towards the generator of the Langevin diffusion with respect to the supremum norm $\|f\|_{\infty}\triangleq\sup |f|$, for test functions in the set $C_c^{\infty}(\R^{2d})$ of infinitely differentiable functions with compact support on $\R^{2d}$.

    \begin{proposition}\label{prop:cv_generator}
   Let $\mathcal{L}_{\lambda,\alpha}^{\rm RH}$ be the generator of the RHMC process, solution of \eqref{eq_randomized_hmc}, with persistence $\alpha\in[0,1)$ and refreshment intensity $\lambda\ge0$. %with refreshment rate $\lambda\ge0$ and persistence $\alpha\in[0,1)$, 
    Let $\mathcal{L}_\gamma^{\rm LD}$ be the generator of the Langevin diffusion, solution of \eqref{eq_langevin}, with damping $\gamma\ge 0$.
    %, with damping $\gamma\ge0$.
     Suppose that $\lambda=\frac{2\gamma}{1-\alpha^2}$. Then for any $f\in C_c^{\infty}$ we have $\|\mathcal{L}_{\lambda,\alpha}^{\rm RH}f-\mathcal{L}_{\gamma}^{\rm LD}f\|_{\infty}\rightarrow 0$ as $\alpha\rightarrow 1$.
    \end{proposition}
    
 A proof of \Cref{prop:cv_generator} is derived in \Cref{proof_cv_generator}. Convergence of the generators is obtained as $\alpha\rightarrow 1$ while $\lambda=2\gamma/(1-\alpha^2)\rightarrow\infty$. %Intuitively, as the refreshments driven by the Poisson Process become more partial and frequent, they get closer to a continuous refreshment induced by the Brownian motion. 
 The convergence of infinitesimal generators is particularly useful to establish weak convergence of Markov processes; see \cite{ethier2009markov}. A formal proof of weak convergence is beyond the scope of this work. Nevertheless, \Cref{prop:cv_generator} and \Cref{thm:randomized_langevin_contraction} describe the Langevin diffusion as a limit of the RHMC process that achieves the fastest mixing rate for targets satisfying \Cref{assumption:strong_cvx}. Remarkably, the refreshment intensities coincide %\Cref{thm:randomized_langevin_contraction} 
 when $\gamma=\sqrt{M+m}$, which is known to yield the $2$-Wasserstein mixing rate $r=m/\sqrt{M+m}$ for the Langevin diffusion; see \cite[Theorem 1]{dalalyan2020sampling}.
For completeness, we extend this mixing rate to the $\mathbb{L}_2(\Pi_*)$-norm in \Cref{prop:langevin_contraction}. %This result is a direct consequence of the $2$-Wasserstein convergence previously derived in \cite[Theorem 1]{dalalyan2020sampling}. This extension is achieved by following similar arguments compared to the ones used for Randomized HMC in \Cref{thm:randomized_langevin_contraction}. 
%A sketch of proof is presented in \Cref{proof_contraction_langevin}.
\begin{proposition}\label{prop:langevin_contraction}
Let $\bfP^t$ be the transition kernel of the Langevin diffusion, solution of \eqref{eq_langevin}, 
%with friction $\gamma=\sqrt{M+m}$.
with friction $\gamma>\sqrt{M}$. Suppose that \Cref{assumption:strong_cvx} holds. 
Then there exist an explicit constant $C'\le 1.56$ such that for any $t>0$ 
$$
\|\bfP^tf\|\le C'e^{-r t}\|f\|\,,\qquad f\in\mathbb{L}_2^0(\Pi)
$$
where\vspace{-0.4cm}
$$
r=\frac{m\wedge(\gamma^2-M)}{\gamma}.
$$
%$$
%r=\frac{m}{\sqrt{M+m}}.
%$$
\end{proposition}
%This result is a direct consequence of the $2$-Wasserstein mixing rates obtained in \cite[Theorem 1]{dalalyan2020sampling}. %This extension is achieved by following similar arguments compared to the ones used for Randomized HMC in \Cref{thm:randomized_langevin_contraction}. 
A sketch of proof is presented in \Cref{proof_contraction_langevin}, using similar arguments to \Cref{proof_contraction}.
The mixing rate in \Cref{prop:langevin_contraction} is a continuous function of $\gamma^2$ that increases on $(M,M+m]$ and decreases on $[M+m,\infty)$. The optimum is achieved for $\gamma=\sqrt{M+m}$ and yields $r=m/\sqrt{M+m}$. This rate matches the limit of \Cref{thm:randomized_langevin_contraction} as $\alpha\rightarrow1$. Combined together, our results describe the Langevin diffusion as a limit of Randomized HMC achieving the fastest exponential mixing rate, for strongly log-concave targets with smooth enough densities. This observation motivates the construction of samplers based on the Langevin diffusion. Our study also enables a uniform control of the ACFs beyond the Gaussian framework, i.e. to any target $\Pi$ satisfying \Cref{assumption:strong_cvx}. Indeed for any $f:\R^d\rightarrow \R$ such that $\int f^2\dd \Pi\in(0,\infty)$, the map $g(\bx,\bv)\triangleq (f(\bx)-\int f \dd\Pi)/(\int f^2\dd\Pi)^{1/2}$ defined for $(\bx,\bv)\in\R^{2d}$ is such that $g\in\mathbb{L}_2^0(\Pi)$ and $\|g\|=1$, therefore \Cref{prop:langevin_contraction} yields
\begin{equation}\label{eq_l2_implies_correlations}
    {\rm Corr}(f(\bX_t),f(\bX_0))=\langle g,\bfP^t g\rangle\le\|g\|\|\bfP^t g\|\le  C'e^{-rt}.
\end{equation}
%In particular, we explain its foundations and compare it to previous approaches.

\section{Numerical experiments}\label{sec_comparisons}

In this section, we investigate the sampling performance of MALT when compared to several variations of HMC, when performing numerical integration. The sampling efficiency is measured by the effective sample size, either at stationarity of after a warm-up phase. We consider the problem of estimating moments of the marginal distributions of $\Pi$, such as the means and the variances. As discussed in \Cref{sec_robustness}, when using HMC on anisotropic targets, good mixing for one component can yield arbitrary poor mixing for other components. We use therefore the minimum Effective Sample Size (ESS) among the $d$ components as a measure of efficiency.
%This section is composed of two subsections. The first one explicitly describes the robustness and the performance of MALT and Randomized HMC on an anisotropic Gaussian model, when compared to standard HMC. Exact computations and numerical comparisons are conducted, which yield explicit illustrations of the problem and helps suggesting tuning heuristics. The second subsection compares MALT to GHMC, standard HMC and Randomized HMC on a Bayesian logistic regression model with real data. The posterior distribution is log-concave and anisotropic, we discuss in particular the tuning of the parameters as well as the numerical accuracies of MALT against GHMC as a function of the step-size. The number of gradient evaluations per sample collected is normalized to allow for a fair comparison between algorithms.

Let $(X_n)_{n\ge0}$ be a real random sequence such that that $\mathbb{E}[X_n]$ and $\mathbb{E}[X_nX_{n+k}]$ do not depend on $n$. Let $f$ be a square integrable function, and define the estimator of $\E[f(X_0)]$ as
$$
\Hat{f}_N=\frac{1}{N}\sum_{n=1}^Nf(X_n).
$$
We assume that the following IAC series converges absolutely
$$
{\rm IAC}_f=1+2\sum_{n=1}^\infty{\rm Corr}(f(X_n),f(X_0)).
$$
 The variance of $\Hat{f}_N$ is equivalent to $({\rm IAC}_f/N)\times\mathbb{V}(f(X_0))$ as $N\rightarrow\infty$. In our framework, $(X_n)_{n\ge 1}$ are one-dimensional functionals of successive trajectories with integration time $T>0$ (or mean integration time $T>0$ for Randomized HMC).
 %In our framework, we consider one-dimensional outputs $(X_n)_{n\ge 1}$ of successive trajectories with (mean) integration time $T>0$}.
 The unit cost of generating each variable is considered proportional to the number of gradient evaluations, therefore proportional to the integration time. To reduce the variance of $\Hat{f}_N$, one can either increase $N$ or increase $T$ to reduce ${\rm IAC}_f(T)/N$. Both solutions having linear costs, increasing $T$ is only profitable if it allows for decreasing ${\rm IAC}_f(T)$ at least linearly. This motivates the optimization of a rescaled ESS per integration time
\begin{equation*}
{\rm ESS}_{f}(T)\propto(T\times {\rm IAC}_f(T))^{-1}.
\end{equation*}
 In the sequel, we compare MALT, GHMC, standard HMC, and Randomized HMC in terms of worst ESS per integration time for various functions. We first provide a detailed explicit study of the anisotropic Gaussian model, before presenting the sampling performance of each algorithms on a Bayesian logistic regression model with real data. Experiments on a unimodal Gaussian mixture, and a multivariate Student's distribution are also included in \Cref{sec_add_experiments}.%\newline

%\textit{Example 1: Gaussian.}
\begin{subsection}{Gaussian distribution}
Assume that $\Phi(\bx)=\sum_{i=1}^d\bx(i)^2/(2\sigma_i^2)$, similarly to \Cref{sec_worst_ACF}. We first derive the explicit trajectories of MALT, RHMC and HMC, and consider the problem of estimating the means and variances of $\Pi$. The chains built upon these successive trajectories are obtained by letting $h\rightarrow 0$ in \Cref{alg:malt} for MALT and HMC. For RHMC we consider the sequence of positions evaluated at the jumping times of the continuous time process \eqref{eq_randomized_hmc} with full refreshments.

 From \eqref{eq_gaussian_langevin_trajectory}, there exist $c_{i,\gamma}(T)>0$ and IID vectors $(\bxi^n)_{n\ge 1}\sim \mathcal{N}_d(\bfzero_d,\bfI_d)$ such that successive Langevin trajectories are defined as follows, starting from $\bX^0\sim\Pi$
\begin{equation}
    \bX^n(i)=\rho_{i,\gamma}(T)\bX^{n-1}(i)+c_{i,\gamma}(T)\bxi^n(i)
\end{equation}
% where
% $$
% s_{i,\gamma}(T)\triangleq 2\gamma\int_0^T\left( e^{-(T-t)\bfA_{i,\gamma}}(2,1)\right)^2\,\dd\bW_{t}(i).
% $$
The sequence $(\bX^n(i))_{n\ge0}$ is an autoregressive process with root $\rho_{i,\gamma}(T)$. By Isserlis' theorem, any Gaussian vector $(X_1,X_2)\in \R^2$ is such that ${\rm Corr}(X_1^2,X_2^2)=({\rm Corr}(X_1,X_2))^2$ therefore
\begin{align}\label{eq_geom_corr_langevin}
\begin{split}
        {\rm Corr}(\bX^n(i),\bX^0(i))&=(\rho_{i,\gamma}(T))^n\\
    {\rm Corr}((\bX^n(i))^2,(\bX^0(i))^2)&=(\rho_{i,\gamma}(T))^{2n}
\end{split}
\end{align}
When $\gamma=2/\sigma_{\rm max}$ for MALT, the worst ACFs are achieved for the largest scale $\sigma_{\rm max}\triangleq \max_{i}\sigma_i$. These ACFs for the mean and variance correspond to $\rho_{\rm max}(T)$ and $\rho^2_{\rm max}(T)$ where
$$
\rho_{\rm max}(T)\triangleq e^{-T/\sigma_{\rm max}}(1+T/\sigma_{\rm max}).
$$
%For simplicity of exposition, we refer to MALT for the choice $\gamma=2/\sigma_{\rm max}$ and to HMC for the choice $\gamma=0$, respectively yielding the worst ACFs
% $\overline{\sigma}\triangleq \max_{i}\sigma_i$
% 
For HMC ($\gamma=0$), the worst ACFs for the mean and variance are respectively
$$
  u_{\rm max}(T)\triangleq \max_{1\le i\le d}\cos(T/\sigma_i),\qquad 
       w_{\rm max}(T)\triangleq \max_{1\le i\le d}\cos^2(T/\sigma_i).
$$
% $$
% \rho_{\rm max}(T)\triangleq e^{-T/\sigma_{\rm max}}(1+T/\sigma_{\rm max})\ge \max_{1\le i\le d}|\rho_{i,\gamma}(T)|.
% $$
We define the RHMC sequence as follows.
Let $(\tau_n)_{n\ge 1}$ be random integration times, drawn IID from the exponential distribution with rate $\lambda>0$. We denote $T=\mathbb{E}[\tau_1]=\lambda^{-1}$ the average duration of each Hamiltonian trajectory. Starting from $\bY^0\sim\Pi$, the solution is
\begin{equation}\label{eq:ar_rhmc}
    \bY^n(i)=\cos\left(\frac{\tau_n}{\sigma_i}\right)\bY^{n-1}(i)+\sigma_i\sin\left(\frac{\tau_n}{\sigma_i}\right)\bxi^n(i).
\end{equation}
\begin{proposition}
\label{prop:geom_corr_rhmc} The sequence $(\bY^{n}(i))_{n\ge 0}$ defined by \eqref{eq:ar_rhmc} is not jointly Gaussian. Moreover
\begin{align*}
    {\rm Corr}(\bY^{n}(i),\bY^0(i))=\E\left[\cos\left(\tau_1/\sigma_i\right)\right]^n=&(r_i(T))^n,\qquad r_i(T)\triangleq\frac{\sigma_i^2}{\sigma_i^2+T^{2}},\\
    {\rm Corr}((\bY^{n}(i))^2,(\bY^0(i))^2)=\E\left[\cos^2\left(\tau_1/\sigma_i\right)\right]^n=&(s_i(T))^n,\qquad s_i(T)\triangleq\frac{\sigma_i^2+2T^{2}}{\sigma_i^2+4T^{2}}.
\end{align*}
\end{proposition}

\noindent\Cref{prop:geom_corr_rhmc} is proven in \Cref{sec:proof_geom_corr_rhmc}. It shows that RHMC's worst ACFs are respectively
% none of the ACFs decays exponentially fast with $T$
$$
r_{\rm max}(T)\triangleq \frac{1}{1+(T/\sigma_{\rm max})^{2}},\qquad s_{\rm max}(T)\triangleq\frac{1+2(T/\sigma_{\rm max})^{2}}{1+4(T/\sigma_{\rm max})^{2}}.
$$
% $$
% r_{\rm max}(T)\triangleq \frac{1}{1+(T/\sigma_{\rm max})^{2}}\ge \max_{1\le i\le d}|r_i(T)|,\qquad s_{\rm max}(T)\triangleq\frac{1+2(T/\sigma_{\rm max})^{2}}{1+4(T/\sigma_{\rm max})^{2}}\ge \max_{1\le i\le d}|s_i(T)|.
% $$
  \Cref{prop:geom_corr_rhmc} also highlights some similarities and discrepancies between RHMC and MALT. In both cases, the worst ACFs are controlled by the largest scale. Yet, we observe that $r_{\rm max}$ has a quadratic decay to $0$ while $s_{\rm max}\rightarrow 1/2$ as $T\rightarrow\infty$. This lack of convergence for the square function has been noticed in \cite{hoffman2021adaptive} for uniformly drawn integration times as well. In particular, both $r_{\rm max}$ and $s_{\rm max}$ decay slower than $\rho_{\rm max}$ and $\rho_{\rm max}^2$ as $T\rightarrow\infty$. This ordering is reversed for small values of $T$ however. Numerically, we obtain that $\rho_{\rm max} <r_{\rm max}$  for $T>5.1$ while $\rho_{\rm max}^2 <s_{\rm max}$ for $T>0.8$. We now consider the problem of optimizing the worst ESS per time over $d=50$ components with heterogeneous variances $\sigma_i^2=i/d$. In \Cref{ess_gauss_graph}, we compare MALT, RHMC, and HMC for estimating the mean with $\varrho_f\in\{\rho_{\rm max},r_{\rm max},u_{\rm max}\}$, and the variance with $\varrho_f\in\{\rho_{\rm max}^2,s_{\rm max},w_{\rm max}\}$. We normalize the worst ESS per time as a proportion of the efficiency achieved for an isotropic target with an ideal HMC sampler as $T=\pi/2$ (independent samples). For $f(x)=x$ and $f(x)=x^2$, \eqref{eq_geom_corr_langevin} and \Cref{prop:geom_corr_rhmc} show that the $n^{th}$ autocorrelations are of the form $\varrho_f^{n}$, the worst ESS corresponds to
$$
{\rm ESS}_{f}(T)=\frac{\pi}{2T}\left(\frac{1-\varrho_f(T)}{1+\varrho_f(T)}\right).
$$
% , \qquad f\in\{{\rm Id},x\mapsto x^2\}.
\vspace{-0.2cm}
\begin{figure}[!ht]
  \centering
\includegraphics[width=.49\linewidth]{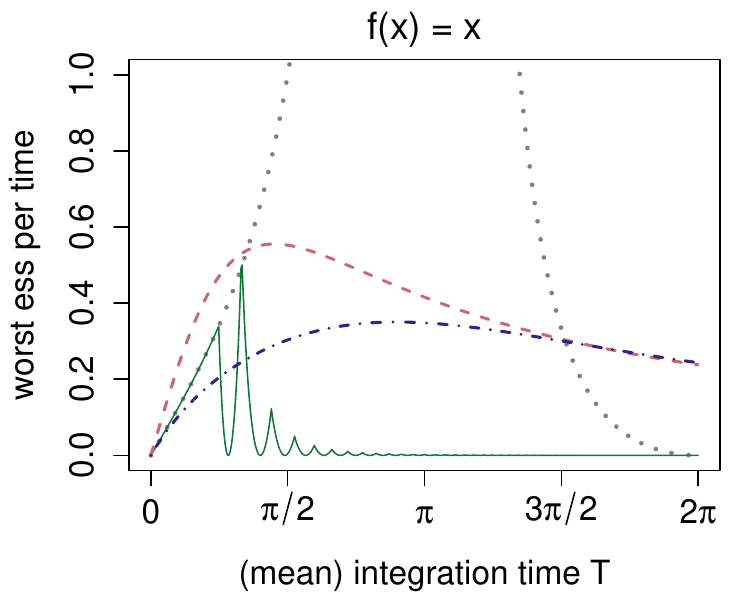}
\includegraphics[width=.49\linewidth]{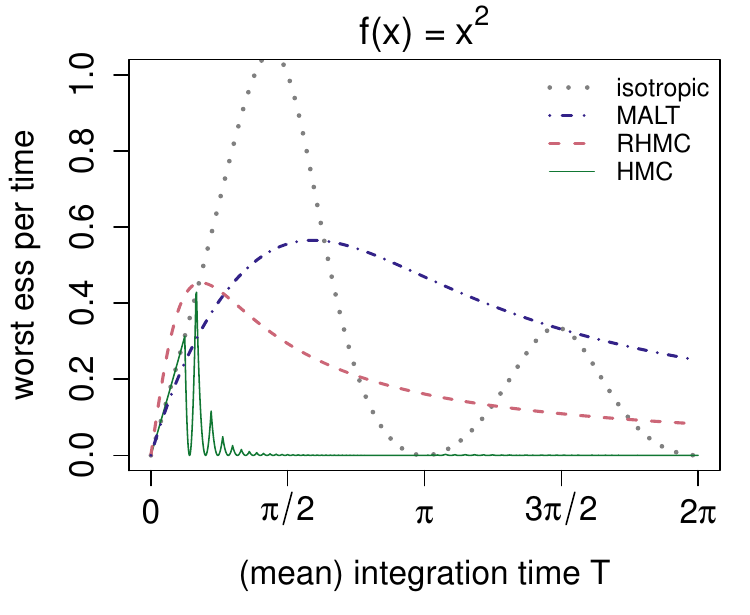}
\vspace{-0.3cm}
\caption{Gaussian distribution. \it Minimum ESS per integration time for estimating the mean and variance (resp. left and right). The dotted grey lines correspond to an ideally preconditioned HMC sampler (isotropic target). The dot-dashed blue lines correspond to MALT for $\gamma=2/\sigma_{\rm max}$. The dashed pink lines correspond to Randomized HMC with full refreshments. The solid green lines correspond to HMC.
\vspace{-0.3cm}}\label{ess_gauss_graph}
\end{figure}
\Cref{ess_gauss_graph} shows that MALT's and RHMC's worst efficiencies are smooth with respect to $T$. This observed smoothness does not depend on $d$ whereas the efficiency of HMC becomes more and more erratic as $d$ increases. This phenomenon is illustrated for $d=50$ components, as opposed to the ideal efficiency obtained for an isotropic target (independent of $d$). On $(\pi/2,3\pi/2)$, ideal HMC produces negatively correlated samples yielding super-efficient estimators of the mean, but sub-optimal estimators of the variance. We observe that RHMC achieves a better efficiency than MALT for estimating the mean whereas this ordering is reversed when estimating the variance. \Cref{prop:geom_corr_rhmc} indicates that the dashed pink lines of RHMC can be interpreted as smoothed versions of the dotted grey lines of ideal HMC, which explains intuitively the differences observed between $f(x)=x$ and $f(x)=x^2$. This discrepancy is illustrated more generally between odd and even functions in the sequel. 

The damping of MALT is tuned by grid search, although we observe that any $\gamma\in[1/\sigma_{\rm max},2/\sigma_{\rm max}]$ achieves a similar efficiency, roughly optimized for $\gamma\approx1.5/\sigma_{\rm max}$. Intuitively the worst ESS per time is optimized in a slightly underdamped regime to adapt to the finite length of the trajectories. We focus on the discrete approximation and consider the problem of choosing a (mean) number of steps $L$ to ensure a relatively good efficiency for every function. The time step $h=0.20$ is chosen to obtain acceptance rates slightly above $65\%$ for MALT for a friction $\gamma=1.5/\sigma_{\rm max}$. The same measure of efficiency is interpreted as ESS per gradient evaluation by setting $T=Lh$. In the sequel, $L$ is chosen to optimize the worst efficiency between $f(x)=x$ and $f(x)=x^2$. \Cref{ess_gauss_graph} indicates that the efficiency of HMC is quite sensitive to small variations of $T$ whereas these have little impact on the efficiency of MALT and RHMC. This problem is emphasized by the time discretization: in our example, the worst ESS for $f(x)=x^2$ is negligible for every value of $L>1$ with HMC. 

In \Cref{tab:odd_vs_even}, worst ESS per gradient evaluation are compared numerically between MALT, RHMC, and HMC for various choices of functions. The choice $L=1$ is optimal for HMC, which reduces to MALA. We also illustrate the efficiency of HMC for $L=3$. Monte Carlo estimates are computed on a sample of size $N=10^6$.
%an easy choice for $T$ is the average between the maximizers of ${\rm ESS}_f(T)$ for $f(x)=x$ and $f(x)=x^2$. We obtain numerically $T=2.3$ for MALT and $T=1.0$ for RHMC. This strategy breaks down for HMC, which efficiency is quite sensitive to small variations of $T$. Instead, its tuning relies on an accurate estimation of the maximizer for $f(x)=x^2$, achieved at $T=0.52$ for HMC.
% This strategy breaks down for HMC, which efficiency is sensitive to small variations of $T$. We consider instead choosing $T$ to be the efficiency maximizer for $f(x)=x^2$
\begin{table}[!ht]
    % \begin{center}
          \caption{\it
   Gaussian distribution. Minimum ESS per gradient evaluation for various odd/even functions.
    }\label{tab:odd_vs_even}
 \begin{tabular}{ %|m{2.5cm}||m{1cm}|m{1cm}|m{1cm}|m{1cm}|m{1cm}|m{1cm}|  }
%|p{2.5cm}||p{1cm}|p{1cm}|p{1cm}|p{1cm}|p{1cm}|p{1cm}|  }
%|c||c|c|c|c|c|c|c|c|}
|C{2.4cm}|C{0.95cm}|C{0.95cm}|C{0.95cm}|C{0.95cm}|C{0.95cm}|C{0.95cm}|C{0.95cm}|C{0.95cm}|}
 \hline
 &\multicolumn{4}{c|}{odd}&\multicolumn{4}{c|}{even} \\
 \hline
 $f(x)$& $x$& $x^3$ &${\rm sgn}(x)$&$\sin(x)$ & $x^2$ & $x^4$    &$e^{-|x|}$&   $\cos(x)$\\
 \hline
  MALT: $L=8$ & 0.25& 0.31& 0.31  & 0.27   &\textbf{0.40}& \textbf{0.42}    &\textbf{0.43}& \textbf{0.40}\\
 RHMC: $L=5$   &\textbf{0.40} & \textbf{0.43}& \textbf{0.45}    &\textbf{0.41}&   0.29& 0.31    &0.31&   0.29\\
 HMC: $L=3$ &  0.19& 0.25& 0.26&0.21  &0.00 & 0.00&0.00 &0.00\\
  MALA ($L=1$) &  0.06& 0.08& 0.09&0.07  &0.12 & 0.12&0.16 &0.13\\
 \hline
\end{tabular}    
    % \end{center}
    % \captionsetup{justification=justified}
    %  \caption{\it
    % \begin{flushleft}
    % Relative ESS per time for estimating the mean and variance (resp. left and right). The solid blue lines correspond to Langevin trajectories for $\gamma=2$. The dashed red lines correspond to Randomized Hamiltonian trajectories with full refreshments. The dotted grey lines correspond to Hamiltonian trajectories of fixed length.  The ESS per time is measured relatively to uncorrelated HMC achieved for $T=\pi/2$.
    % \end{flushleft}
    % }\label{tab:odd_vs_even}
\end{table}
%\vspace{-0.5cm}

%\begin{table}[H]
    % \begin{center}
 %         \caption{\it
  %  Worst ESS per time for various odd/even functions (numerical).
  %  }\label{tab:odd_vs_even}
% \begin{tabular}{ %|m{2.5cm}||m{1cm}|m{1cm}|m{1cm}|m{1cm}|m{1cm}|m{1cm}|  }
%|p{2.5cm}||p{1cm}|p{1cm}|p{1cm}|p{1cm}|p{1cm}|p{1cm}|  }
%|c||c|c|c|c|c|c|}
%|C{2.9cm}|C{0.9cm}|C{0.9cm}|C{1.1cm}|C{1cm}|C{0.9cm}|C{0.9cm}|C{1.1cm}|C{1cm}|}
 %\hline
 %&\multicolumn{4}{c|}{odd}&\multicolumn{4}{c|}{even} \\
% \hline
 %$f(x)$& $x$& $x^3$ &${\rm atan}(x)$&$\sin(x)$ & $x^2$ & $x^4$    &$e^{-|x|}$&   $\cos(x)$\\
 %\hline
  %MALT: $T=2.3$ & 0.34& 0.42& 0.35  & 0.37   &\textbf{0.55}& \textbf{0.58}    &\textbf{0.59}& \textbf{0.56}\\
% RHMC: $T=1.0$   &\textbf{0.52} & \textbf{0.58}& \textbf{0.62}    &\textbf{0.55}&   0.39& 0.42    &0.50&   0.40\\
 %HMC: $T=0.52$ &  0.21& 0.29& 0.24&0.24  &0.43 & 0.49&0.44 &0.44\\
 %\hline
%\end{tabular}    
    % \end{center}
    % \captionsetup{justification=justified}
    %  \caption{\it
    % \begin{flushleft}
    % Relative ESS per time for estimating the mean and variance (resp. left and right). The solid blue lines correspond to Langevin trajectories for $\gamma=2$. The dashed red lines correspond to Randomized Hamiltonian trajectories with full refreshments. The dotted grey lines correspond to Hamiltonian trajectories of fixed length.  The ESS per time is measured relatively to uncorrelated HMC achieved for $T=\pi/2$.
    % \end{flushleft}
    % }\label{tab:odd_vs_even}
%\end{table}
%\vspace{-0.5cm}
\Cref{tab:odd_vs_even} shows that MALT and RHMC achieve a better performance than standard HMC or MALA. Neither MALT nor RHMC dominates the other for every function: MALT achieves slightly higher ESS for even functions while the ordering is reversed for odd functions.
\end{subsection}

\begin{subsection}{Bayesian logistic regression}
    This subsection presents a numerical comparison between MALT, GHMC, standard HMC, and Randomized HMC, on a Bayesian logistic regression model with real data. The posterior distribution is both log-concave and anisotropic. We first compare the numerical accuracies of MALT and GHMC, and show that MALT allows for choosing a larger step-size than GHMC. We then show that MALT outperforms GHMC, and standard HMC in terms of sampling performance, while it is proven competitive with Randomized HMC. The sampling performance is measured empirically by the minimum ESS across dimensions when estimating the posterior mean and marginal variances.
    
The dataset is chosen from the Framingham Heart Study\footnote{see: https://github.com/GauravPadawe/Framingham-Heart-Study/tree/master; file=framingham.csv}. After removing the missing data, the resulting dataframe is composed of $n=3656$ observations, a binary response variable indicating the presence of a Coronary Heart Disease, and $d=15$ explanatory variables, composed of several risk factors; see hyperlink. The explanatory variables are centered and scaled before performing the logistic regression, in order to reduce the conditionning of the Hessian of the log-likelihood. This data pre-processing allows the condition number at the MLE to go from $\kappa=2.3\times 10^7$ to $\kappa=657$. This residual condition number is enough to impact severely the sampling performance of standard HMC, which is optimized for 1 leapfrog step. We ran MALT, GHMC, standard HMC and Randomized HMC, starting from a draw of the Laplace approximation of the posterior distribution. We let the chain warm for 100 iterations, after which we collect $N=100\, 000$ samples for each method. The number of leapfrog steps of each method optimizes its minimum ESS per gradient evaluation. We choose to thin the chains in order to normalize the amount of storage required between algorithms. Both standard HMC and GHMC use 1 leapfrog step per iteration, so we thin these chains by $L$, the number of MALT steps. Randomized HMC is optimized for a number of steps roughly equal to $L/2$, so we thin the RHMC chain by a factor 2. As a result, independently of the method, each of the samples collected has been obtained with $L$ gradient computations.

We first compare the accuracies of MALT and GHMC as a function of $h$, when approximating the Langevin diffusion. Based on the Gaussian heuristic of last subsection: we tune $\gamma$ as the inverse square root of the maximum eigen value of the covariance matrix of $\Pi$; see also \cite{riou2023adaptive}. On \Cref{plots_excel_graph}, we plot the acceptance rate and the minimum ESS among the components when estimating the mean, as a function of the step-size, for both MALT and GHMC.

\begin{figure}[!ht]
  \centering
\includegraphics[width=.99\linewidth,trim=0.5cm 0.3cm 0.7cm 0.5cm]{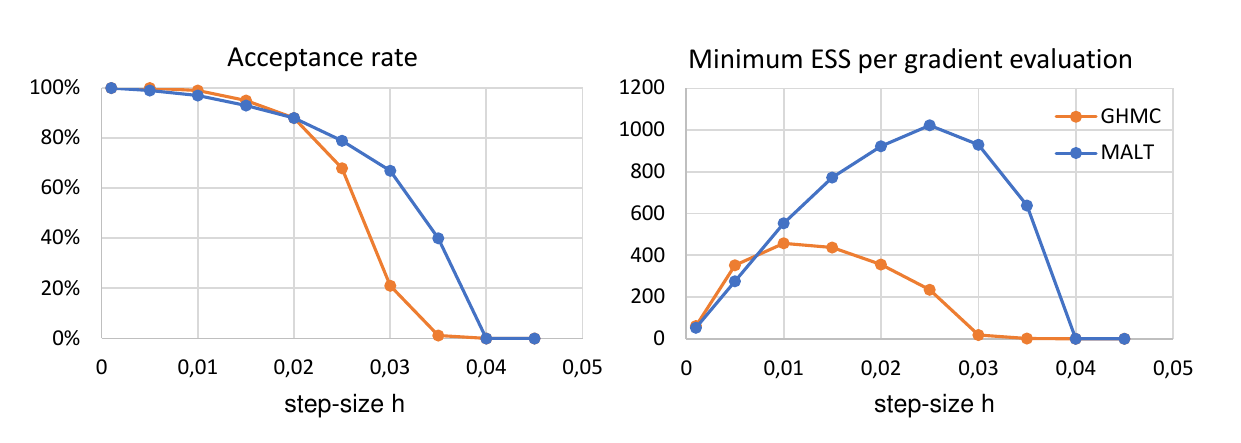}
\vspace{-0.3cm}
\caption{Bayesian logistic regression. \it Acceptance rates and minimum ESS per gradient evaluation (when estimating the posterior means) as a function of the step-size for MALT and GHMC.
\vspace{-0.3cm}}\label{plots_excel_graph}
\end{figure}
We observe that MALT dominates GHMC for both metrics. We see that MALT allows for choosing a larger step-size than GHMC to reach the same acceptance rate. Moreover, the minimum ESS of MALT is optimized when the acceptance rate is close to $80\%$, while for GHMC, the minimum ESS is optimized when the acceptance rate is much closer to one, which yields an even smaller step size and complicates its tuning. Both of these observations support MALT’s construction, which avoids momentum flips, but also puts less constraint on the energy error. An additional interpretation of this fact can be deduced from \Cref{sec_malt}: the overall numerical error of MALT unfolds as
$
\Delta=\sum_{i=1}^L\mathcal{E}_i
$,
where 
$$\mathcal{E}_i=\Phi(\bx_i)-\Phi(\bx_{i-1})+\big(|\bv_i'|^2-|\bv_{i-1}'|^2\big)/2$$ 
is the energy difference incurred by the $i^{th}$ leapfrog step. As a consequence, the probability of accepting an entire trajectory for MALT is higher than the probability of accepting $L$ consecutive leapfrog steps for GHMC. Indeed:
$$
\PP_{\rm GHMC}^{\rm accept}=e^{-\sum_{i=1}^L\max(0\,,\,\mathcal{E}_i)} \leq e^{-\max\left(0\,,\,\sum_{i=1}^L\mathcal{E}_i\right)} =\PP_{\rm MALT}^{\rm accept}.
$$

In \Cref{tab:logreg}, we compare the optimal sampling performances of MALT, GHMC, standard HMC and Randomized HMC. We compute both the minimum ESS when estimating the mean and the minimum ESS when estimating the marginal variances. Both metrics are normalized by the number of gradient computations. We also precise the tuning values for the step-sizes, acceptance rates, number of leapfrog steps per iteration and damping. 

\begin{table}[!ht]
    % \begin{center}
          \caption{\it
   Bayesian Logistic Regression. Minimum ESS per gradient evaluation when estimating the posterior means and the marginal variances, for MALT, GHMC, HMC and RHMC, and their corresponding tuning parameters.
    }\label{tab:logreg}
 \begin{tabular}{ %|m{2.5cm}||m{1cm}|m{1cm}|m{1cm}|m{1cm}|m{1cm}|m{1cm}|  }
%|p{2.5cm}||p{1cm}|p{1cm}|p{1cm}|p{1cm}|p{1cm}|p{1cm}|  }
%|c||c|c|c|c|c|c|c|c|}
|C{3cm}|C{1.5cm}|C{1.5cm}|C{1.2cm}|C{1.2cm}|C{0.8cm}|C{0.8cm}|}
 \hline
&\multicolumn{2}{c|}{Min ESS per grad}&\multicolumn{4}{c|}{Tuning parameters} \\
 \hline
Algorithm& means & variances &$h$& \% accept& $L$ &$\gamma$ \\
 \hline
  MALT & 1023& \textbf{1413}& 0.025  & 79\%& 36 & 2   \\
 GHMC   & 457 & 576 & 0.01     &99\% & 1 & 2\\
 Standard HMC &  54& 118& 0.025 &81\% & 1 & NA\\
  Randomized HMC &  \textbf{1237}& 989& 0.025 &85\%  & 18 & NA\\
 \hline
\end{tabular}    
\end{table}

We observe that MALT outperforms standard HMC, whose sampling performance is severely impacted by the heterogeneity of scales of the target distribution. We observe that the performance of MALT and Randomized HMC are competitive: none of the two algorithms dominates the other for estimating both functions. Akin to the last subsection, we see that MALT is preferable when estimating the variances while Randomized HMC is preferable when estimating the means. Numerical comparisons between MALT and variations of HMC are further discussed in \cite{riou2023adaptive}. In particular, the question of adaptive tuning is addressed, and multiple experiments are conducted on a benchmark of real data Bayesian models. 

\end{subsection}

\section{Discussion}Our work introduces MALT: a new Metropolis Adjusted sampler to approximate the kinetic Langevin diffusion. We have shown both theoretically and empirically that MALT preserves several useful properties of HMC whereas GHMC does not, which allows MALT to use a larger step-size than GHMC and improves its sampling performance. We have shown that MALT is significantly more robust, and offers a better sampling performance than standard HMC on anisotropic targets. We have compared MALT to another method praised for its robustness to anisotropic targets known as Randomized HMC. We have shown theoretically that the Langevin diffusion achieves a faster mixing rate for strongly log-concave measure, while we have demonstrated that both algorithms are competitive when performing numerical integration, on explicit toy models as well as on a Bayesian model with real data. Our work opens many prospects, among which the question of MALT’s adaptive tuning is addressed in \cite{riou2023adaptive}, with further numerical comparisons on multiple Bayesian models. Our work also provides a more efficient way to construct Metropolis adjusted samplers for skew-reversible Markov processes. This path will be investigated in future works.

\section{Acknowledgements} LRD was supported by the EPSRC grant EP/R034710/1. JV was supported by 
%a UK Engineering and Physical Sciences Research Council 
the EPSRC grants EP/R022100/1 and EP/T004134/1. We wish to thank the following people for helpful discussions at various stages of the project: Nicolas Chopin, George Deligiannidis, Samuel Livingstone, Jesús María Sanz-Serna and Giorgos Vasdekis. 

\newpage
\appendix
\section{Table of notations}\label{sec_notations}
%\subsection{Table of notations by order of first appearance - appendix}

%Preprint
\renewcommand{\arraystretch}{1.22}

\begin{table}[!ht]
    % \begin{center}
          %\caption{\it Table of notations by order of first appearance - manuscript}
          %\label{tab:notations_manuscript}
 \begin{tabular}{ %|m{2.5cm}||m{1cm}|m{1cm}|m{1cm}|m{1cm}|m{1cm}|m{1cm}|  }
%|p{2.5cm}||p{1cm}|p{1cm}|p{1cm}|p{1cm}|p{1cm}|p{1cm}|  }
%|c||c|c|c|c|c|c|c|c|}
|C{1.6cm}|C{1cm}|p{10.4cm}|}
 \hline
\textbf{Notation} & \textbf{Section} & \textbf{Definition}  \\
 \hline
  $\Pi$ & 1& Target density with respect to Lebesgue's measure on $\R^d$ \\
 $\Pi_*$   & 1 & Extended target density with respect to Lebesgue's measure on $\R^{2d}$\\
$\Phi$ &  1& Potential function $\Phi:\R^d\rightarrow\R$ such that $\Phi(\bx)\propto \exp\{-\Phi(\bx)\}$\\
  $|\bx|$ &  1 &  Euclidean norm of $\bx\in\R^d$\\
  $M$ & 1 & Lipschitz constant of $\nabla \Phi$\\
  $h$ & 1 & Step-size, $h>0$\\
  $T$ & 1 & Integration time, $T>0$\\
  $\btheta_h$ & 1& Störmer-Verlet update with step-size $h>0$, $\btheta_h:\R^{2d}\rightarrow\R^{2d}$\\
  $L$ & 1& Number of leapfrog steps, $L=\lfloor T/h\rfloor$\\
  $\gamma$ &1 & Damping (a.k.a. friction) of the Langevin diffusion\\
  $\varphi$ & 1 & Momentum flip function, $\varphi(\bx,\bv)=(\bx,-\bv)$\\
  $N$ & 1& Sample size\\
  $\bW_t$ & 1& Standard Brownian Motion on $\R^d$ at time $t>0$\\
  $\bxi$ & 2 & Standard Gaussian random vector on $\R^d$\\
  $\eta$ & 2 & Shorthand notation for $e^{-\gamma h/2}$\\
  $\mathbf{Q}_{h,\gamma}(\bz,.)$ & 2 & Markov kernel corresponding to the OBABO update from $z=(\bx,\bv)\in\R^{2d}$\\
  $q_{h,\gamma}(\bz,.)$ & 2 & Corresponding density with respect to Lebesgue's measure (when $\gamma>0$)\\
  $\delta_{\bx}$ & 2& Dirac measure at $\bx\in\R^d$\\
  $\mathcal E$ & 2& Energy error incurred by one leapfrog step\\
  $\Delta$ & 2& Total energy error (of an entire trajectory)\\
  $\bfP$ & 2 & Markov kernel corresponding to one iteration of MALT\\
  $\mu$ & 2 & Extended invariant distribution on the space of trajectories: $\R^{2d(L+1)}$\\
  $\|\nu-\nu'\|_{\rm TV}$ & 2 & Total Variation distance between $\nu$ and $\nu'$\\
  $\bx_{0:L}$  & 3 & Discrete trajectory of the positions $\bx_{0:L}=(\bx_0,\cdots, \bx_L)$, on $\R^{d(L+1)}$\\
  $\mu_{\bx}$ & 3 & Marginal distribution of $\mu$ with respect to the positions\\
  $\phi$ & 3 & Univariate potential function when $\Pi$ has product form\\
  $\Rightarrow$ & 3 & Weak convergence (of measure)\\
  $\ell$ & 3 & Re-scaled step-size, $h=\ell\times d^{-1/4}$\\
  $a(\ell)$ & 3 & Asymptotic acceptance rate of MALT, as $d\rightarrow\infty$\\
  $\ell^*$ & 3 & Asymptotically optimal (re-scaled) step-size, such that $a(\ell^*)=0.651$\\
  $\sigma_i$ & 4 & Scale of the $i^{th}$ component of $\Pi$\\
$\sigma_{\max}$ & 4 & Maximum scale among the components of $\Pi$\\
  $\rho_i$ &4& Autocorrelation function (wrt $T>0$) of the $i^{th}$ component of $\Pi$\\
    $\bN_t$ & 4& Homogeneous Poisson Process with intensity $\lambda>0$ at time $t>0$\\
   $m$ & 4 & Strong convexity constant of $\Phi$\\
   $W_2(\nu,\nu')$ & 4 & $2$-Wasserstein distance between $\nu$ and $\nu'$ wrt the Euclidean norm \\
   $\mathbb{L}_2(\Pi_*)$ & 4 & Set of square integrable functions with respect to $\Pi_*$\\
   $\langle f,g\rangle$ & 4 & Scalar-product between $f$ and $g$ on $\mathbb{L}_2(\Pi_*)$\\
   $\|f\|$ & 4 & Norm of $f$ on $\mathbb{L}_2(\Pi_*)$, defined as $\|f\|=\langle f,f\rangle^{1/2}$\\
   $\mathbb{L}_2^0(\Pi_*)$ &4& Set of centered functions in $\mathbb{L}_2(\Pi_*)$\\
$\mathbb{L}_2^0(\Pi)$&4&Set of functions in $\mathbb{L}_2^0(\Pi_*)$ that depend only on the position\\
$\bfP^t$ & 4& Markov transition kernel wrt $t>0$, either of RHMC or the Langevin diffusion\\
$\|f\|_{\infty}$ & 4 & Supremum norm of functions $\|f\|_{\infty}\triangleq\sup |f|$\\
$C_c^{\infty}(\R^{2d})$ & 4 & Set of infinitely differentiable functions with compact support on $\R^{2d}$\\
 $\mathcal{L}_{\lambda,\alpha}^{\rm RH}$ & 4 & Generator of RHMC, with persistence $\alpha\in[0,1)$ and refreshment rate $\lambda>0$ \\
 $\mathcal{L}_\gamma^{\rm LD}$& 4 & Generator of the Langevin diffusion, with damping $\gamma>0$\\
${\rm IAC}_{f}$& 5 & Integrated Auto-Correlation relative to the test function $f$\\
 ${\rm ESS}_{f}$&5 & Effective Sample Size relative to the test function $f$\\
 $\pi$& 5 & 3.14159265358979323846264338327950288419716939937510 (more or less)\\
 $\kappa$ & 5& Condition number of the Hessian of the log-likelihood at the MLE\\
 \hline
\end{tabular}    
\end{table}

\begin{table}[!ht]
    % \begin{center}
         % \caption{\it Table of notations by order of first appearance - appendix}\label{tab:notations_manuscript}
 \begin{tabular}{ %|m{2.5cm}||m{1cm}|m{1cm}|m{1cm}|m{1cm}|m{1cm}|m{1cm}|  }
%|p{2.5cm}||p{1cm}|p{1cm}|p{1cm}|p{1cm}|p{1cm}|p{1cm}|  }
%|c||c|c|c|c|c|c|c|c|}
|C{1.7cm}|C{1cm}|p{10.3cm}|}
 \hline
\textbf{Notation} & \textbf{Section} & \textbf{Definition}  \\
 \hline
  $\mathbb{E}^{\bz_0}[.]$ & B& Conditional expectation starting from $\bz_0\in\R^{2d}$ \\
 $\|.\|_{\mathbb{L}_2}^{\bz_0}$ & B& Corresponding $\mathbb{L}_2$-norm, defined as $\|.\|_{\mathbb{L}_2}^{\bz_0}=\mathbb{E}^{\bz_0}[(.)^2]^{1/2}$  \\
 $\bx^*$   &B & Minimizer of $\Phi$ on $\R^d$\\
$\mathcal{D}$ &  B& Euclidean norm of $\bx^*$\\
  $\Gamma$ &  B &  Standard Gaussian measure on $\R^d$\\
    $\bfG$ & B & Proposal Markov kernel on the space of trajectories (with initial refreshment)\\
  $\beta(\bz_{0:L})$ & B & Backward trajectory with flipped momentums, $(\varphi(\bZ_L),\cdots,\varphi(\bz_0))$\\
 $\bfM$ & B & Metropolis correction kernel on the space of trajectories (deterministic)\\
$\bfK$ & B & Markov kernel defined as $\bfK=\bfG\bfM\bfG$, coincides marginally with $\bfP$\\
  $\varepsilon_h$ & C & Univariate energy error incurred by one leafrog step\\
  $\Delta_h$ & C & Univariate energy error over a whole trajectory\\
    $\Delta_{h,j}$ & C & Energy error of the $j^{th}$ component over a whole trajectory, IID copy of $\Delta_h$\\
  $S$ & C &  Function defined on $\R^{2}$ as $S(x,v)=\frac{1}{12}v^3\phi^{(3)}(x)+\frac{1}{4}v\phi''(x)\phi'(x)$\\
 $\binom{n}{j_1,j_2,\dots,j_m}$ & C& Multinomial coefficient, defined as $\binom{n}{j_1,j_2,\dots,j_m}=\frac{n!}{j_1!j_2!\dots j_m!}$\\
  $\Psi$ & C & Standard Gaussian Cumulative Distribution Function\\
    $\psi$ & C & Standard Gaussian Probability Density Function\\
  $\mathcal{L}_{\lambda,\alpha}^{\rm Couple}$ & D & Generator of a synchronous coupling of RHMC processes, wrt $\alpha>0$ and $\lambda>0$\\
  $|z|_{\bfA}$& D& Weighted $\bfA$-norm of $\bz$ wrt a pos. def. matrix $\bfA$, defined as $|z|_{\bfA}=(z^\top\bfA z)^{1/2}$\\
  $W_{q,\bfA}(\nu,\nu')$ & D& $q$-Wasserstein distance between $\nu$ and $\nu'$ wrt the weighted $\bfA$-norm\\
    $(\bfP^t)^*$ & D & Adjoint of the semigroup $\bfP^t$\\
  $(\mathcal{L}_{\lambda,\alpha}^{\rm Couple})^*$ & D & Adjoint of the generator $\mathcal{L}_{\lambda,\alpha}^{\rm Couple}$\\
  $ \|f\|_{{\rm Lip},\bfA}$ & D & Lipschitz norm, defined as  $ \|f\|_{{\rm Lip},\bfA}=\sup_{\bz_1\neq\bz_2}|f(\bz_1)-f(\bz_2)|/|\bz_1-\bz_2|_{\bfA}$\\
  $\mathcal{L}^{\rm H}$ & D & Generator of Hamiltonian dynamics, a.k.a Liouville operator\\
  $\mathcal{R}_\alpha^{\rm PP}$ & D& Generator of discrete refreshments induced by a Poisson Process, wrt $\alpha\in[0,1)$\\
  $\mathcal{R}^{\rm BM}$ & D& Generator of continuous refreshments induced by a Brownian Motion \\
 \hline
\end{tabular}    
\end{table}

%Preprint
\renewcommand{\arraystretch}{1}

\section{Proofs of Section 2}
\subsection{Proof of \Cref{prop:StrongAccuracy}}\label{sec:StrongAccuracy}
We refer to a generic theorem established in \cite[Theorem 1.1]{milstein2013stochastic}. The claim of the proposition follows from checking that sufficient conditions to invoke this theorem are satisfied. This task reduces to establishing local accuracy for the numerical Langevin trajectories. We introduce a few notations before stating these conditions, established in \Cref{lem:local_accuracy}. 

For any $t\ge0$, $i\ge0$ and any function $f$ such that $f(\bZ_t,\bz_i)$ is integrable, we denote the conditional expectation starting from $\bz_0$ by $\mathbb{E}^{\bz_0}[f(\bZ_t,\bz_i)]\triangleq\mathbb{E}[f(\bZ_t,\bz_i)\,|\,\bZ_0=\bz_0]$. Whenever $f(\bZ_t,\bz_i)$ is square integrable, we denote the $\mathbb{L}_2$-norm from $\bz_0$ by $\|f(\bZ_t,\bz_i)\|_{\mathbb{L}_2}^{\bz_0}\triangleq \mathbb{E}^{\bz_0}[f(\bZ_t,\bz_i)^2]^{1/2}$.

% In the following
% $\bZ_t=(\bX_t,\bV_t)$ denotes the Langevin diffusion and $\bz_k=(\bx_k,\bv_k)$ denotes its numerical approximation for a given time-step $h>0$. TO BE DEFINED EXPLICITLY
% For any $0\le s\le t<\infty$ we define the stochastic flow
% \begin{equation}\label{eq:stochastic_flow}
%     \bpsi_{t,s}(\bx,\bv)=\left(\bx,e^{-\gamma(t-s)}\bv+\sqrt{2\gamma}\int_s^te^{-\gamma (t-u)}\dd \bW_u\right).
% \end{equation}
% We denote $\btheta_h$ the one-step St\"ormer-Verlet update
% $$
% \btheta_h(\bx,\bv)=\begin{pmatrix*}\bx+h\bv-h^2/2\nabla\Phi(\bx)\\\bv-(h/2)(\nabla\Phi(\bx)+\nabla\Phi(\bx+h\bv-(h^2/2)\nabla\Phi(\bx)))\end{pmatrix*}
% $$
% We introduce the numerical approximation as follows
% \begin{equation*}
% \bz_k=\bpsi_{(k+1)h,kh+h/2}\circ \btheta_h\circ\bpsi_{kh+h/2,kh}(z_0)
% \end{equation*}

\begin{lemma}\label{lem:local_accuracy} (Local accuracy) Suppose that \Cref{assumption:grad_lipschitz} holds. For any fixed $d\ge 1$, $T>0$ and $\gamma\ge0$,
there exists $C>0$ such that for any $h\in (0,T]$ and $\bz_0=(\bx_0,\bv_0)\in\R^{2d}$
\begin{align*}
\left|\E^{\bz_0}[\bZ_h-\bz_1]\right|&\le C(1+|\bz_0|^2)^{1/2}h^{2}\\
    \left(\E^{\bz_0}[|\bZ_h-\bz_1|^2]\right)^{1/2}&\le C(1+|\bz_0|^2)^{1/2}h^{3/2}
\end{align*}
\end{lemma}
 Combined with \cite[Theorem 1.1]{milstein2013stochastic}, this result yields the claim of the proposition. The proof of \Cref{lem:local_accuracy} is derived hereafter.
The density $\Pi\propto e^{-\Phi}$ is integrable, therefore the potential $\Phi\in C^1(\mathbb{R}^d)$ has a minimum, which is also a zero of $\nabla \Phi$. In other words there exists $\bx^*\in\arg\min \Phi$ such that $\nabla\Phi(\bx^*)=\bfzero_d$, and we denote $\mathcal{D}\triangleq|\bx^*|$. \Cref{assumption:grad_lipschitz} ensures that for any $\bx\in \R^d$, we have $|\nabla \Phi (\bx)|\le M(|\bx|+\mathcal{D})$. We make repeated use of this bound in the sequel. We recall that $\mathcal{L}_\gamma^{\rm LD}$ stands for the generator of the Langevin diffusion; see \Cref{sec_mixing_rates}.

We set $d\ge 1$, $T>0$ and $\gamma\ge0$ to be fixed for the rest of the proof.
We first show that the constant $A\triangleq(1+2T(d\gamma+M^2\mathcal{D}^2))e^{(2M^2+2)T}$ is such that for any $\bz_0=(\bx_0,\bv_0)\in\R^{2d}$ and $t\in(0,T]$
\begin{equation}\label{eq:second_moment_langevin}
    \E^{\bz_0}[|\bZ_t|^2]\le A(1+|\bz_0|^2).
\end{equation}
For $f(\bx,\bv)=|\bx|^2+|\bv|^2$, Young's inequality applied to the generator $\mathcal{L}_\gamma^{\rm LD}$ yields
\begin{align*}
    \mathcal{L}_\gamma^{\rm LD}f(\bx,\bv)&=-2\langle\gamma \bv +\nabla\Phi(\bx),\bv\rangle+2\langle \bv,\bx\rangle +2d\gamma\\
    &\le -2\gamma |\bv|^2+(|\nabla\Phi(\bx)|^2+|\bv|^2)+(|\bx|^2+|\bv|^2)+2d\gamma\\
    &\le (2M^2+2) f(\bx,\bv)+ 2(d\gamma+M^2\mathcal{D}^2)
\end{align*}
Dynkin's formula applies to the norm-like function $f$; see \cite[Theorem 1.1]{meyn_tweedie_1993} for a justification. For $a=2M^2+2$ and $b=2(d\gamma+M^2\mathcal{D}^2)$ this bound yields
$$
    \E^{\bz_0}[f(\bZ_t)]=f(\bz_0)+\int_0^t\E^{\bz_0}[\mathcal{L}_\gamma^{\rm LD}f(\bZ_s)]\dd s\le (bt+ f(\bz_0))+a\int_0^t\E^{\bz_0}[f(\bZ_s)]\dd s.
$$
Gr\"onwall's inequality in its integral form yields
$
\E^{\bz_0}[f(\bZ_t)]\le (bt+ f(\bz_0))e^{at}
$, and \eqref{eq:second_moment_langevin} follows for $A=(1+bT)e^{aT}$.

We apply the inequality \eqref{eq:second_moment_langevin} in the sequel to bound the local errors between $\bZ_h$ and $\bz_1$. To this end, we write down their respective explicit solutions. Integrating the Langevin SDE defined in \eqref{eq_langevin} yields the following solution; see \cite[Lemma 7.1, eqs 43, 44]{bou2010long}.
\begin{align*}
    \bV_h&=e^{-\gamma h}\bv_0-h\nabla\Phi(\bx_0)-\int_0^h((h-s)\nabla^2\Phi(\bX_s)\bV_s+(1-e^{-\gamma (h-s)})\nabla \Phi (\bX_s))\dd s+\bG_h\\
    \bX_h&=\bx_0+h\bv_0-\int_0^h(h-s)(\nabla\Phi(\bX_s)+\gamma \bV_s)\dd s+\sqrt{2\gamma}\int_0^h(h-s)\dd \bW_s
\end{align*}
where 
$$
\bG_h\triangleq\sqrt{2\gamma}\int_0^he^{-\gamma(h-s)}\dd \bW_s.
$$
% From \eqref{eq:stochastic_flow}, we recall that there exist IID random vectors  $\bxi_k\sim\mathcal{N}_d(\bfzero_d,\bfI_d)$ such that for any $k\in\mathbb{N}$
% $$
% \bpsi_{(k+1)h/2,kh/2}(\bx,\bv)=(\bx,\eta \bv\sqrt{1-\eta^2}\bxi_k).
% $$
% The numerical approximation unfolds as
Define $\bxi=\bxi_{0,h/2}$ and $\bxi'=\bxi_{h/2,h}$ with respect to \eqref{eq_synchronization_langevin_obabo}.
The OBABO update unfolds as
\begin{align}\label{eq:OBABOupdate}
    \bv_1&=\eta\left(\eta \bv_0 +\sqrt{1-\eta^2}\bxi-h\nabla\Phi(\bx_0)-(h/2)(\nabla\Phi(\bx_1)-\nabla\Phi(\bx_0))\right)+\sqrt{1-\eta^2}\bxi'\nonumber\\
    \bx_1&=\bx_0+h(\eta \bv_0 +\sqrt{1-\eta^2}\bxi)-(h^2/2)\nabla\Phi(\bx_0)
\end{align}
From \eqref{eq_synchronization_langevin_obabo} we obtain that $\bG_h=\sqrt{1-\eta^2}(\eta\bxi+\bxi')$.
As a result, the difference between the Langevin solution and the OBABO update yields
\begin{align*}
    \bV_h-\bv_1=&-h(1-\eta)\nabla\Phi(\bx_0)+\frac{h}{2}(\nabla\Phi(\bx_1)-\nabla\Phi(\bx_0))-\int_0^h(h-s)\nabla^2\Phi(\bX_s)\bV_s\dd s\\&+\int_0^h(1-e^{-\gamma (h-s)})\nabla \Phi (\bX_s)\dd s\\
    \bX_h-\bx_1=&\,h(1-\eta) \bv_0 +\int_0^h(h-s)[\nabla\Phi(\bx_0)-\nabla\Phi(\bX_s)-\gamma \bV_s]\dd s+\bG_h^{\prime}
\end{align*}
where
$$
\bG_h^{\prime}\triangleq\sqrt{2\gamma}\int_0^h(h-s)\dd \bW_s-h\sqrt{1-\eta^2}\bxi.
$$
Young's inequality yields
\begin{equation*}
    \E[|\bG_h^{\prime}|^2]\le 2d(2\gamma h^3/3+\gamma h^3)\le 4d\gamma h^3.
\end{equation*}
The vector $\bG_h^{\prime}$ is centered. Applying Jensen's and Young's inequalities yields the decompositions
\begin{align}
    \left|\E^{\bz_0}[\bZ_h-\bz_1]\right|&\le\|\bV_h-\bv_1\|_{\mathbb{L}_2}^{\bz_0}+\|\bX_h-\bx_1-\bG_h^{\prime}\|_{\mathbb{L}_2}^{\bz_0}\label{eq:decomposition_1}\\
    \left(\E^{\bz_0}[|\bZ_h-\bz_1|^2]\right)^{1/2}&\le \sqrt{2}\left(\|\bV_h-\bv_1\|_{\mathbb{L}_2}^{\bz_0}+\|\bX_h-\bx_1-\bG_h^{\prime}\|_{\mathbb{L}_2}^{\bz_0}+h^{3/2}\sqrt{4d\gamma}\right)\label{eq:decomposition_2}
\end{align}
The claim of \Cref{lem:local_accuracy} follows from bounding the two norms on the right hand sides. For any $s\in (0,h]$, we have $\|\bZ_s\|_{\mathbb{L}_2}^{\bz_0}\le \sqrt{A}(1+|\bz_0|^2)^{1/2}$ by \eqref{eq:second_moment_langevin}. Minkowski's inequality yields
\begin{align*}
\|\bV_h-\bv_1\|_{\mathbb{L}_2}^{\bz_0}&\le(\gamma h^2M/2)(|\bx_0|+\mathcal{D})+(h^2M/2)\left(|\bv_0|+\|\bxi\|_{\mathbb{L}_2}+(h/2)(|\bx_0|+\mathcal{D})\right)\\&+\int_0^h(h-s)M\|\bV_s\|_{\mathbb{L}_2}^{\bz_0}\dd s+\gamma\int_0^h(h-s)M(\|\bX_s\|_{\mathbb{L}_2}^{\bz_0}+\mathcal{D})\dd s\\
&\le(h^2M/2)\bigg[(\gamma+\sqrt{T}/2)(|\bx_0|+\mathcal{D})+|\bv_0|+\sqrt{d}+(1+\gamma)\underset{0\le s\le h}{\sup}\|\bZ_s\|_{\mathbb{L}_2}^{\bz_0}\bigg]\\
&\le(M/2)\left[(2\gamma+\sqrt{T}) (1+\mathcal{D})+2(1+\sqrt{d})+(1+\gamma)\sqrt{A}\right](1+|\bz_0|^2)^{1/2}h^2
\end{align*}
\begin{align*}
\|\bX_h-\bx_1-\bG_h^{\prime}\|_{\mathbb{L}_2}^{\bz_0}&\le(\gamma h^2/2)|\bv_0|+\int_0^h(h-s)\left(M|\bx_0|+M\|\bX_s\|_{\mathbb{L}_2}^{\bz_0}+\gamma\|\bV_s\|_{\mathbb{L}_2}^{\bz_0}\right)\dd s\\
&\le(h^2/2)\left[\gamma  |\bv_0|+ M|\bx_0|+(M+\gamma)\underset{0\le s\le h}{\sup}\|\bZ_s\|_{\mathbb{L}_2}^{\bz_0}\right]\\
&\le (1/2)\left[(\gamma+M)(1+A^{1/2})\right](1+|\bz_0|^2)^{1/2}h^2
\end{align*}
As a result, for $B\triangleq2M(\gamma+\sqrt{T}/2) (1+\mathcal{D})+2M(1+\sqrt{d})+(\gamma+\gamma M+ M)A^{1/2}$ we obtain
$$
\|\bV_h-\bv_1\|_{\mathbb{L}_2}^{\bz_0}+\|\bX_h-\bx_1-\bG_h^{\prime}\|_{\mathbb{L}_2}^{\bz_0}\le B(1+|\bz_0|^2)^{1/2}h^2.
$$
Plugging this bound into the decompositions \eqref{eq:decomposition_1} and \eqref{eq:decomposition_2}, we get  for any $h\in (0,T]$
\begin{align*}
    \left|\E^{\bz_0}[\bZ_h-\bz_1]\right|&\le B(1+|\bz_0|^2)^{1/2}h^2\\
    \left(\E^{\bz_0}[|\bZ_h-\bz_1|^2]\right)^{1/2}&\le (\sqrt{2T}B+\sqrt{8d\gamma})(1+|\bz_0|^2)^{1/2}h^{3/2}
\end{align*}
This yields the claim of \Cref{lem:local_accuracy} for $C=B\vee(\sqrt{2T}B+\sqrt{8d\gamma})$.
\subsection{Proof of \Cref{prop_malt_reversible}}\label{sec:proof_reversibility}
% For $\gamma=0$, MALT reduces to HMC, and the claim of the proposition follows from \cite[Remark 13]{andrieu2020general}. We now suppose that $\gamma>0$. The distribution $\mu$ admits a density with respect to Lebesgue's measure on $\R^{2d(L+1)}$ such that
% $$
% \mu(\bz_{0:L})=\Pi_*(\bz_0)\prod_{i=1}^Lq_{h,\gamma}(\bz_{i-1},\dd \bz_i).
% $$
% Noting the backward trajectory $\beta(\bz_{0:L})\triangleq(\varphi(\bz_L),\varphi(\bz_{L-1}),\cdots,\varphi(\bz_0))$, we obtain from \eqref{eq_total_error} and \cite[Eq 5.13]{bou2010pathwise} that
% \begin{equation}\label{eq_sum_log_ratio}
%     -\Delta(\bz_{0:L})=\sum_{i=1}^L\log\left(\frac{\Pi_*(\bz_i)q_{h,\gamma}(\varphi(\bz_i),\varphi(\bz_{i-1}))}{\Pi_*(\bz_{i-1})q_{h,\gamma}(\bz_{i-1},\bz_i)}\right)=\log\left(\frac{\mu(\beta(\bz_{0:L}))}{\mu(\bz_{0:L})}\right).
% \end{equation} 
% Denote the corresponding kernel $$\bfF(A)\triangleq\delta_{\beta(\bz_{0:L})}(A)$$
% We introduce the stationary measure of the trajectory $\bz_{0:L}$,  defined by
% \[
% \mu(\dd\bz_{0:L})
% =
% \Pi_*(\dd\bz_0)\prod_{i=1}^L Q_{h,\gamma}(\bz_{i-1},\dd\bz_i)\,.
% \]
% We note $\mu_{\bx}(A)\triangleq\mu(\{\bx_{0:L}\in A\})$ the marginal measure of $\bx_{0:L}$
% defined on Borel sets $A$ of $\R^{d(L+1)}$. 
% $$
% \mu_{\bx}(A)\triangleq\mu\left(\left\{\bx_{0:L}\in A\right\}\,,\, \bv_{0:L}\in \R^{d(L+1)}\right)
% $$
We introduce $\Gamma(B)\triangleq\PP(\bxi\in B)$, where $\bxi\sim\mathcal{N}_d(\bfzero_d,\bfI_d)$, defined for any Borel set $B$ of $\R^d$. We also define the Gibbs update, corresponding to the conditional distribution of $\mu$ given $\bx_0$.
\begin{equation}\label{eq_gibbs}
    \bfG(\bz_{0:L}, \dd \bz_{0:L}')\triangleq\delta_{\bx_0}(\dd \bx_0')\Gamma(\dd \bv_0')\prod_{i=1}^L\bfQ_{h,\gamma}(\bz_{i-1},\dd \bz_i)
\end{equation}
The Gibbs kernel $\bfG$ is reversible with respect to $\mu$ by construction. Built upon a deterministic proposal of the backward trajectory $\beta(\bz_{0:L})\triangleq(\varphi(\bz_L),\varphi(\bz_{L-1}),\cdots,\varphi(\bz_0))$, we introduce a Metropolis update defined for any Borel set $A$ of $\R^{2d(L+1)}$
$$
\bfM(\bz_{0:L}, A)\triangleq(1\wedge e^{-\Delta(x_{0:L})})\delta_{\beta(\bz_{0:L})}(A)+(1-1\wedge e^{-\Delta(x_{0:L})})\delta_{\bz_{0:L}}(A).
$$
The Metropolis kernel $\bfM$ is also reversible with respect to $\mu$. For $\gamma=0$, this follows from \eqref{eq_total_error}, \eqref{eq_energy_difference} and \cite[Theorem 3]{andrieu2020general}. For $\gamma>0$ the distribution $\mu$ admits a density with respect to Lebesgue's measure. From \eqref{eq_total_error} and \cite[Eq 5.13]{bou2010pathwise} we obtain that the corresponding density $\mu$ is such that
\begin{equation}\label{eq_sum_log_ratio}
    -\Delta(\bx_{0:L})=\sum_{i=1}^L\log\left(\frac{\Pi_*(\bz_i)q_{h,\gamma}(\varphi(\bz_i),\varphi(\bz_{i-1}))}{\Pi_*(\bz_{i-1})q_{h,\gamma}(\bz_{i-1},\bz_i)}\right)=\log\left(\frac{\mu(\beta(\bz_{0:L}))}{\mu(\bz_{0:L})}\right).
\end{equation} 
Finally, we consider the cycle $\bfK=\bfG\bfM\bfG$ defined on Borel sets $A$ of $\R^{2d(L+1)}$
$$
\bfK(\bz_{0:L}, A)\triangleq\int\bfG(\bz_{0:L}, \dd \bz_{0:L}')\bfM(\bz_{0:L}', \dd \bz_{0:L}'')\bfG(\bz_{0:L}'', A).
$$
The palindromic structure of $\bfK=\bfG\bfM\bfG$ ensures reversibility with respect to $\mu$. Since the transition $\bfG(\bz_{0:L},.)$ only depends on the starting position $\bx_0\in \R^d$ and $\Pi$ is the marginal of $\mu$, we obtain that $\bfP(\bx_0,A)\triangleq\bfK(\bz_{0:L},A\times \R^{d(2L+1)})$ defines marginally a Markov kernel on $\R^d$, reversible with respect to $\Pi$. In particular, the distribution of $(\bX_n)_{n\ge0}$ in \Cref{alg:malt} coincides with the distribution of a Markov chain generated by $\bfP$.

\subsection{Proof of \Cref{prop:ErgodicityMALT}}\label{sec:proof_ergodicity}
The claim of the proposition follows from establishing that $\bfP$ is both $\Pi$-irreducible and aperiodic; see \cite[Theorem 4]{roberts2004general}. We prove these two results here. For any $\gamma>0$ and any position $\bx_0\in\R^d$, the kernel $\bfG$ in \eqref{eq_gibbs} defines a density with respect to Lebesgue's measure on $\R^{d(2L+1)}$
$$
g(\bx_0,(\bv_0,\bz_{1:L}))\triangleq\psi(\bv_0)\prod_{i=1}^Lq_{h,\gamma}(\bz_{i-1},\bz_i).
$$
By definition, this conditional density is positive everywhere, and $\Delta(\bx_{0:L})<\infty$ for any $\bx_{0:L}\in\R^{d(L+1)}$. Subsequently, for any Borel set $A$ of $\R^d$ such that $\Pi(A)>0$, we have
\begin{equation}\label{eq_irreducibility}
    \bfP(\bx_0,A)\ge \int g(\bx_0,(\bv_0,\bz_{1:L}))(1\wedge e^{-\Delta(x_{0:L})})\dd (\bv_0,\bz_{1:L})>0
\end{equation}
Since $\bfP(\bx_0,A)>0$ for any $\bx_0\in\R^d$ and any $A$ such that $\Pi(A)>0$, we conclude that $\bfP$ is $\Pi$-irreducible. Suppose now that $\bfP$ is periodic. Then there exist two disjoint sets $A_1,A_2$ with $\Pi(A_1)>0$ and $\Pi(A_2)>0$ such that $\bfP(\bx_0,A_2)=1$ for any $\bx_0\in A_1$. However \eqref{eq_irreducibility} implies that $\bfP(\bx_0,A_1)>0$ for any such $\bx_0\in A_1$, yielding a contradiction. We conclude that $\bfP$ is aperiodic.
\section{Proofs of Section 3}

\subsection{Proof of \Cref{thm:MALT_CLT}}\label{sec:proofs:Theorem2}

Analogously to \eqref{eq_local_error} and \eqref{eq_total_error} denote for $x,y\in\R$
\begin{equation}\label{eq:LocalErrorCoordinate}
    \varepsilon_h(x,y)
    ~=~
    \phi(y)-\phi(x)-\frac{y-x}{2} \left(\phi'(y)+\phi'(x\right))
    -
    \frac{h^2}{8}\left(\phi'(y)^2-\phi'(x)^2\right)
\end{equation}
and for $x_{0:L}\in\R^{L+1}$
\begin{equation}\label{eq:TotalErrorCoordinate}
    \Delta_{h}(x_{0:L})
    ~=~
    \sum_{i=1}^L\varepsilon_h(x_{i-1},x_i)\,.
\end{equation}
The product structure of the potential assumed in \Cref{Assumption:ProductForm} enforces a product structure also on the associated total energy differences in the following sense
\begin{equation}\label{eq:ProductFormDelta}
\Delta(\bx_{0:L})
~=~
\sum_{j=1}^d\Delta_{h,j}(\bx_{0:L}(j))\,,
\end{equation}
where $\Delta_{h,j}(\bx_{0:L}(j))$ are IID copies of a random variable $\Delta_h(\bx_{0:L})$ given by a single component of the Langevin trajectory $x_{0:L}$, with respect to the one-dimensional potential $\phi$ (initiated in stationarity).

To prove that total energy difference random variables $\Delta(\bx_{0:L})$ satisfy a form of CLT, we use an extension of the framework introduced in \cite[Section~3]{vogrinc2021counterexamples}. More precisely, \Cref{thm:MALT_CLT} is as a direct application of \cite[Theorem~8]{vogrinc2021counterexamples}. To invoke it, we need to understand the asymptotic behavior of a single component of the Langevin trajectory as the dimension increases. Specifically, we need to verify the following two conditions:

\begin{proposition}\label{prop:DeltaMALT}
Let $T>0$, $\gamma\geq0$ and $\ell>0$ and take a sequence of time-steps $h\to 0$ and $L=\lfloor T/h\rfloor$. Assume the one-dimensional potential $\phi$ satisfies \Cref{Assumption:ProductForm} and let $x_0\sim\Pi$. Then
\begin{enumerate}
    \item[(i)]
\[
    \frac{1}{h^4}\E\left[\Delta_h^2(x_{0:L})\right]
    \quad\xrightarrow{h\to0}\quad
    \Sigma
    ~=~
    \E\left[\left(\int_0^TS(X_t,V_t)\dd t\right)^2\right]
\]
    \item[(ii)]
\[
    \frac{1}{h^4}\E\left[\Delta_h^2(x_{0:L})\mathds{1}_{\Delta_h>h}\right]
    \quad\xrightarrow{h\to0}\quad
    0\,.
\]
\end{enumerate}
\end{proposition}

A proof is given is \Cref{sec:DeltaMALTproof}.

Some care is required to justify that the results of \cite[Section~3]{vogrinc2021counterexamples} do actually carry over from the classical Metropolis-Hastings setting to MALT.
The results of \cite[Section~3]{vogrinc2021counterexamples} concern objects called \emph{log Metropolis-Hastings random variables} that are related to the Kullback-Leibler divergence between the forward $\Pi(\dd x)q(x,\dd y)$ and reverse $\Pi(\dd y)q(y,\dd x)$ transition kernels. The total energy difference random variables defined here are a generalisation, they can also be related to the Kullback-Leibler divergence between the forward distribution $\mu(z_{0:L})$ and the skew-backward distribution $\mu(\beta (z_{0:L}))$.

The proofs of \cite[Section~3]{vogrinc2021counterexamples} depend entirely on the symmetry property of log-MH random variables featured in the following proposition. They do not rely on the exact definition of log-MH random variables, only on their symmetric property. Therefore, the only thing we need to do in order to extend the results to the trajectory setting is to verify this symmetry is still satisfied, the proofs from that point on remain literally the same.

\begin{proposition}\label{prop:log-AR_properties}
Let $\Delta(\bx_{0:L})$ be a total energy difference function of MALT and abbreviate $\bx_{L:0}=(\bx_L,\dots,\bx_0)$. Then 
\begin{enumerate}
\item[(i)] $\Delta(\bx_{0:L})~=~ -\Delta(\bx_{L:0})$.
\item[(ii)] Let $f: \R\to\R$ be a measurable functions such that $f\circ\Delta$ is integrable with respect to
$\mu$. Then 
% (with respect to $\mu(\dd z_{0:L})$)
$\E_\mu[f(-\Delta )e^{-\Delta }]=\E_\mu[f(\Delta)]$.
\end{enumerate}
\end{proposition}

\begin{proof} 
Part (i) of the proposition
is immediate from \eqref{eq_local_error} and \eqref{eq_total_error}. To establish Part (ii) we split the cases $\gamma>0$ and $\gamma=0$.
In the positive friction case, by \eqref{eq_sum_log_ratio} the measure $\mu$ has a positive density with respect to Lebesgue measure on $\R^{2d(L+1)}$ such that
\[
\mu(\bz_{0:L})e^{-\Delta(\bx_{0:L})}
~=~
\mu(\beta(\bz_{0:L}))\,.
\]
Combining this equality to the change of variable $\bz'_{0:L}=\beta(\bz_{0:L})$ with unit Jacobian, we obtain
\begin{align*}
 \E_\mu[|f(-\Delta(\bx_{0:L}))|e^{-\Delta(\bx_{0:L})} ]
 ~&=~
 {\int}_{\R^{2d(L+1)}} |f(-\Delta(\bx_{0:L}))|e^{-\Delta(\bx_{0:L})} \mu(\bz_{0:L})\dd \bz_{0:L}
 \\&=~
 {\int}_{\R^{2d(L+1)}} |f(\Delta(\bx_{L:0}))|\mu(\beta(\bz_{0:L}))\dd \bz_{0:L}
 \\&=~
  {\int}_{\R^{2d(L+1)}}  |f(\Delta(\bx'_{0:L}))|\mu(\bz'_{0:L})\dd\bz'_{0:L}
  ~<~
  \infty\,.
\end{align*}
% For the third equality we have used change of variables $\bw_i=(\bu_i,\bv_i)=\varphi(\bz_{L-i})$ for $i=0,\cdots,L$ with unit Jacobian determinant. 
This ensures that the quantity $\E[f(-\Delta(\bx_{0:L}))e^{-\Delta(\bx_{0:L})}]$ is integrable. Repeating the calculation without the absolute values yields the desired identity. Now, when $\gamma=0$ the measure $\mu$ is characterized by the density $\Pi_*$ on $\R^{2d}$. Each trajectory is a deterministic function of its starting point $\bz_{0:L}=(\bz_0,\cdots,\btheta_h^L(\bz_0))$. 
% Denote with $\btheta^i_{h,x}$ the position part of the function $\btheta^i_h$. 
We apply the change of variable $\bz'_0=\varphi\circ\btheta^{L}_h(\bz_0)$ with unit Jacobian. Denote $\bz'_{0:L}=(\bz'_0,\cdots,\btheta_h^L(\bz'_0))$ and remark that time-reversibility of the Leapfrog update yields $\bz'_{0:L}=(\varphi\circ\btheta_h^L(\bz_0),\cdots,\varphi(\bz_0))$. Since $\varphi$ leaves $\Pi_*$ invariant, we obtain
%We make use of equations \eqref{eq_total_error} and \eqref{eq_energy_difference} together wih part (i) for the second equality and change of variables $\bw_0=(\bu_0,\bv_0)=\varphi(\btheta^{L}_h(\bz_0))$ with unit Jacobian together with $\Pi_*(\varphi(\bw_0))=\Pi_*(\bw_0)$ for the next one:
\begin{align*}
 \E_\mu[|f(-\Delta(\bx_{0:L}))|e^{-\Delta(\bx_{0:L})} ]&\quad=~
 {\int}_{\R^{2d}} \left|f\left(-\Delta\left(\bx_{0:L}\right)\right)\right|e^{-\Delta\left(\bx_{0:L}\right)} \Pi_*(\bz_{0})\dd \bz_0
 \\&\quad=~
 {\int}_{\R^{2d}} \left|f\left(\Delta\left(\bx_{L:0}\right)\right)\right|\Pi_*(\btheta^L_h(z_0))\dd \bz_0
 \\&\quad=~
  {\int}_{\R^{2d}}  \left|f\left(\Delta\left(\bx'_{0:L}\right)\right)\right|\Pi_*(\bz'_0)\dd\bz'_0
  ~<~
  \infty\,.
\end{align*}
% \begin{align*}
%  &\E[|f(-\Delta(\bx_{0:L}))|e^{-\Delta(\bx_{0:L})} ]
%  \\&\quad=~
%  {\int}_{\R^{2d}} \left|f\left(-\Delta\left(\bx_0,\btheta^1_{h,x}(\bz_0),\dots,\btheta^L_{h,x}(\bz_0)\right)\right)\right|e^{-\Delta\left(\bx_0,\btheta^1_{h,x}(\bz_0),\dots,\btheta^L_{h,x}(\bz_0)\right)} \Pi_*(\bz_{0})\dd \bz_0
%  \\&\quad=~
%  {\int}_{\R^{2d}} \left|f\left(\Delta\left(\btheta^L_{h,x}(\bz_0),\dots,\btheta^1_{h,x}(\bz_0),\bx_0\right)\right)\right|\Pi_*(\btheta^L_h(z_0))\dd \bz_0
%  \\&\quad=~
%   {\int}_{\R^{2d}}  \left|f\left(\Delta\left(\bu_0,\btheta^1_{h,x}(\bw_0),\dots,\btheta^L_{h,x}(\bw_L)\right)\right)\right|\Pi_*(\boldsymbol{w}_{0})\dd\bw_0
%  \\&\quad=~
%   \E[|f(\Delta(\bx_{0:L}))| ]
%   ~<~
%   \infty\,.
% \end{align*}
 Again, the same calculation without absolute values establishes the desired identity.
\end{proof}

\subsection{Proof of \Cref{prop:DeltaMALT}}\label{sec:DeltaMALTproof}

The center of the proof for both (i) and (ii) is showing that
\begin{equation}\label{eq:logAR_ConditionsAux1}
\frac{1}{h^2}\Delta_h(x_{0:L})
~=~
\int_0^TS(X_t,V_t) \dd t
~+~R_h\,,
\end{equation}
where the reminder $R_h$ satisfies $\E[R_h^2]\to 0$ as $h\to\infty$. This is done by approximating the leading terms of a Taylor's expansion of $\Delta_h(x_{0:L})$ over the entire Langevin trajectory and controlling the square expectations of the error terms. For this task we need to extend strong accuracy results (see \Cref{prop:StrongAccuracy}) beyond Lipschitz functions. This is enabled by the following lemma proven in \Cref{sec:proofTechnicalLemma}.

\begin{lemma}\label{lem:VarianceForm}
Let $T>0$, $\gamma\geq0$ and $h\in (0,1]$ and take $L$ so that $hL\leq T$. Assume the one-dimensional potential $\phi$ satisfies \Cref{Assumption:ProductForm} and let $x_0\sim\Pi$. There exists  constants $A_1,A_2,A_3,A_4>0$ independent of the choices of $h$ and $L$ such that following statements hold:
\begin{enumerate}
\item[(i)]
 \[
 \max_{0\leq i\leq L}\max( \E[x_{i}^8],\E[v_{i}^8])~\leq~A_1\,,
 \]

\item[(ii)]
    \[
    \E\left[\left(\phi'(x_{L})^2-\phi'(X_{h L})^2\right)^2\right]
    ~\leq~
    A_2h\,,
    \]
\item[(iii)]
    \[
    \sup_{1\leq i\leq L}\E\left[\left(\left(\frac{x_{i}-x_{i-1}}{h}\right)^3-\hat     V_{i-1}^3\right)^2\right]
    ~\leq~
    A_3h\,,
    \]
\item[(iv)]
    \[
    \sup_{1\leq i\leq L}\E\left[\left( V_{(i-1)h}^3-v_{i-1}^3\right)^2\right]
    ~\leq~
    A_4h^{2/3}\,.
\]
\end{enumerate}
\end{lemma}

To show \eqref{eq:logAR_ConditionsAux1} we transform the expression in \eqref{eq:TotalErrorCoordinate} into the desired form. The sum of third terms of \eqref{eq:LocalErrorCoordinate} terms appearing in  \eqref{eq:TotalErrorCoordinate} can be dealt with using It\^o's lemma
\[
\int_0^TV_t\phi''(X_t)\phi'(X_t)\dd t
~=~
\int_0^T\frac{\partial}{\partial x}\left(\frac{1}{2}\phi'(X_t)^2\right)\dd X_t
~=~
\frac{1}{2}\left(\phi'(X_T)^2-\phi'(X_0)^2\right).
\]
Since $x_{0}=X_0$, this implies that
\begin{align}\label{eq:logAR_ConditionsAux2}
\begin{split}
        \E\bigg[\Big(&\sum_{i=1}^L\big(\phi'(x_{i})^2-\phi'(x_{i-1})^2\big)
    ~-~
    2\int_0^TV_t\phi''(X_t)\phi'(X_t)\dd t\Big)^2\bigg]
    \\&=~
    \E\left[\left(\phi'(x_{L})^2-\phi'(x_{0})^2-\phi'(X_T)^2+\phi'(X_0)^2\right)^2\right]
    \\&\leq~
    2\E\left[\left(\phi'(x_{L})^2-\phi'(X_{hL})^2\right)^2\right]
    ~+~
    2\E\left[\left(\phi'(X_T)^2- \phi'(X_{hL})^2\right)^2\right]
    \quad\xrightarrow{h\to0}\quad
    0\,.
\end{split}
\end{align}
The first term vanishes by \Cref{lem:VarianceForm} (ii). The second vanishes because $hL\to T$, the function $\phi'$ is $M$-Lipschitz (grows at most linearly) and $X_t$ is a stationary process with continuous paths and finite fourth moments.

To control the rest of \eqref{eq:TotalErrorCoordinate} we take advantage of  \cite[Lemma~A.2.]{vogrinc2021counterexamples} and the fundamental theorem of calculus. For each $1\leq i\leq L$ we rewrite using $\int_0^1(1-2u)\dd u=0$
\begin{align*}
   \phi(x_{i})-\phi(x_{i-1})&-\frac12(x_{i}-x_{i-1}) \left(\phi'(x_{i})+\phi'(x_{i-1}\right)
    \\&=~
    \frac12(x_{i}-x_{i-1})^2
    \int_0^1(1-2u)\phi''(x_{i-1}+u(x_{i}-x_{i-1})) \dd u
    \\&=~
    \frac12(x_{i}-x_{i-1})^3
    \int_0^1(1-2u)u\int_0^1\phi^{(3)}(x_{i-1}+su(x_{i}-x_{i-1})) \dd s \dd u\,.
\end{align*}
Noting that $\int_0^1(1-2u)u \dd u=-\tfrac{1}{6}$, we can relate this to a term expressed with Langevin dynamics:
\begin{align}\label{eq:logAR_ConditionsAux3}
\begin{split}
        \phi(x_{i})-\phi(x_{i-1})&-\frac12(x_{i}-x_{i-1}) \left(\phi'(x_{i})+\phi'(x_{i-1}\right)+\frac{1}{12}h^3V_{(i-1)h}^3\phi^{(3)}(X_{(i-1)h})
    \\
    &=~
    R_{i,1}~+~R_{i,2}~+~R_{i,3}\,,
\end{split}
\end{align}
where 
\begin{align*}
    R_{i,1}
    &=
    \frac12\left((x_{i}-x_{i-1})^3-h^3v_{i-1}^3\right)
    \int_0^1(1-2u)u\int_0^1\phi^{(3)}(x_{i-1}+su(x_{i}-x_{i-1})) \dd s \dd u
    \\
    R_{i,2}
    &=
    \frac12h^3\left(v_{i-1}^3-V_{(i-1)h}^3\right)
    \int_0^1(1-2u)u\int_0^1\phi^{(3)}(x_{i-1}+su(x_{i}-x_{i-1})) \dd s \dd u
    \\
    R_{i,3}
    &=
    \frac{1}{2}h^3V_{(i-1)h}^3
    \int_0^1(1-2u)u\int_0^1\left(\phi^{(3)}(x_{i-1}+su(x_{i}-x_{i-1}))-\phi^{(3)}(X_{(i-1)h})\right) \dd s \dd u\,.
\end{align*}
We will show that of the $k=1,2,3$ we have
\begin{equation}\label{eq:logAR_ConditionsAux4}
    \frac{1}{h^6}\sup_{1\leq i\leq L}\E[R_{i,k}^2]
    \quad\xrightarrow{h\to0}\quad
    0\,.
\end{equation}
For the first two term $R_{i,1}$ and $R_{i,2}$ we use boundedness of $\phi^{(3)}$ together with respectively points (iii) and (iv) of \Cref{lem:VarianceForm}. For the last term $R_{i,3}$ we first use Cauchy's inequality (recall $\E[V_{(i-1)h}^{12}]$ is finite since $V_t$ is a stationary standard Gaussian process) together with Jensen's inequality, Fubini's theorem and the fact that both $\phi^{(3)}$ and $\phi^{(4)}$ are bounded.
\begin{multline*}
    \E\left[\left(\phi^{(3)}(x_{i-1}+su(x_{i}-x_{i-1}))-\phi^{(3)}(X_{(i-1)h})\right)^4\right]
    \\\leq~
    8\|\phi^{(4)}\|_\infty\|\phi^{(3)}\|_\infty^3 \left(
    \E\left[\left|x_{i-1}-X_{(i-1)h}\right|\right]
    +\E\left[\left|x_{i}-x_{i-1}\right|\right]
    \right)\,.
\end{multline*}
Then use \Cref{prop:StrongAccuracy} and \Cref{lem:VarianceForm} (iii) and (iv) to control this bound as follows:
\begin{align*}
    \E[|x_{i}-&x_{i-1}|]
   \leq
    \E\left[\left(x_{i}-x_{i-1}\right)^6\right]^{1/6}
    \\&=
    \E\left[\left(\left(x_{i}-x_{i-1}\right)^3-h^3v_{i-1}^3+h^3\left(v_{i-1}^3-V_{(i-1)h}^3\right)+h^3V_{(i-1)h}^3\right)^2\right]^{1/6}
    \\&\leq
    3h\left(\E\Big[\big(\left((x_{i}-x_{i-1})/h\right)^3-v_{i-1}^3\big)^2\Big]+\E\Big[\big(v_{i-1}^3-V_{(i-1)h}^3\big)^2\Big]+\E\left[V_{(i-1)h}^6\right]\right)^{1/6}
\end{align*}

Finally, paths of $(X_t,V_t)$ are almost surely continuous and therefore Riemann integrable. Hence almost surely
%\begin{multline}\label{eq:logAR_ConditionsAux5}
\[
    h\sum_{i=1}^{L}V_{(i-1)h}^3\phi^{(3)}(X_{(i-1)h})
    \quad\xrightarrow{h\to 0}\quad
    \int_0^TV^3_t\phi^{(3)}(X_t) \dd t\,.
\]
Since $hL\to T$, $\phi^{(3)}$ is bounded and $V_t$ is standard Gaussian for all $t$, the convergence also holds in $\mathbb{L}_2$-sense
\[
    \E\left[\left(h\sum_{i=1}^{L}V_{(i-1)h}^3\phi^{(3)}(X_{(i-1)h}) -    \int_0^TV^3_t\phi^{(3)}(X_t) \dd t\right)^2\right]
    \quad\xrightarrow{h\to 0}\quad 0\,.
\]
Combining this convergence with \eqref{eq:logAR_ConditionsAux3} and \eqref{eq:logAR_ConditionsAux4}, we obtain as $h\to0$
\[
\frac{1}{h^4}\E\bigg[\Big(\sum_{i=1}^{L}\big(\phi(x_{i})-\phi(x_{i-1})-\frac12(x_i-x_{i-1}) (\phi'(x_{i})+\phi'(x_{i-1}))\big)-\int_0^TV^3_t\phi^{(3)}(X_t) \dd t\Big)^2\bigg]
\to
0
\]
which, combined with \eqref{eq:logAR_ConditionsAux2}, yields \eqref{eq:logAR_ConditionsAux1}. Point (i) follows by integrating
\[
\frac{1}{h^4}\E\left[ \Delta^2_h(x_{0:L})\right]
~=~
\E\left[\left(\int_0^TS(X_t,V_t) \dd t\right)^2\right]-2\E\left[R_h\int_0^TS(X_t,V_t) \dd t\right]+\E\left[R_h^2\right]
\]
where the last two terms vanish as $h\to0$ (use Cauchy's inequality for the middle term).

Proof of point \Cref{prop:DeltaMALT}(ii) uses similar arguments. Recall that $V_t$ is a stationary Gaussian and that by \Cref{Assumption:ProductForm} there exists constants $M,D$ such that $\|\phi''\|_{\infty}\leq M$ and $|\phi'(x)|\leq M(|x|+D)$. Note that by \Cref{Assumption:ProductForm}, Jensen's inequality and Fubini's theorem
\begin{align*}
\E\left[\Big(\int_0^TS(X_t,V_t) \dd t\Big)^4\right]
&\leq
\int_0^T\E\left[S(X_t,V_t)^4\right] \dd t
\\&=
4\int_0^T\E\left[\frac{1}{12^4}V_t^{12}\phi^{(3)}(X_t)^4+\frac{1}{4^4}V_t^4\phi''(X_t)^4\phi'(X_t)^4\right] \dd t
\\&\leq
\frac{T}{3^4 4^3}\|\phi^{(3)}\|_\infty^4\E[V_t^{12}]+\frac{T}{4^3}M^4\E[V_t^8]^{1/2}\E\left[(M(|X_t|+D))^8\right]^{1/2}
\\&<\infty.
\end{align*}
This and \eqref{eq:logAR_ConditionsAux1} show
\begin{align*}
    \frac{1}{h^4}\E\left[ \Delta^2_h(x_{0:L})\mathds{1}_{\Delta_h(x_{0:L})>h}\right]
    &=
    \E\left[\Big(\int_0^TS(X_t,V_t) \dd t\Big)^2\mathds{1}_{\Delta_h(x_{0:L})>h}\right]+\E\left[R_h^2\mathds{1}_{\Delta_h(x_{0:L})>h}\right]
    \\&+
    2\E\left[R_h\left(\int_0^TS(X_t,V_t)\dd t\right)\mathds{1}_{\Delta_h(x_{0:L})>h}\right].
\end{align*}
The last two terms vanish since $\E[R_h^2]\to 0$. The first term vanishes by a combination of Cauchy and Markov inequalities together with part (i)
\begin{align*}
        \E\bigg[\Big(\int_0^TS(X_t,V_t) \dd t\Big)^2&\mathds{1}_{\Delta_h(x_{0:L})>h}\bigg]
    \leq
    \bigg(\E\bigg[\Big(\int_0^TS(X_t,V_t) \dd t\Big)^4\bigg]\PP\left( \Delta^2_h(x_{0:L})>h^2\right)\bigg)^{1/2}
    \\&~\leq
    h\,\E\left[\Big(\int_0^TS(X_t,V_t) \dd t\Big)^4\right]^{1/2}\left(\frac{1}{h^4}\E\left[ \Delta_h^2(x_{0:L})\right]\right)^{1/2}
    \xrightarrow{h\to0}~
    0\,.
\end{align*}

\subsection{Proof of \Cref{lem:VarianceForm}}\label{sec:proofTechnicalLemma}
A uniform control of the distance between the numerical Langevin trajectory and the Langevin diffusion is given by \Cref{prop:StrongAccuracy}, which extends trivially to functional values along the trajectories for all Lipschitz functions. Point (i) allows us to extend these findings for functions that are not globally Lipschitz but do not grow to rapidly. This is then done for specific functions required for the analysis of MALT in points (ii), (iii) and (iv).

The proof of (i) we use a "discrete Gr\"onwall's inequality". We will show that for $0\leq i\leq L-1$ and some constants $A,B>0$
\begin{equation}\label{eq:ProofLemmaTechnicalLemmaAux1}
    \max\left(\E[x_{i+1}^8],\E[v_{i+1}^8],B\right)
    ~\leq~
    (1+Ah)\times\max\left(\E[x_{i}^8],\E[v_{i}^8],B\right)\,.
\end{equation}
Iterating this inequality and using $h\leq T/L$ and the inequality $(1+t)\leq e^t$ gives us the following uniform bound, valid for all $0\leq i\leq L$:
\begin{align*}
\max\left(\E\left[x_{i}^8\right],\E\left[v_{i}^8\right],B\right)
~&\leq~
(1+Ah)^i\times\max\left(\E\left[x_{0}^8\right],\E\left[v_{0}^8\right],B\right)
\\&\leq~
(1+AT/L)^L\times\max\left(\E\left[x_{0}^8\right],\E\left[v_{0}^8\right],B\right)
\\&\leq~
e^{AT}\times\max\left(\E\left[x_{0}^8\right],\E\left[v_{0}^8\right],B\right)\,.
\end{align*}
This establishes (i) with $A_1=e^{AT}\max\left(\E\left[x_{0}^8\right],\E\left[v_{0}^8\right],B\right)$, which is finite by \Cref{Assumption:ProductForm}, as the Langevin trajectory is initiated in stationarity.

\Cref{Assumption:ProductForm} implies that $\phi'$ is $M$-Lipschitz and grows at most linearly: $|\phi'(x)|\leq M(|x|+D)$ for appropriate constants $M,D$, so that $(\phi'(x))^8\leq 8M^8 (|x|^8+D^8)$.

First, we deal with the bound on the eighth moment of the position. We use the Multinomial theorem to expand the expression of $x_{i+1}^8$ in the OBABO update \eqref{eq:OBABOupdate}. We then separate the zero-order term with respect to $h$ from the rest and use the H\"older's inequality to bound them. Note that the number of terms is fixed and finite and recall that we are assuming $h<1$. Using the fact that $\phi'$ grows at most linearly, we get for appropriate constants $A',B'>0$
\begin{align}\label{eq:VarianceFormAux1}
\begin{split}
     \E[x_{i+1}^8]
    &=
    \E\left[\left(x_{i}+he^{-\gamma h/2}v_{i}+h\sqrt{1-e^{-\gamma h}}\xi_i-\frac{h^2}{2}\phi'(x_{i})\right)^8\right]
    \\
    &\leq
    \E[x_{i}^8]+\sum_{\substack{j_1+j_2+j_3+j_4=8\\j_1<8}}h^{j_2+j_3+2j_4}\binom{8}{j_1,j_2,j_3,j_4}\E\left[|x_{i}|^{j_1}|v_{i}|^{j_2}|\xi_i|^{j_3}|\phi'(x_{i})|^{j_4}\right]
    \\
    &\leq
    \E[x_{i}^8]+(8!)h\sum_{\substack{j_1+j_2+j_3+j_4=8\\j_1<8}}\left(\E\left[x_{i}^8\right]^{j_1}\E\left[v_{i}^8\right]^{j_2}\E\left[\xi_i^8\right]^{j_3}\E\left[\phi'(x_{i})^8\right]^{j_4}\right)^{1/8}
    \\
    &\leq
    \E[x_{i}^8]+A'h\max\left(\E[x_{i}^8],\E[v_{i}^8],B'\right)\,.
\end{split}
\end{align}
We derive a similar identity for the update of the velocities where we use Multinomial theorem to expand the expression of $v_{i+1}^8$ in the OBABO update \eqref{eq:OBABOupdate}.
Note that since $\xi_i,\xi_{i+1}$ are standard Gaussian and independent of each other and of $v_{i}$, there exists a standard Gaussian $\xi_i'$ that is independent of $v_{i}$ and satisfies $e^{-\gamma h/2}\sqrt{1-e^{-\gamma h}}\xi_i+\sqrt{1-e^{-\gamma h}}\xi_{i+1}
=\sqrt{1-e^{-2\gamma h}}\xi'_i$. Also note that the inequality $1-e^{-t}\leq t$ implies that $\sqrt{1-e^{-2\gamma h}}<\sqrt{2 \gamma h}$. We separate the zero-order and half-order terms with respect to $h$ and bound the remaining terms using the Lipschitz property of $\phi'$ and \eqref{eq:VarianceFormAux1}. The half-order term vanishes since $\E[\xi'_i]=0$. For appropriate constants $A,A',B>0$ we get
\begin{align}\label{eq:VarianceFormAux2}
\begin{split}
    \E[v_{i+1}^8]
    ~&=~
    \E\left[\left(e^{-\gamma h}v_{i}-h\phi'(x_{i})-\frac{h}{2}\left(\phi'(x_{i+1})-\phi'(x_{i})\right)+\sqrt{1-e^{-2\gamma h}}\xi'_i\right)^8\right]\\
    &\leq~
    \E[v_{i}^8]~+~
    8e^{-7\gamma h}\sqrt{1-e^{-2\gamma h}}\E\left[v_{i}^7\right]\E\left[\xi'_i\right]\\
    &+~A'\sum_{\substack{j_1+j_2+j_3+j_4=8\\j_2+ j_3+\tfrac12j_4\geq 1}}h^{j_2+ j_3+\tfrac12j_4}\E\left[|v_{i}|^{j_1}| \phi'(x_{i})|^{j_2}\left|\phi'(x_{i+1})-\phi'(x_{i})\right|^{j_3}|\xi'_i|^{j_4}\right]\\
    &\leq~
    \E[v_{i}^8]~+~
    Ah\max\left(\E\left[x_{i}^8\right],\E\left[v_{i}^8\right],B\right)\,.
\end{split}
\end{align}
Together, \eqref{eq:VarianceFormAux1} and \eqref{eq:VarianceFormAux2} imply \eqref{eq:ProofLemmaTechnicalLemmaAux1} and establish (i).

We obtain (ii) from combining: (i), the Lipschitz property of $\phi'$ and Cauchy's inequality
\begin{align*}
    \E\left[\left(\phi'(x_{L})^2-\phi'(X_{hL})^2\right)^2\right]
    ~&=~
    \E\left[\left(\phi'(x_{L})-\phi'(X_{hL})\right)^2\left(\phi'(x_{L})+\phi'(X_{hL})\right)^2\right]
    \\&\leq~
    \E\left[\left|\phi'(x_{L})-\phi'(X_{hL})\right|\left(|\phi'(x_{L})|+|\phi'(X_{hL})|\right)^3\right]
    \\&\leq~
    M^4\times\E\left[\left|x_{L}-X_{hL}\right|\left(|x_{L}|+|X_{hL}|+2D\right)^3\right]
        \\&\leq~
    M^4\times\E\left[\left|x_{L}-X_{hL}\right|^2\right]^{\frac12}\times
    \E\left[\left(|x_{L}|+|X_{hL}|+2D\right)^6\right]^{\frac12}
     \\&\leq~A_2 h\,.
\end{align*}
The final bound follows for an appropriate constant by \Cref{prop:StrongAccuracy} and (i).

To show (iii)
consider the following bound for $a,b\in\R$
\begin{align*}
(a^3-b^3)^2
~&=~
(a-b)^2(a^2+ab+b^2)^2
\\&=~
(a-b)^2((a-b)^2+3(a-b)b+3b^2)^2
\\&\leq~
3(a-b)^6+27(a-b)^4b^2+27(a-b)^2b^4\,.
\end{align*}

Fix $1\leq i\leq L$.
Insert $a=h^{-1}(x_{i}-x_{i-1})$ and $b=v_{i-1}$ in the above bound together with H\"older's inequality gives us
\begin{multline*}
\E\left[\left(h^{-3}(x_{i}-x_{i-1})^3-v_{i-1}^3\right)^2\right]\,
\leq\, 27\,\E\left[\left(h^{-1}(x_{i}-x_{i-1})-v_{i-1}\right)^6\right]^{1/3}\E\left[v_{i-1}^6\right]^{2/3}
\\+~
27\,\E\left[\left(h^{-1}(x_{i}-x_{i-1})-v_{i-1}\right)^6\right]^{2/3}\E\left[v_{i-1}^6\right]^{1/3}
~+~3\,\E\left[\left(h^{-1}(x_{i}-x_{i-1})-v_{i-1}\right)^6\right]
\,.
\end{multline*}
Hence, it is sufficient to show that for an appropriate constant $A_3'>0$
\[
\E\left[\left(h^{-1}(x_{i}-x_{i-1})-v_{i-1}\right)^6\right]
~\leq~A_3'h^3\,.
\]
Using the linear growth of $\phi'$ and the inequality $1-e^{-t}\leq t$
\begin{align*}
\E\left[\left(h^{-1}(x_{i}-x_{i-1})-v_{i-1}\right)^6\right]
~&=~
\E\Big[\Big((e^{-\gamma h/2}-1)v_{i-1}+\sqrt{1-e^{-\gamma h}}\xi_0-\frac{h}{2}\phi'(x_{i-1})\Big)^6\Big]
\\&\leq~
6(\gamma h/2)^6\E\left[v_{i-1}^6\right]
~+~
6(\gamma h)^3\E\left[\xi_0^6\right]
+\frac{3h^6}{32}\left[\phi'(x_{i-1})^6\right]
\\&\leq~A'_3h^3\,.
\end{align*}

To prove (iv) consider the following bound for $a,b\in\R$
\begin{align*}
(a^3-b^3)^2
~&=~
(a-b)^2(a^2+ab+b^2)^2
\\&\leq~
(a-b)^2(|a|+|b|)^4
~\leq~
(a-b)^{2/3}(|a|+|b|)^{16/3}\,.
\end{align*}

Fix $1\leq i\leq L$.
Inserting $a=V_{(i-1)h}$ and $b=v_{i-1}$ in the above bound together with H\"older's inequality gives us
\[
\E\left[\left( V_{(i-1)h}^3-v_{i-1}^3\right)^2\right]
\leq
\E\left[\left( V_{(i-1)h}-v_{i-1}\right)^2\right]^{1/3}
\E\left[\left( |V_{(i-1)h}|+|v_{i-1}|\right)^8\right]^{2/3}\,.
\]
\Cref{prop:StrongAccuracy} implies that $\E\left[\left( V_{(i-1)h}-v_{i-1}\right)^2\right]\leq A'_4h^2$ for an appropriate constant. Point (i) together with the stationarity of process $V_t$ ensure the right factor is bounded.

\subsection{Proofs of miscellaneous results}
\begin{proof}[Proof of \Cref{prop:SingleCoordinateDisplacement}]
(i) This is a direct consequence of weak convergence in \Cref{thm:MALT_CLT} applied to the Lipschitz function $t\to1\wedge e^t$. We use Proposition~2.4 in \cite{Roberts::1997} to evaluate $\E[1\wedge e^Y]$ for a normal random variable $Y$ on $\R$.

(ii) By independence of the coordinates and the product structure \eqref{eq:ProductFormDelta}
\begin{align*}
\E\Big[\big(f(\bX^{n+1}(1))- &f(\bX^n(1))\big)^2\Big]
=~
\E\left[\left(f(\bx_L(1))- f(\bx_0(1))\right)^2 1\wedge e^{-\sum_{j=1}^d\Delta_{h,j}}\right]
\\&=~\E\left[\left(f(\bx_L(1))- f(\bx_0(1))\right)^2 \left(1\wedge e^{-\sum_{j=1}^d\Delta_{h,j}}-1\wedge e^{-\sum_{j=2}^d\Delta_{h,j}}\right)\right]\\
&+~\E\left[\left(f(\bx_L(1))- f(\bx_0(1))\right)^2\right]\E\left[ 1\wedge e^{-\sum_{j=2}^d\Delta_{h,j}}\right]
\end{align*}
We first show that the second term vanishes with increasing $d$. Using Cauchy's inequality and the fact that $t\mapsto1\wedge e^t$ is $1$-Lipschitz we argue
\begin{align*}
    \E\Big[\big(f(\bx_L(1))- &f(\bx_0(1))\big)^2 \Big(1\wedge e^{-\sum_{j=1}^d\Delta_{h,j}}-1\wedge e^{-\sum_{j=2}^d\Delta_{h,j}}\Big)\Big]
    \\&\leq~
     \E\left[\left(f(\bx_L(1))- f(\bx_0(1))\right)^4\right]^{1/2}\E\left[\Delta^2_{h,j}\right]^{1/2}
     \\&\leq~
     C^2\E\left[\left(\bx_L(1)- \bx_0(1)\right)^8\right]^{1/4}
     \E\left[\left(|\bx_L(1)|+ |\bx_0(1)|\right)^8\right]^{1/4}
     \E\left[\Delta^2_{h,j}\right]^{1/2}.
\end{align*}
The first two factors are finite by \Cref{lem:VarianceForm}(i) and the third converges to zero by \Cref{prop:DeltaMALT}.

Hence, as $\E\left[ 1\wedge e^{-\sum_{j=2}^d\Delta_{h,j}}\right]\to a(\ell)$ by (i), we only need to show 
\[
\E\left[\left(f(\bx_L(1))- f(\bx_0(1))\right)^2\right]
\quad\xrightarrow{d\to\infty}\quad
\E\left[\left(f(X_T)-f(X_0)\right)^2\right]\,.
\]
Note that for every dimension $d$ there exists a Langevin diffusion $\bZ_t$ linked to the numerical Langevin trajectory $\bz_{0:L}$ via \eqref{eq_synchronization_langevin_obabo}. This ensures $\bx_{0}=\bX_0$ and implies that
\begin{align*}
\E\left[\left(f(\bx_{L}(1))- f(\bx_{0}(1))\right)^2\right]
=&~
\E\left[\left(f(\bx_{L}(1))-f(\bX_T(1))+f(\bX_T(1)) -f(\bX_0(1))\right)^2\right]
\\=&~2\,\E\left[\left(f(X_T) -f(X_0)\right)\left(f(\bx_{L}(1))-f(\bX_T(1))\right)\right]
\\+&~\E\left[\left(f(X_T) -f(X_0)\right)^2\right]+\E\left[\left(f(\bx_{L}(1))-f(\bX_T(1))\right)^2\right]\,.
\end{align*}
We will show that the last term vanishes. By Cauchy's inequality that means, that the mixed term vanishes as well. By another Cauchy inequality and the bound $|A-B|\leq |A|+|B|$
\begin{align*}
\E\left[\left(f(\bx_{L}(1))-f(\bX_T(1))\right)^2\right]
&\leq
C^2\E\left[\left(\bx_{L}(1)-\bX_T(1)\right)^2
\left(\left|\bx_{L}(1)\right|+\left|\bX_T(1)\right|\right)^2\right]
\\&\leq
C^2\E\left[\left|\bx_{L}(1)-\bX_T(1)\right|
\left(\left|\bx_{L}(1)\right|+\left|\bX_T(1)\right|\right)^3\right]
\\&\leq
C^2\Big(\E\big[\left(\bx_{L}(1)-\bX_T(1)\right)^2\big]\E\big[
\left(\left|\bx_{L}(1)\right|+\left|\bX_T(1)\right|\right)^6\big]\Big)^{1/2}
\end{align*}
The first factor vanishes by \Cref{prop:StrongAccuracy} and the second is bounded by \Cref{lem:VarianceForm}.

(iii)
Again we synchronise the Langevin diffusion with the numerical Langevin trajectory via \eqref{eq_synchronization_langevin_obabo}, hence $\bX^n(1)=\bx_0(1)=X_0$. We will prove that
\[
\E\left[f(\bX^{n+1}(1))\right]
\quad\xrightarrow{d\to\infty}\quad
\E\left[f(X_T')\right]
\]
for an arbitrary bounded Lipschitz function $f$. By Portmanteau theorem this implies the result.

Note that since the function $t\mapsto 1\wedge e^t$ is $1$-Lipschitz, \Cref{prop:DeltaMALT} guarantees
\begin{multline*}
\PP\left(U\text{ is between } 1\wedge e^{-\sum_{j=1}^d\Delta_{h,j}}\text{ and }1\wedge e^{-\sum_{j=2}^d\Delta_{h,j}}\right)
\\\leq
\E\left[\left|1\wedge e^{-\sum_{j=1}^d\Delta_{h,j}}-1\wedge e^{-\sum_{j=2}^d\Delta_{h,j}}\right|\right]
\leq
\E\left[\left|\Delta_h\right|\right]
\quad\xrightarrow{d\to\infty}\quad
0\,.
\end{multline*}
Since $f$ is bounded this implies that
\[
\left|\E\left[f(\bx_L(1))\mathds{1}_{U\leq 1\wedge e^{-\sum_{j=1}^d\Delta_{h,j}}}\right]-\E\left[f(\bx_L(1))\mathds{1}_{U\leq 1\wedge e^{-\sum_{j=2}^d\Delta_{h,j}}}\right]\right|
\quad\xrightarrow{d\to\infty}\quad
0
\]
and
\[
\left|\E\left[f(\bx_0(1))\mathds{1}_{U> 1\wedge e^{-\sum_{j=1}^d\Delta_{h,j}}}\right]-\E\left[f(\bx_0(1))\mathds{1}_{U> 1\wedge e^{-\sum_{j=2}^d\Delta_{h,j}}}\right]\right|
\quad\xrightarrow{d\to\infty}\quad0\,.
\]
Finally, this gives us
\begin{align*}
    &\lim_{d\to\infty}\E\left[f(\bX^{n+1}(1))\right]
    \\&=
    \lim_{d\to\infty}\left(\E\left[f(\bx_L(1))\mathds{1}_{U\leq 1\wedge e^{-\sum_{j=1}^d\Delta_{h,j}}}\right]
    +\E\left[f(\bx_0(1))\mathds{1}_{U> 1\wedge e^{-\sum_{j=1}^d\Delta_{h,j}}}\right]\right)
    \\&=
    \lim_{d\to\infty}\left(\E\left[f(\bx_L(1))\mathds{1}_{U\leq 1\wedge e^{-\sum_{j=2}^d\Delta_{h,j}}}\right]
    +\E\left[f(\bx_0(1))\mathds{1}_{U> 1\wedge e^{-\sum_{j=2}^d\Delta_{h,j}}}\right]\right)
    \\&=
    \lim_{d\to\infty}\left(\E\left[f(\bx_L(1))\right]\PP\left[U\leq 1\wedge e^{-\sum_{j=2}^d\Delta_{h,j}}\right]
    +\E\left[f(\bx_0(1))\right]\PP\left[U> 1\wedge e^{-\sum_{j=2}^d\Delta_{h,j}}\right]\right)
    \\&=
    \E\left[f(X_T)\right]a(\ell)+\E\left[f(X_0)\right](1-a(\ell))
    ~=~
     \E\left[f(X_T')\right]
    \,.
\end{align*}
The first equality holds by construction of the accept-reject test, the third by the product form assumption and the fourth by \Cref{prop:StrongAccuracy} and point (i).
\end{proof}

\begin{proof}[Proof of \Cref{prop:ScalingOptimality}]
The proof relies on a similar argument as in the proof of \cite[Theorem~3]{vogrinc2022optimal}. We first establish that
\begin{align}\label{eq:ScalingOptimalityAux1}
\begin{split}
    \E\left[\left(f(\bX^{n+1}(1))-f(\bX^{n}(1))\right)^2\right]
~&=~
\E\left[\left(f(\bx_L(1))-f(\bx_0(1))\right)^21\wedge e^{-\sum_{j=1}^d\Delta_{h,j}}\right]\\&\leq~
2\E\left[\left(f(\bx_L(1))-f(\bx_0(1))\right)^2\right]\E\left[1\wedge e^{-\sum_{j=2}^d\Delta_{h,j}}\right]\,.
\end{split}
\end{align}
The statement without the factor of two is apparent if $\Delta_{h,1}$ is positive. If it is negative we use \Cref{prop:log-AR_properties} and it is crucial that the expression $(f(\bx_L(1))-f(\bx_0(1))^2$ is symmetric with respect to the start $\bx_0(1)$ and end $\bx_L(1)$ point of the trajectory:
\begin{align*}
    &\E\left[\left(f(\bx_L(1))-f(\bx_0(1))\right)^21\wedge e^{-\sum_{j=1}^d\Delta_{h,j}}\mathds{1}_{(-\infty,0]}(\Delta_{h,1})\right]
    \\&\quad\leq~
    \E\left[\left(f(\bx_L(1))-f(\bx_0(1))\right)^21\wedge e^{-\sum_{j=2}^d\Delta_{h,j}}e^{-\Delta_{h,1}}\mathds{1}_{[0,\infty)}(-\Delta_{h,1})\right]
    \\&\quad=~
    \E\left[\left(f(\bx_0(1))-f(\bx_L(1))\right)^21\wedge e^{-\sum_{j=2}^d\Delta_{h,j}}\mathds{1}_{[0,\infty)}(\Delta_{h,1})\right]
    \\&\quad\leq~
    \E\left[\left(f(\bx_L(1))-f(\bx_0(1))\right)^21\wedge e^{-\sum_{j=2}^d\Delta_{h,j}}\right]\,.
\end{align*}
As $h\to0$ we have $\E[\left(f(\bx_L(1))-f(\bx_0(1))\right)^2]\to\Upsilon_f$; see proof of \Cref{prop:SingleCoordinateDisplacement}(ii). Moreover, if $d^{1/4}h\to0$, then $d^{1/4}h\,\E[1\wedge \exp\{-\sum_{j=2}^d\Delta_{h,j}\}]\to 0$ as $d$ increases, since the acceptance rate is bounded by one. Assume now $d^{1/4}h\to\infty$. By \cite[Theorem~7]{vogrinc2021counterexamples}, \Cref{prop:DeltaMALT} implies that
\[
\frac{\E\left[\Delta_h\right]}{\E\left[\Delta_h^2\right]}
~\xrightarrow{h\to0}~
\frac{1}{2}\,.
\]
We split the probability space with respect to the event
$\mathcal{A}_d=\left\{\sum_{j=2}^d\Delta_{h,j}\geq \tfrac{d-1}{2}\E[\Delta_h]\right\}$.
Using the convergence above and \Cref{prop:DeltaMALT} with $dh^4\to\infty$ we obtain:
\begin{multline*}
    \limsup_{d\to\infty}d^{1/4}h\E\left[1\wedge e^{-\sum_{j=2}^d\Delta_{h,j}}\mathds{1}_{\mathcal{A}_d}\right]
    ~\leq~
    \limsup_{d\to\infty}d^{1/4}h e^{-\frac{d-1}{2}\E[\Delta_h]}
    \\\leq~
    \limsup_{d\to\infty}d^{1/4}h e^{-\frac{d}{5}\E[\Delta^2_h]}
    =~
    \limsup_{d\to\infty}d^{1/4}h e^{-\frac{dh^4}{5}h^{-4}\E[\Delta^2_h]}
    ~\leq~
    \limsup_{d\to\infty}d^{1/4}h e^{-dh^4\frac{\Sigma}{6}}
    ~=~0\,.
\end{multline*}
On the complement $\mathcal{A}_d^c$ we use Chebyshev's inequality in addition 
\begin{align*}
        \limsup_{d\to\infty}&\,d^{1/4}h\,\E\left[1\wedge e^{-\sum_{j=2}^d\Delta_{h,j}}\mathds{1}_{\mathcal{A}_d^c}\right]
    \leq
    \limsup_{d\to\infty}\,d^{1/4}h\,\PP(\mathcal{A}_d^c)
    \\&=
    \limsup_{d\to\infty}\,d^{1/4}h\, \PP\bigg(\sum_{j=1}^d\left(\E[\Delta_{h,j}]-\Delta_{h,j}\right)\geq \frac{d-1}{2}\E[\Delta_h]\bigg)
    \\&\leq
    \limsup_{d\to\infty}\,d^{1/4}h\,\frac{4d\text{Var}[\Delta_h]}{d^2\E[\Delta_h]^2}
    \leq
    \limsup_{d\to\infty}
    \,d^{1/4}h\,\frac{17}{d\E[\Delta_h^2]}
    \leq
    \limsup_{d\to\infty}
    \frac{18}{\Sigma}\frac{1}{d^{3/4}h^3}
    =
    0\,.
\end{align*}
The last two conclusions together with \eqref{eq:ScalingOptimalityAux1} show the sub-optimality of scaling of time-step $h$ different to $d^{-1/4}$. We conclude that $d^{-1/4}$ scaling is optimal by noting that the claimed non-zero limit is achieved in the case when $d^{1/4}h\to\ell$ by \Cref{prop:SingleCoordinateDisplacement}(ii). 

Finally, to optimize over the choice of $\ell$ note that the function $\ell\mapsto 2\ell \Psi\left(-\tfrac12\ell^2\sqrt{\Sigma}\right)$ is smooth and converges to zero both as $\ell\to0$ and $\ell\to\infty$. Since it is positive, its maximum must be attained at a stationary point. Substituting $s=\ell^2\sqrt{\Sigma}/2$ we can find stationary points of $s\mapsto 2^{3/2} \Sigma^{-1/4}\sqrt{s}\Psi(-s)$ which is equivalent to finding solutions of $\frac{1}{2}=s\frac{\psi(-s)}{\Psi(-s)}$, where $\psi$ is the density of the standard Gaussian. There exists a unique solution $s^*$ since the function $s\mapsto s\frac{\psi(-s)}{\Psi(-s)}$ is strictly increasing. Hence, ${\rm eff}(\ell)$ attains its greatest value at a specific value $\ell^*=\sqrt{2s^*}\Sigma^{-1/4}$ that corresponds to an average acceptance rate $2\Psi(-s^*)$, which turns out to numerically equal $0.651$ to three decimal places. It also implies that
\[
{\rm eff}(\ell^*)
~=~
 2^{3/2}\sqrt{s^*}\Psi(-s^*)\times \Sigma^{-1/4}
~\approx~
0.619219\times \Sigma^{-1/4}\,.
\]
\end{proof}

\section{Proofs of Section 4}

\subsection{Proof of \Cref{thm:randomized_langevin_contraction}}\label{proof_contraction}
The proof relies on establishing a contraction rate for a twisted Wasserstein distance between two copies of Randomized HMC. The coupling is synchronized with the same homogeneous Poisson process $(\bN_t)_{t\ge0}$ with rate $\lambda$, and the same independent standard Gaussian refreshments $(\bxi_k)_{k\in \N}$.
Let us denote for any $t\ge 0$ the two processes $\bZ_t\triangleq(\bX_t,\bV_t)$ and $\bZ_t'\triangleq(\bX_t',\bV_t')$, defined through the following system of SDEs:
\begin{align}\label{eq_sdes_coupling}
\begin{split}
        \dd \bV_t&=-\nabla\Phi(\bX_t)\dd t+\left(\sqrt{1-\alpha^2}\bxi_{\bN_t}-(1-\alpha)\bV_t\right)\dd \bN_t,\qquad
    \dd \bX_t=\bV_t\dd t,\\
    \dd \bV_t'&=-\nabla\Phi(\bX_t')\dd t+\left(\sqrt{1-\alpha^2}\bxi_{\bN_t}-(1-\alpha)\bV_t'\right)\dd \bN_t,\qquad
    \dd \bX_t'=\bV_t'\dd t.
\end{split}
\end{align}
Let $\nu$ and $\nu^{\prime}$ be any two probability measures on $\R^{2d}$. We initialize the two processes form an arbitrary coupling $(\bZ_0,\bZ_0')\sim\zeta$ such that $\bZ_0\sim\nu$, $\bZ_0'\sim\nu^{\prime}$. Both copies are Randomized HMC processes, yet we remark that the joint process $(\bZ_t,\bZ_t')$ is also a Markov process. We denote $\mathcal{L}_{\lambda,\alpha}^{\rm Couple}$ its joint infinitesimal generator, characterized by the system of SDEs \eqref{eq_sdes_coupling} for $j\in\{1,2\}$.

The main arguments of this proof rely on a uniform bound on the generator $\mathcal{L}_{\lambda,\alpha}^{\rm Couple}$ applied to a twisted distance between the coupled processes, established in \Cref{lem_generator_bound}. This bound is proven in \Cref{proof_generator_bound}. The twist of the metric is determined by three real numbers $a,b,c$ such that
\begin{equation}\label{eq_A_pos_def}
    \bfA\triangleq\begin{pmatrix}a\bfI_d&b\bfI_d\\
    b\bfI_d& c\bfI_d\end{pmatrix}\succ\bfzero_{d\times d}.
\end{equation}
For any such positive definite matrix $\bfA$ and any vector $\bz\in\R^{2d}$, we denote $|\bz|_{\bfA}\triangleq(\bz^\top \bfA\bz)^{1/2}$ its $\bfA$-norm. For any $q\ge1$, we define the $(q,\bfA)$-Wasserstein metric between any two probability measures $\nu$ and $\nu'$ on $\R^{2d}$ by
$$ W_{q,\bfA}(\nu,\nu')\triangleq \inf \{\E[|\bX-\bX'|_{\bfA}^q]^{1/q}, \bX\sim\nu, \bX'\sim\nu'\}.$$
\begin{lemma}\label{lem_generator_bound}
The values $a=2(M+m)/(1+\alpha)$, $b=\sqrt{M+m}$ and $c=2$ are such that the matrix $\bfA$ defined in \eqref{eq_A_pos_def} is positive definite. Moreover almost surely for $t\ge0$, we have
\begin{equation}\label{eq_gen_bound}
    \mathcal{L}_{\lambda,\alpha}^{\rm Couple}|\bZ_t-\bZ_t'|_{\bfA}^2\le-2r |\bZ_t-\bZ_t'|_{\bfA}^2,\qquad r=\frac{(1+\alpha)m}{2\sqrt{M+m}}.
\end{equation}
\end{lemma}
 An application of Gr\"onwall's inequality to \Cref{lem_generator_bound} yields contraction for the twisted $\mathbb{L}_2$-norm. We obtain for any $t\ge 0$
  \begin{equation}\label{eq_l2_contraction}
\E[|\bZ_t-\bZ_t'|_{\bfA}^2]\le e^{-2r t}\E[|\bZ_0-\bZ_0'|_{\bfA}^2].
\end{equation}
Taking square roots and considering the infimum over the couplings $(\bZ_0,\bZ_0')\sim\zeta$ such that $\bZ_0\sim\nu$, $\bZ_0'\sim\nu'$ yields 
\begin{equation}\label{eq_wass_contraction}
W_{2,\bfA}(\nu \bfP^t,\nu^{\prime} \bfP^t)\le e^{-rt}W_{2,\bfA}(\nu,\nu^{\prime})
\end{equation}
Establishing the $2$-Wasserstein convergence for the Euclidean distance follows from unfolding the twist of the metric. We consider the inequality \eqref{eq_l2_contraction} for a particular coupling $(\bZ_0,\bZ_0')\sim\zeta$ with marginals $\nu=\nu_{\bx}\otimes\mathcal{N}_d(\bfzero_d,\bfI_d)$ and $\nu'=\Pi_*=\Pi\otimes\mathcal{N}_d(\bfzero_d,\bfI_d)$ defined as follows. We start from $\bV_0=\bV_0'\sim\mathcal{N}_d(\bfzero_d,\bfI_d)$, independently drawn from an arbitrary coupling of the positions $(\bX_0,\bX_0')\sim\zeta_{\bx}$ such that $\bX_0\sim\nu_{\bx}^{}$ and $\bX_0'\sim\Pi$. 
We remark that for any $\bx,\bv\in\R^d$ we have
\begin{equation*}
        c(a|\bx|^2+2b\bx^\top \bv + c |\bv|^2)=(ac-b^2)|\bx|^2+|b\bx+cv|^2\ge (ac-b^2)|\bx|^2.
\end{equation*}
Combining this inequality with \eqref{eq_l2_contraction} yields
\begin{equation*}
\E\left[\left|\bX_t-\bX_t'\right|^2\right]\le \frac{c}{ac-b^2}\, \E\left[\left|\bZ_t-\bZ_t'\right|_{\bfA}^2\right]\le \frac{ce^{-2rt}}{ac-b^2}\, \E\left[\left|\bZ_0-\bZ_0'\right|_{\bfA}^2\right].
\end{equation*}
Using that $\bV_0-\bV_0'=0$ the last term can be substituted by
\begin{equation*}
\E\left[\left|\bZ_0-\bZ_0'\right|_{\bfA}^2\right]=a\, \E\left[\left|\bX_0-\bX_0'\right|^2\right].
\end{equation*}
Taking square roots and considering the infimum over the couplings $(\bX_0,\bX_0')\sim\zeta_{\bx}$ such that $\bX_0\sim\nu_{\bx}$ and $\bX_0'\sim\Pi$ yields
\begin{equation*}
W_2((\nu \bfP^t)_{\bx},\Pi)\le Ce^{-r t} W_2(\nu_{\bx},\Pi),\qquad C\triangleq\sqrt{\frac{ac}{ac-b^2}}.
\end{equation*}
In particular, $(\Pi_*\bfP^t)_{\bx}=\Pi$ for any $t>0$ because $\Pi_*$ is invariant. A direct computation of the constant $C$ yields
\begin{equation}\label{eq_C}
    C=\sqrt{4/(3-\alpha)}\le \sqrt{2}.
\end{equation}

We now extend the convergence to the $\mathbb{L}_2(\Pi_*)$-norm using a similar sketch of proof as for \cite[Theorem 3]{deligiannidis2021randomized}. We introduce $(\bfP^t)^*$ and $(\mathcal{L}_{\lambda,\alpha}^{\rm Couple})^*$ the respective adjoints of $\bfP^t$ and $\mathcal{L}_{\lambda,\alpha}^{\rm Couple}$. These adjoints are characterized by the distribution of Randomized HMC ran backwards in time. In a weak sense, the backward dynamics are similar to the forward dynamics up to a flip of the drift term.
%The time reversibility of Hamiltonian dynamics yields tractable SDEs for the backward process, which are similar to the forward process up to a flip of the velocity. We consider the same synchronous coupling defined through the system of SDEs \eqref{eq_sdes_coupling} for $j\in\{1,2\}$ except $\bV_t^{(j)}$ is replaced by $-\bV_t^{(j)}$. We note $(\mathcal{L}_{\lambda,\alpha}^{\rm Couple})^*$ the generator of this coupling and introduce an analogue twist of the metric determined by 
In \Cref{lem_adjointgenerator_bound} we present a uniform bound on the adjoint generator, similar to the one obtained in \Cref{lem_generator_bound} for the forward process. The twist of the metric differs to take into account the change of sign in the velocities. It is determined by 
\begin{equation}\label{eq_Aprime_pos_def}
    \bfA^{\prime}\triangleq\begin{pmatrix}a\bfI_d&-b\bfI_d\\
    -b\bfI_d& c\bfI_d\end{pmatrix}\succ\bfzero_{d\times d}.
\end{equation}
\begin{lemma}\label{lem_adjointgenerator_bound}
The values $a=2(M+m)/(1+\alpha)$, $b=\sqrt{M+m}$ and $c=2$ are such that the matrix $\bfA'$ defined in \eqref{eq_Aprime_pos_def} is positive definite. Moreover almost surely for $t\ge0$, we have
\begin{equation}\label{eq_adjointgen_bound}
    (\mathcal{L}_{\lambda,\alpha}^{\rm Couple})^*|\bZ_t-\bZ_t'|_{\bfA'}^2\le-2r |\bZ_t-\bZ_t'|_{\bfA'}^2,\qquad r=\frac{(1+\alpha)m}{2\sqrt{M+m}}.
\end{equation}
\end{lemma}
The result of \Cref{lem_adjointgenerator_bound} can be derived following the same arguments used for \cite[Theorem 3]{deligiannidis2021randomized}, its proof is therefore omitted. Similarly to \eqref{eq_l2_contraction} and \eqref{eq_wass_contraction}, this result yields a $(2,\bfA')$-Wasserstein contraction for the adjoint $(\bfP^t)^*$. For any $\nu,\nu'$ two probability measures of $\R^{2d}$
\begin{equation}\label{eq_adjointwass_contraction}
W_{2,\bfA'}(\nu (\bfP^t)^*,\nu^{\prime} (\bfP^t)^*)\le e^{-rt}W_{2,\bfA'}(\nu,\nu^{\prime}).
\end{equation}
For $f:\R^{2d}\rightarrow \R$ we introduce the Lipschitz norm with respect to $\bfA$ as
\begin{equation*}
    \|f\|_{{\rm Lip},\bfA}\triangleq\underset{\bz_1\neq\bz_2}{\sup}\frac{|f(\bz_1)-f(\bz_2)|}{|\bz_1-\bz_2|_{\bfA}}.
\end{equation*}
Using the dual formulation of the $(1,\bfA)$-Wasserstein distance, for any measures $\nu,\nu'$ on $\R^{2d}$
\begin{equation*}
    W_{2,\bfA}(\nu ,\nu^{\prime})\ge W_{1,\bfA}(\nu ,\nu^{\prime})=\underset{\|f\|_{{\rm Lip},\bfA}\le 1}{\sup}\int f \, \dd(\nu-\nu').
\end{equation*}
Let $\bz_1,\bz_2\in\R^{2d}$. We combine this dual formulation with \eqref{eq_wass_contraction} and \eqref{eq_adjointwass_contraction} for $\nu=\delta_{\bz_1}$ and $\nu'=\delta_{\bz_2}$. In particular for any $f,g$ such that $\|f\|_{{\rm Lip},\bfA}\le 1$ and $\|g\|_{{\rm Lip},\bfA'}\le 1$ we obtain
\begin{align}\label{eq_ricci}
\begin{split}
             |\bfP^t f(\bz_1)-\bfP^t f(\bz_2)|\le  & e^{-rt}|\bz_1-\bz_2|_{\bfA}  \\
     |(\bfP^t)^* g(\bz_1)-(\bfP^t)^*g(\bz_2)|\le  & e^{-rt}|\bz_1-\bz_2|_{\bfA'}
\end{split}
\end{align}
We also note that \cite[Equation (4.12)]{deligiannidis2021randomized} yields
\begin{equation}\label{eq_twist_relation}
    |\bz_1-\bz_2|_{\bfA'}\le (C')^2|\bz_1-\bz_2|_{\bfA},\qquad C'\triangleq\left(\frac{ac+b^2+2\sqrt{acb^2}}{ac-b^2}\right)^{1/4}.
\end{equation}
Let $f\in\mathbb{L}_2^0(\Pi)$. By definition, $f$ depends only on the position therefore $\|f\|_{{\rm Lip},\bfA}=\|f\|_{{\rm Lip},\bfA'}$. Applying successively \eqref{eq_ricci} with $g=\bfP^t f$ and \eqref{eq_twist_relation} yields
\begin{align}
    \begin{split}
        \|(\bfP^t)^*\bfP^t f\|_{{\rm Lip},\bfA'}\le & e^{-rt}\|\bfP^t f\|_{{\rm Lip},\bfA'}\le (C')^2e^{-rt}\|\bfP^t f\|_{{\rm Lip},\bfA}\\
        \le& (C')^2e^{-2rt}\|f\|_{{\rm Lip},\bfA}=(C')^2e^{-2rt}\|f\|_{{\rm Lip},\bfA'}.
    \end{split}
\end{align}
We refer to the proof of \cite[Proposition 30]{ollivier2009ricci} with $\kappa=1-(C')^2e^{-2rt}$ for any $t>\log(C')/r$. We argue that $(\bfP^t)^*\bfP^t$ is a reversible kernel with spectral radius at most $(C')^2e^{-2rt}$ on the Lipschitz functions of $\mathbb{L}_2^0(\Pi)$. This subset being dense, the spectral radius extends to every function of $\mathbb{L}_2^0(\Pi)$. Noting $\langle,\rangle$ for the scalar product in $\mathbb{L}_2^0(\Pi_*)$, we obtain that for any $f\in\mathbb{L}_2^0(\Pi)$
\begin{equation}
    \|\bfP^t f\|^2=\langle f,(\bfP^t)^*\bfP^t f\rangle\le\|f\|\|(\bfP^t)^*\bfP^t f\|\le (C')^2e^{-2rt}\|f\|^2.
\end{equation}
The second claim of the Theorem follows from taking square roots. A direct computation of $C'$ yields
\begin{equation}\label{eq_C'}
    C'=\left(\frac{5+\alpha+4\sqrt{1+\alpha}}{3-\alpha}\right)^{1/4}\le (3+2\sqrt{2})^{1/4}\le 1.56.
\end{equation}

 %=\frac{ac}{ac-b^2}e^{-\mu t}W(\nu_X,\nu_X^{\prime})
%$\{a |\Tilde{X}_t|^2+2b\Tilde{X}_t^\top\Tilde{V}_t+c|\Tilde{V}_t|^2\}$

\subsection{Proof of \Cref{lem_generator_bound}}\label{proof_generator_bound}
 In the sequel, we denote $\bfH_t$ the $d\times d$ matrix defined by
$$
\bfH_t\triangleq\int_0^1\nabla^2\Phi(s\bX_t+(1-s)\bX_t')\dd s.
$$
Taylor's theorem yields $\nabla\Phi(\bX_t)-\nabla\Phi(\bX_t')=\bfH_t \Tilde{\bX}_t$, and $m\bfI_d\preceq \bfH_t \preceq M\bfI_d$ by \Cref{assumption:strong_cvx}. For the purpose of computations, we make repeated use of the shorthand notations $\Tilde{\bX}_t\triangleq \bX_t-\bX_t'$, and $\Tilde{\bV}_t\triangleq \bV_t-\bV_t'$ in the following.  From \eqref{eq_sdes_coupling} we rewrite $(\Tilde{\bX}_t,\Tilde{\bV}_t)_{t\ge0}$ as a stochastic jump process driven by the following SDE
\begin{align*}
    \dd \Tilde{\bV}_t&=-\bfH_t\Tilde{\bX}_t\dd t-(1-\alpha)\Tilde{\bV}_t\dd \bN_t\\
    \dd \Tilde{\bX}_t&=\Tilde{\bV}_t\dd t
\end{align*}
Applying the product rule to this SDE yields
\begin{align*}
    \dd |\Tilde{\bX}_t|^2&=2\Tilde{\bX}_t^\top\Tilde{\bV}_t\dd t\\
    \dd \Tilde{\bX}_t^\top\Tilde{\bV}_t&=(-\Tilde{\bX}_t^\top \bfH_t\Tilde{\bX}_t+|\Tilde{\bV}_t|^2)\dd t-(1-\alpha)\Tilde{\bX}_t^\top \Tilde{\bV}_t\dd \bN_t\\
    \dd |\Tilde{\bV}_t|^2&=-2\Tilde{\bX}_t^\top \bfH_t\Tilde{\bV}_t\dd t-(1-\alpha^2)|\Tilde{\bV}_t|^2\dd \bN_t.
\end{align*}
Letting $a,c>0$ and $b\in\R$ such that $ac-b^2>0$, we consider the positive definite matrix
$$
\bfA\triangleq\begin{pmatrix}a\bfI_d&b\bfI_d\\
    b\bfI_d& c\bfI_d\end{pmatrix}.
$$
Applying the generator $\mathcal{L}_{\lambda,\alpha}^{\rm Couple}$ on $|\bZ_t-\bZ_t'|_{\bfA}^2=a|\Tilde{\bX}_t|^2+2b\Tilde{\bX}_t^\top \Tilde{\bV}_t + c |\Tilde{\bV}_t|^2$ yields
\begin{align*}
    \mathcal{L}_{\lambda,\alpha}^{\rm Couple}|\bZ_t-\bZ_t'|_{\bfA}^2=&\, 2a\Tilde{\bX}_t^\top\Tilde{\bV}_t-2b\Tilde{\bX}_t^\top \bfH_t\Tilde{\bX}_t-2b(1-\alpha)\lambda \Tilde{\bX}_t^\top \Tilde{\bV}_t+2b|\Tilde{\bV}_t|^2\\&-2c\Tilde{\bX}_t^\top \bfH_t\Tilde{\bV}_t-c(1-\alpha^2)\lambda|\Tilde{\bV}_t|^2\\
    =&-(\bZ_t-\bZ_t')^\top\,\bfS_t\,(\bZ_t-\bZ_t')
\end{align*}
where
$$
\bfS_t\triangleq\begin{pmatrix}2b\bfH_t&(b(1-\alpha)\lambda-a)\bfI_d+c\bfH_t\\
    (b(1-\alpha)\lambda-a)\bfI_d+c\bfH_t& (c(1-\alpha^2)\lambda-2b)\bfI_d\end{pmatrix}.
$$
The inequality $\mathcal{L}_{\lambda,\alpha}^{\rm Couple}|\bZ_t-\bZ_t'|_{\bfA}^2\le-2r |\bZ_t-\bZ_t'|_{\bfA}^2$ is therefore satisfied whenever $\bfS_t-2r \bfA$ is a positive semi-definite matrix.
By diagonalizing each block of this last matrix in the same basis, we obtain necessary and sufficient conditions for ensuring that its $2d$ eigen-values are simultaneously non-negative: $\bfS_t-2r \bfA\succeq0$ iff for each eigen-value $\sigma$ of $\bfH_t$, we have
\begin{align}
    -2r a+2b\sigma&\ge0 \label{eigen:eq1}\\
    -2r c+c(1-\alpha^2)\lambda-2b&\ge 0 \label{eigen:eq2}\\
    [-2r b+ b (1-\alpha)\lambda-a+ch]^2&\le [-2r a+2b\sigma][-2r c+c(1-\alpha^2)\lambda-2b]\label{eigen:eq3}
\end{align}
The inequalities \eqref{eigen:eq1}, \eqref{eigen:eq2} and \eqref{eigen:eq3} ensure that the two diagonal elements and the determinant of the $2\times 2$ submatrix corresponding to a given eigen-value $\sigma$ (composed by the corresponding diagonal elements of each block), are non-negative. To solve these inequalities, we first choose $a,b,c$ such that the constraints \eqref{eigen:eq1} and \eqref{eigen:eq3} are saturated for the less favorable choice of eigen-value $\sigma$, then we check that \eqref{eigen:eq2} holds for the resulting values of $a,b,c$.

Without loss of generality we fix $c=2$, and then choose $a$ such that \eqref{eigen:eq1} is an equality when $\sigma=m$. Consequently, $b$ is also fixed since both sides of \eqref{eigen:eq3} must be zero when $\sigma=m$, we get
\begin{equation}\label{eq:choice_ab}
    \left\{\begin{matrix*}[l]2r a= 2b m\\
2r b=b(1-\alpha)\lambda-a+2m\end{matrix*}\right.\qquad \Leftrightarrow\qquad\left\{\begin{matrix*}[l]a=2m/s\\
b=2r/s\end{matrix*}\right.,\qquad s\triangleq 1- \frac{r}{m}((1-\alpha)\lambda-2r).
\end{equation}
For any real numbers $x,y,z,w$ such that $xz>0$ and $w/z=y/x=m$ the map $\sigma\mapsto(x\sigma-y)^2/(z\sigma-w)$ is increasing when $\sigma>m$. As a consequence, \eqref{eigen:eq3} holds for any $\sigma\in [m,M]$ iff it holds for $\sigma=M$. This condition is satisfied whenever
\begin{align*}
    [2(M-m)]^2&\le[2b(M-m)][2((1-\alpha^2)\lambda-2r)-2b]\\
    \Leftrightarrow\qquad\qquad M-m&\le b((1-\alpha^2)\lambda-2r-b)\\
    \Leftrightarrow\qquad\qquad M+m&\le a+\alpha(1-\alpha)\lambda b -b^2
\end{align*}
where the last equivalence follows from substituting $b((1-\alpha)\lambda-2r)$ by $a-2m$, in line with \eqref{eq:choice_ab}. By multiplying both sides by $s^2$, we conclude that \eqref{eigen:eq3} holds for any $h\in [m,M]$ iff
\begin{equation}\label{eq:2nd_degree_s}
    (M+m)s^2-(2m+2\alpha(1-\alpha)\lambda r)s+4 r^2\le0.
\end{equation}
Moreover, the set of solutions $s\in \R$ to \eqref{eq:2nd_degree_s} is non-empty only when 
\begin{equation}\label{eq:solutions_existence}
    16 r^2(M+m)-(2m+2\alpha(1-\alpha)\lambda r)^2\le 0.
\end{equation}
For the purpose of computations, we parametrize the refreshment angle, intensity, and the exponential convergence rate with respect to $x, y\ge0$ such that
$$
\alpha(1-\alpha)\lambda=x\sqrt{M+m}, \qquad r= y\left(\frac{m}{\sqrt{M+m}}\right).
$$
Plugging this parametrization into \eqref{eq:solutions_existence} simplifies to
\begin{align*}
    16 y^2m^2-(2m+2 y x m)^2\le0\qquad&\Leftrightarrow\qquad16 y^2-(2+2 y x)^2\le0\\
    &\Leftrightarrow\qquad(4 y-2-2 y x)(2 x+2+2 y x)\le0\\
    &\Leftrightarrow\qquad y(2- x)\le 1.
\end{align*}
We remark that the largest exponential convergence rate is obtained by choosing $ y=1/(2- x)$ for $0\le x<2$.
This choice saturates \eqref{eq:solutions_existence}, therefore \eqref{eq:2nd_degree_s} now has a unique positive solution that fixes the choice of $\delta$, indeed 
\begin{align*}
    s=\frac{m+\alpha(1-\alpha)\lambda r}{M+m}\qquad&\Leftrightarrow\qquad1-\frac{r}{m}((1-\alpha)\lambda-2r)=1-\frac{(M-\alpha(1-\alpha)\lambda r)}{(M+m)}\\
    &\Leftrightarrow\qquad y x(M+m)-2\alpha y m)=\alpha M-\alpha y x m\\
    &\Leftrightarrow\qquad x (M+m)-(2- x)\alpha M+ x \alpha m= 2 \alpha m\\
    &\Leftrightarrow\qquad x(M+m+\alpha M+\alpha m)=2\alpha m + 2\alpha M\\
        &\Leftrightarrow\qquad x=2\alpha/(1+\alpha).
\end{align*}
We remark therefore that $ x\in[0,1)$ because $\alpha\in [0,1)$. Since the constraint \eqref{eq:2nd_degree_s} is also saturated, we know
 \eqref{eigen:eq2} holds for any $ x\in[0,1)$ because $(1-\alpha^2)\lambda-2r- b=(M-m)/b$ is non-negative. Moreover, direct computations of $s$ and $r$ in \eqref{eq:choice_ab} yields explicit values for $a,b,c$ as follows:
 \begin{equation}\label{def:abc}
     a=2(M+m)/(1+\alpha),\qquad b=\sqrt{M+m},\qquad c=2.
 \end{equation}
Both $a$ and $c$ are positive, therefore $A$ is positive definite for any $\alpha\in[0,1)$ since
 $$
 ac-b^2=(M+m)(3-\alpha)/(1+\alpha)>0.
 $$
 
\subsection{Proof of \Cref{prop:cv_generator}}\label{proof_cv_generator}
We note $\mathcal{L}^{\rm H}$ the infinitesimal generator of Hamiltonian dynamics (a.k.a Liouville operator), defined on test functions $f\in C_c^{\infty}(\R^{2d})$ by
\begin{equation*}
    \mathcal{L}^{\rm H}f(\bx,\bv)\triangleq \bv^{\top}\nabla_{\bx} f(\bx,\bv) -\nabla\Phi(\bx)^\top\nabla_{\bv} f(\bx,\bv),\qquad (\bx,\bv)\in\R^{2d}.
\end{equation*}
In addition, we note $\mathcal{L}_{\lambda,\alpha}^{\rm RH}$ the generator of Randomized HMC with refreshment rate $\lambda\ge0$ and persistence $\alpha\in[0,1)$, and we note $\mathcal{L}_\gamma^{\rm LD}$ the generator of the Langevin diffusion with damping $\gamma\ge0$. These are built upon two types of momentum refreshments. The generator $\mathcal{R}_\alpha^{\rm PP}$ refers to discrete refreshments with persistence $\alpha\in[0,1)$ driven by a standard Poisson Process. The generator $\mathcal{R}^{\rm BM}$ refers to continuous refreshments induced by a Brownian motion on $\R^d$. These generators are formally defined for $\bxi\sim\mathcal{N}_d(\bfzero_d,\bfI_d)$ by
    \begin{align*}
        \mathcal{L}_{\lambda,\alpha}^{\rm RH}\triangleq\mathcal{L}^{\rm H}+\lambda\mathcal{R}_\alpha^{\rm PP},\,\,\,\qquad\mathcal{R}_\alpha^{\rm PP}f(\bx,\bv)&\triangleq \mathbb{E}\left[f(\bx,\alpha \bv +\sqrt{1-\alpha^2}\bxi)\right]-f(\bx,\bv),\\
    \mathcal{L}_{\gamma}^{\rm LD}\triangleq\mathcal{L}^{\rm H}+\gamma\mathcal{R}^{\rm BM},\qquad\mathcal{R}^{\rm BM}f(\bx,\bv)&\triangleq -\bv^\top\nabla_{\bv} f(\bx,\bv) +{\rm tr}(\nabla^2_{\bv} f(\bx,\bv)).
\end{align*}
% Built upon the two types of refreshment introduced, explicit definitions for these generators are as follows
%\begin{align*}
%    \mathcal{L}_{\lambda,\alpha}^{\rm RH}&\triangleq\mathcal{L}^{\rm H}+\lambda\mathcal{R}_\alpha^{\rm PP}\\
 %   \mathcal{L}_{\gamma}^{\rm LD}&\triangleq\mathcal{L}^{\rm H}+\gamma\mathcal{R}^{\rm BM}.
%\end{align*}
Since $\mathcal{L}^{\rm H}$ is a common element of these two generators, their proximity depends only on the proximity of the refreshments. Indeed
\begin{equation*}
    \|\mathcal{L}_{\lambda,\alpha}^{\rm RH}f-\mathcal{L}_{\gamma}^{\rm LD}f\|_{\infty}=\|\lambda\mathcal{R}_{\alpha}^{\rm PP}f-\gamma\mathcal{R}^{\rm BM}f\|_{\infty}.
\end{equation*}
Therefore the proof does not rely on any assumption with respect to the potential function.
Let $f\in C_c^{\infty}$. By assumption there is a constant $B>0$ such that: $f$ and its derivatives are zero on the complement of the compact set $S(B)\triangleq\{(\bx,\bv)\in\R^{2d}: |\bx|\vee|\bv|\le B \}$, and for any $(\bx,\bv)\in S(B)$ we have $|\nabla_{\bv} f(\bx,\bv)|\le B$ and $\nabla_{\bv}^2 f(\bx,\bv)\preceq B\bfI_d$.  We denote $\beta=\sqrt{1-\alpha^2}$ and assume that $\beta\in (0,1/2]$. Let $\bxi\sim\mathcal{N}_d(\bfzero_d,\bfI_d)$, and define
$$
g_{\bxi}(\beta)\triangleq f\left(\bx,\sqrt{1-\beta^2}\bv+\beta \bxi\right)-f(\bx,\bv).
$$
We remark that for any $(\bx,\bv)\notin S(2B+|\bxi|)$ we have $g_{\bxi}(\beta)=0$ for any $\beta\in (0,1/2]$. Taylor’s theorem ensures that there is a function $\beta\mapsto R_{\bxi}(\beta)$ such that 
$$
g_{\bxi}(\beta)= g_{\bxi}(0)+ \beta g'_{\bxi}(0)+ (\beta^2/2)g''_{\bxi}(0)+(\beta^2/2)R_{\bxi}(\beta).
$$
We have $R_{\bxi}(\beta)\rightarrow 0$ almost surely as $\beta\rightarrow 0$ and for any $\beta\in(0,1/2]$ there exists a random variable $\delta_{\bxi}\in(0,\beta]$ such that  $ R_{\bxi}(\beta)= g''_{\bxi}(\delta_{\bxi})- g''_{\bxi}(0)$. Direct computations yields
\begin{align*}
g'_{\bxi}(\beta)=&\left(-\frac{\beta}{(1-\beta^2)^{1/2}}\bv+\bxi\right)^\top\nabla_{\bv} f\left(\bx,\sqrt{1-\beta^2}\bv+\beta \bxi\right)\\
    g''_{\bxi}(\beta)=&-\frac{1}{(1-\beta^2)^{3/2}}\bv^\top\nabla_{\bv} f\left(\bx,\sqrt{1-\beta^2}\bv+\beta \bxi\right)\\
    &+\left(-\frac{\beta}{(1-\beta^2)^{1/2}}\bv+\bxi\right)^\top\nabla_{\bv}^2 f\left(\bx,\sqrt{1-\beta^2}\bv+\beta \bxi\right)\left(-\frac{\beta}{(1-\beta^2)^{1/2}}\bv+\bxi\right).
\end{align*}
Furthermore for $\lambda=\frac{2\gamma}{1-\alpha^2}$, we have
\begin{equation*}
    \|\mathcal{L}_{\lambda,\alpha}^{\rm RH}f-\mathcal{L}_{\gamma}^{\rm LD}f\|_{\infty}=\gamma\|(2/(1-\alpha^2))\mathcal{R}_{\alpha}^{\rm PP}f-\mathcal{R}^{\rm BM}f\|_{\infty}.
\end{equation*}
For any deterministic square matrix $\bfA$ we have $\E[\bxi^\top\bfA\bxi]={\rm tr}(\bfA)$. This yields
\begin{align*}
        (2/(1-\alpha^2))\mathcal{R}_{\alpha}^{\rm PP}f(\bx,\bv)=&(2/\beta)\E[g'_{\bxi}(0)]+\E[g''_{\bxi}(0)]+\E[R_{\bxi}(\beta)]\\
        =&-\bv^\top\nabla_{\bv}f(\bx,\bv)+\E[\bxi^\top\nabla_{\bv}^2f(\bx,\bv)\bxi]+\E[R_{\bxi}(\beta)]\\
        =&\mathcal{R}^{\rm BM}f(\bx,\bv)+\E[R_{\bxi}(\beta)].
\end{align*}
We see that $ \|\mathcal{L}_{\lambda,\alpha}^{\rm RH}f-\mathcal{L}_{\gamma}^{\rm LD}f\|_{\infty}\rightarrow 0$ as soon as $\E[R_{\bxi}(\beta)]$ converges to zero uniformly over $(\bx,\bv)\in\R^{2d}$. This condition is satisfied because $R_{\bxi}(\beta)$ can be uniformly dominated by an integrable random variable, i.e.
\begin{align*}
        |R_{\bxi}(\beta)|\le& |g''_{\bxi}(\delta_{\bxi})|+|g''_{\bxi}(0)|\\
        \le& 2\left(8|\bv|B+(|\bv|+|\bxi|)^2B\right)\mathds{1}_{(\bx,\bv)\in S(2B+|\bxi|)}\\
        \le& 2\left(16B^2+8B|\bxi|+(2B+2|\bxi|)^2B\right).
\end{align*}
By Taylor's theorem $R_{\bxi}(\beta)\rightarrow 0$ almost surely as $\beta\rightarrow 0$. The claim of the Proposition follows from applying the dominated convergence theorem.

 \subsection{Proof of \Cref{prop:langevin_contraction}}\label{proof_contraction_langevin}
 For $a=\gamma^2$, $b=\gamma$ and $c=2$, consider the matrices
 $$
\bfA\triangleq\begin{pmatrix}a\bfI_d&b\bfI_d\\
    b\bfI_d& c\bfI_d\end{pmatrix},\qquad \bfA'\triangleq\begin{pmatrix}a\bfI_d&-b\bfI_d\\
    -b\bfI_d& c\bfI_d\end{pmatrix}.
$$
 Denote the forward Langevin processes $(\bZ_t)_{t\ge0}$ and $(\bZ_t)_{t\ge0}$, solutions of \eqref{eq_langevin} coupled with the same Brownian motion. We refer to the proof of \cite[Proposition 1]{dalalyan2020sampling} (for $v=0$), which obtains
$$
\frac{\dd}{\dd t}|\bZ_t-\bZ'_t|^2_{\bfA}\le - 2r|\bZ_t-\bZ'_t|^2_{\bfA}.
$$
 We remark that a similar inequality can be proven for the backward processes, with respect to the $\bfA'$-norm. Applying Gronwall's inequality then yields both
%\label{eq_wass_contraction_langevin}
\begin{align*}
    W_{2,\bfA}(\nu \bfP^t,\nu^{\prime} \bfP^t)&\le e^{-rt}W_{2,\bfA}(\nu,\nu^{\prime})\\
W_{2,\bfA'}(\nu (\bfP^t)^*,\nu^{\prime} (\bfP^t)^*)&\le e^{-rt}W_{2,\bfA'}(\nu,\nu^{\prime})
\end{align*}
The claim of the proposition follows from using the exact same arguments as in the proof of \Cref{thm:randomized_langevin_contraction} starting from \eqref{eq_adjointwass_contraction}; see \Cref{proof_contraction}. The same constant $C'=(3+2\sqrt{2})^{1/4}\le 1.56$ is obtained for any choice of friction.
\section{Additional proofs and experiments}\label{sec_add_experiments}
\begin{subsection}{Proof of \Cref{prop:geom_corr_rhmc}}\label{sec:proof_geom_corr_rhmc}
Direct computations are derived for $\sigma_i=1$. Denote $Y_n=\bY^{n}(i)$ and $\xi_n=\bxi^n(i)$.
\begin{align*}
    {\rm Corr}(Y_n,Y_0)=\,&\E[\cos(\tau_n)Y_{n-1}Y_0]=\E[\cos(\tau_1)]{\rm Corr}(Y_{n-1},Y_0)=\E[\cos(\tau_1)]^n\\
    a_n\triangleq{\rm Corr}(Y_n^2,Y_0^2)=\,&{\rm Corr}\left(\left(\cos^2(\tau_n)Y_{n-1}^2+\sin^2(\tau_n)\xi_n^2+2\cos(\tau_n)\sin(\tau_n)Y_{n-1}\xi_n\right),Y_0^2\right)\\
    =\,&(1/2)\left(\E[\cos^2(\tau_1)]\E[Y_n^2Y_0^2]+\E[\sin^2(\tau_1)]\E[\xi_n^2Y_0^2]-1\right)\\
    =\,&(1/2)\left(\E[\cos^2(\tau_1)](2a_{n-1}+1)+\E[\sin^2(\tau_1)]-1\right)\\
    =\,&a_{n-1}\E[\cos^2(\tau_1)]=\E[\cos^2(\tau_1)]^n
\end{align*}
Computing the moments of the exponential distribution yields $\E[\cos(\tau_1)]=1/(1+\lambda^{-2})$ while
$$
\E[\cos^2(\tau_1)]=(1/2)\E[\cos(2\tau_1)+1]=\frac{1}{2}\left(\frac{1}{1+4\lambda^{-2}}+1\right)=\frac{1+2\lambda^{-2}}{1+4\lambda^{-2}}
$$
In particular $(r_i(T))^2\neq s_i(T)$, therefore $(\bY^{n}(i))_{n\ge 0}$ cannot be jointly Gaussian, otherwise Isserlis' theorem would yield a contradiction.
\end{subsection}

\begin{subsection}{Gaussian mixture}
For $\ba\in\R^d$ and a positive definite matrix $\bfSigma\in\R^{d\times d}$, we consider the mixture of $\mathcal{N}_d(\ba,\bfSigma)$ and $\mathcal{N}_d(-\ba,\bfSigma)$ with equal weights.
Noting $\bb=\bfSigma^{-1}\ba$, we define $\Phi$ and its derivatives at $\bx\in\R^d$ as
\begin{align*}
    \Phi(\bx)&=\frac{1}{2}|\bx-\ba|_{\bfSigma^{-1}}^2-\log\left(1+\exp(-2\bx^\top\bb)\right)\\
    \nabla\Phi(\bx)&=\bfSigma^{-1}\bx-\bb+2\bb\left(1+\exp(-2\bx^\top\bb)\right)^{-1}\\
    \nabla^2\Phi(\bx)&=\bfSigma^{-1}-4\bb\bb^\top\exp(-2\bx^\top\bb)\left(1+\exp(-2\bx^\top\bb)\right)^{-2}
\end{align*}
The target is strongly log-concave if $|\ba|_{\bfSigma^{-1}}<1$. In that case, the bound $0\le e^u/(1+e^u)^2\le1/4$ shows that \Cref{assumption:strong_cvx} holds with constants $m=(1-|\ba|_{\bfSigma^{-1}}^2)/\sigma_{\rm max}(\bfSigma)$ and $M=1/\sigma_{\rm min}(\bfSigma)$, where $\sigma_{\rm min}(\bfSigma)$ and $\sigma_{\rm max}(\bfSigma)$ denote the smallest and largest eigen-values of $\bfSigma$. The target has $d=50$ components such that $\bfSigma={\rm diag}_{1\le i\le d}(\sigma_i^2)$ where $\sigma_i^2=i/d$. We set $\ba(i)=\sqrt{i}/(2d)$ so that $|\ba|_{\bfSigma^{-1}}=1/2$. In \Cref{ess_mixture_graph} and \Cref{tab:odd_vs_even_mixture}, we compare the minimum ESS per gradient evaluation for MALT, RHMC and HMC for a time step $h=0.20$ on samples of size $N=10^6$. The friction in MALT is chosen as $\gamma=1/\sigma_{\rm max}$, while the number of steps in \Cref{tab:odd_vs_even_mixture} is chosen in order to optimize the worst efficiency between $f(x)=x$ and $f(x)=x^2$. The optimal number of steps for HMC is $L=1$ (i.e. MALA). We compute the efficiency of HMC for $L=3$ as well.
\vspace{-0.2cm}
\begin{figure}[!ht]
  \centering
\includegraphics[width=.49\linewidth]{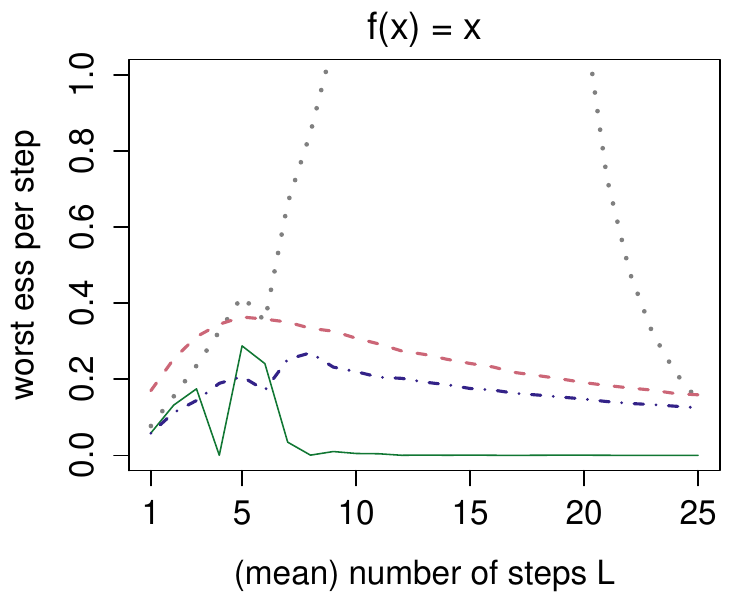}
\includegraphics[width=.49\linewidth]{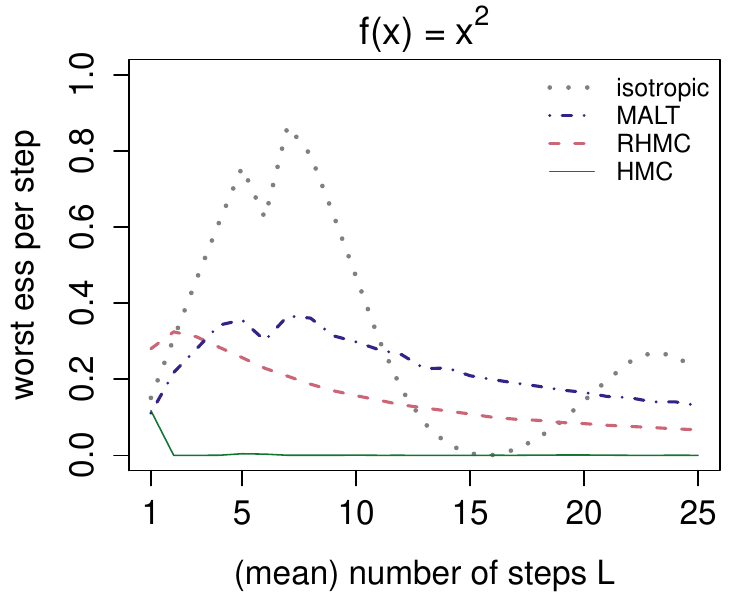}
\vspace{-0.4cm}
\caption{Gaussian mixture. \it Minimum ESS per gradient evaluation for estimating the mean and variance (resp. left and right). The dotted grey lines correspond to an ideally preconditioned HMC sampler (isotropic: $\bfSigma=\bfI_d$). The dot-dashed blue lines correspond to MALT for $\gamma=1/\sigma_{\rm max}$. The dashed pink lines correspond to Randomized HMC with full refreshments. The solid green lines correspond to HMC.
\vspace{-0.2cm}}\label{ess_mixture_graph}
\end{figure}
\vspace{-0.3cm}
\begin{table}[!ht]
    % \begin{center}
          \caption{\it
    Gaussian mixture. Minimum ESS per gradient evaluation for various odd/even functions.
    }\label{tab:odd_vs_even_mixture}
 \begin{tabular}{ %|m{2.5cm}||m{1cm}|m{1cm}|m{1cm}|m{1cm}|m{1cm}|m{1cm}|  }
%|p{2.5cm}||p{1cm}|p{1cm}|p{1cm}|p{1cm}|p{1cm}|p{1cm}|  }
%|c||c|c|c|c|c|c|c|c|}
|C{2.4cm}|C{0.95cm}|C{0.95cm}|C{0.95cm}|C{0.95cm}|C{0.95cm}|C{0.95cm}|C{0.95cm}|C{0.95cm}|}
 \hline
 &\multicolumn{4}{c|}{odd}&\multicolumn{4}{c|}{even} \\
 \hline
 $f(x)$& $x$& $x^3$ &${\rm sgn}(x)$&$\sin(x)$ & $x^2$ & $x^4$    &$e^{-|x|}$&   $\cos(x)$\\
 \hline
  MALT: $L=8$ & 0.27& 0.32& 0.31  & 0.27   &\textbf{0.36}& \textbf{0.37}    &\textbf{0.38}& \textbf{0.36}\\
 RHMC: $L=4$   &\textbf{0.35} & \textbf{0.38}& \textbf{0.41}    &\textbf{0.36}&   0.29& 0.29    &0.33&   0.29\\
 HMC: $L=3$ &  0.17& 0.23& 0.24&0.19  &0.00 & 0.00&0.00 &0.00\\
  MALA ($L=1$) &  0.06& 0.08& 0.09&0.07  &0.11 & 0.13&0.16 &0.12\\
 \hline
\end{tabular}  
\end{table}

\Cref{ess_mixture_graph} and \Cref{tab:odd_vs_even_mixture} show that MALT and RHMC both achieve higher efficiency than HMC and MALA, on a unimodal Gaussian mixture. This observation confirms the better performance of MALT and RHMC on the Gaussian mixture model. A similar discrepancy between odd and even functions is observed when comparing MALT and RHMC.
\end{subsection} 

\begin{subsection}{Student's distribution}
For a pos. def. matrix $\bfSigma\in\R^{d\times d}$, we consider the Student's distribution with $k\ge 1$ degrees of freedom. The potential and its gradient at $\bx\in\R^d$ are
% illustrate numerically the efficiency of MALT
%-\frac{(k+d)}{2}\left(1+\frac{\bx^\top\bfSigma^{-1}\bx}{k}\right)
\begin{align*}
    \Phi(\bx)&=\frac{(k+d)}{2}\log\left(k+|\bx|_{\bfSigma^{-1}}^2\right)\\
    \nabla\Phi(\bx)&=\bfSigma^{-1}\bx(k+d)\left(k+|\bx|_{\bfSigma^{-1}}^2\right)^{-1}
\end{align*}
The target has $k-1$ moments. Here, \Cref{assumption:strong_cvx} does not hold while \Cref{assumption:grad_lipschitz} does. For any $k>2$, the covariance matrix is proportional to $\bfSigma$. Heterogeneity of scales is introduced among $d=50$ components by setting $\bfSigma={\rm diag}_{1\le i\le d}(\sigma_i^2)$ where $\sigma_i^2=i/d$. In \Cref{ess_student_graph} and \Cref{tab:odd_vs_even_student}, we compare numerically the minimum ESS per gradient evaluation for MALT, RHMC and HMC on a Student with $k=20$ degrees of freedom for a time step $h=0.20$. The friction in MALT is chosen empirically as $\gamma=1/\sigma_{\rm max}$, while the number of steps in \Cref{tab:odd_vs_even_student} is chosen in order to optimize the worst efficiency between $f(x)=x$ and $f(x)=x^2$. HMC's optimal number of steps is $L=3$. 
\vspace{-0.2cm}
\begin{figure}[!ht]
  \centering
\includegraphics[width=.49\linewidth]{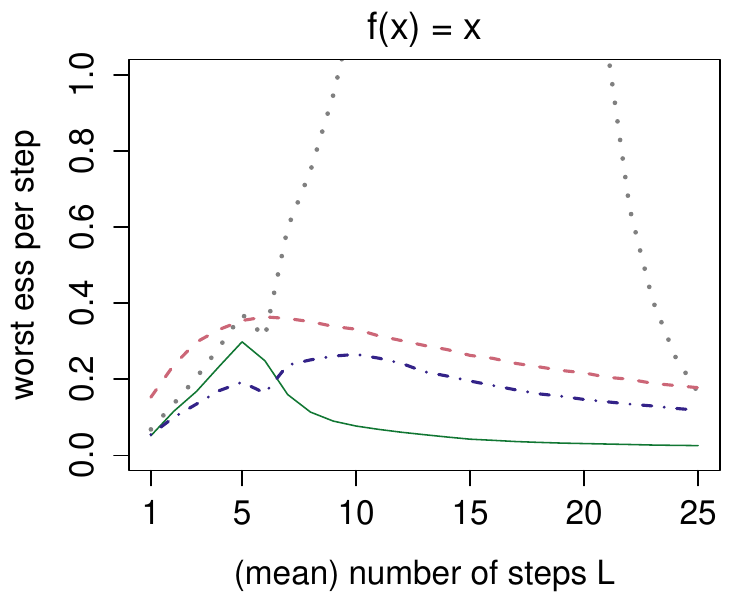}
\includegraphics[width=.49\linewidth]{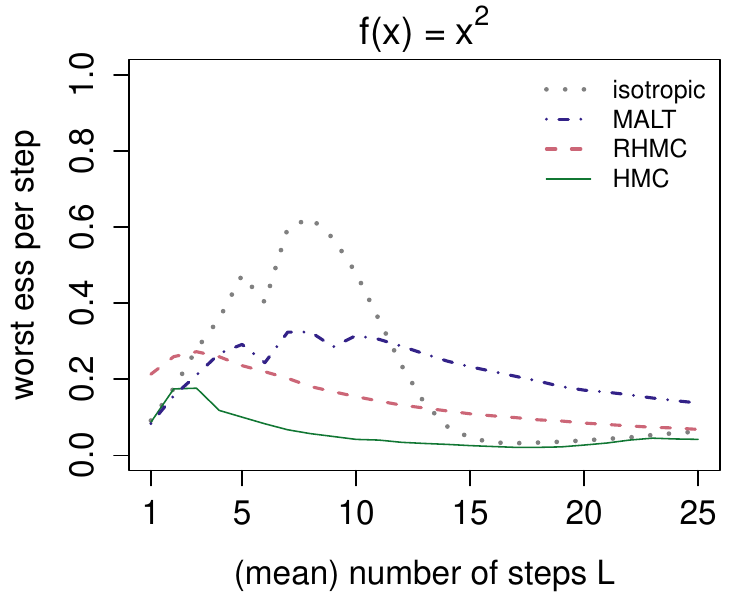}
\caption{Student's distribution. \it Minimum ESS per gradient evaluation for estimating the mean and variance (resp. left and right). The dotted grey lines correspond to an ideally preconditioned HMC sampler (isotropic: $\bfSigma=\bfI_d$). The dot-dashed blue lines correspond to MALT for $\gamma=1/\sigma_{\rm max}$. The dashed pink lines correspond to Randomized HMC with full refreshments. The solid green lines correspond to HMC.
%\vspace{-0.3cm}
}\label{ess_student_graph}
\end{figure}
\vspace{-0.3cm}
\begin{table}[!ht]
    % \begin{center}
          \caption{\it
    Student's distribution. Minimum ESS per gradient evaluation for various odd/even functions.
    }\label{tab:odd_vs_even_student}
 \begin{tabular}{ %|m{2.5cm}||m{1cm}|m{1cm}|m{1cm}|m{1cm}|m{1cm}|m{1cm}|  }
%|p{2.5cm}||p{1cm}|p{1cm}|p{1cm}|p{1cm}|p{1cm}|p{1cm}|  }
%|c||c|c|c|c|c|c|c|c|}
|C{2.4cm}|C{0.95cm}|C{0.95cm}|C{0.95cm}|C{0.95cm}|C{0.95cm}|C{0.95cm}|C{0.95cm}|C{0.95cm}|}
 \hline
 &\multicolumn{4}{c|}{odd}&\multicolumn{4}{c|}{even} \\
 \hline
 $f(x)$& $x$& $x^3$ &${\rm sgn}(x)$&$\sin(x)$ & $x^2$ & $x^4$    &$e^{-|x|}$&   $\cos(x)$\\
 \hline
  MALT: $L=8$ & 0.25& 0.30& 0.29  & 0.28   &\textbf{0.33}& \textbf{0.37}    &0.26& \textbf{0.33}\\
 RHMC: $L=5$   &\textbf{0.35} & \textbf{0.37}& \textbf{0.40}    &\textbf{0.37}&   0.24& 0.26    &\textbf{0.28}&   0.25\\
 HMC: $L=3$ &  0.17& 0.19& 0.24&0.20  &0.18 & 0.19&0.17 &0.18\\
 MALA ($L=1$) &0.05&0.07 &0.08 & 0.07& 0.09& 0.08& 0.14& 0.11\\
 \hline
\end{tabular}  
\end{table}
%\vspace{-0.5cm}

\Cref{ess_student_graph} and \Cref{tab:odd_vs_even_student} give an illustration of the sampling performances on a distribution with polynomial tails. We observe that heterogeneity of scales in the covariance matrix have slightly less impact on the worst ESS of HMC, although its performance still breaks down in higher dimension. When comparing MALT and RHMC: apart from one example, a similar discrepancy between odd and even functions is observed. 
\end{subsection}
%Preprint
%\bibliographystyle{abbrvnat}
\bibliographystyle{imsart-number} % Style BST file (imsart-number.bst or imsart-nameyear.bst)
%\bibliography{main}
{\renewcommand{\addtocontents}[2]{}
\bibliography{main}}

\end{document}